\documentclass[envcountsame,envcountsect]{llncs}

\usepackage{fancyhdr}
\usepackage{amsmath}
\usepackage{amsfonts}
\usepackage{amssymb}
\usepackage{mathrsfs}
\usepackage{textcomp}
\usepackage{stmaryrd}
\usepackage{latexsym} 
\usepackage{mathpartir}
\usepackage{longtable}
\usepackage[counterclockwise]{rotating}
\usepackage{relsize}
\usepackage{microtype}
\usepackage[utf8]{inputenc}
\usepackage[OT4]{fontenc}
\usepackage[all]{xy}

\newcommand{\authcount}[1]{}

\newcommand{\reduces}{\ensuremath{\stackrel{*}{\rightarrow}}}
\newcommand{\contr}{\ensuremath{\rightarrow}}

\newcommand{\ipreduces}{\ensuremath{\stackrel{*}{\Rightarrow}}}
\newcommand{\ipcontr}{\ensuremath{\Rightarrow}}
\newcommand{\ipeqvred}{\ensuremath{\stackrel{*}{\Leftrightarrow}}}

\newcommand{\trsrule}[2]{\ensuremath{#1\rightarrow#2}}
\newcommand{\valuation}[3]{\ensuremath{\llbracket#1\rrbracket_{#2}^{#3}}}
\newcommand{\set}[2]{\ensuremath{\{ #1 \;|\; #2 \}}}
\newcommand{\pos}[2]{\ensuremath{{#1}_{|#2}}}
\newcommand{\upos}[2]{\ensuremath{{#1}^{|#2}}}
\newcommand{\cut}[2]{\ensuremath{{#1}_{\upharpoonright#2}}}
\newcommand{\proves}{\ensuremath{\vdash}}

\newcommand{\pair}[2]{\langle#1,#2\rangle}

\newcommand{\transl}[1]{\ensuremath{\lceil#1\rceil}}

\newcommand{\Nbb}{\ensuremath{\mathbb{N}}}

\newcommand{\Pos}{\mbox{\cal Pos}}

\newcommand{\rank}{\ensuremath{\mathrm{rank}}}

\newcommand{\Type}{\ensuremath{\mathrm{Type}}}
\newcommand{\Prop}{\ensuremath{\mathrm{Prop}}}
\newcommand{\Bool}{\ensuremath{\mathrm{Bool}}}

\newcommand{\Nat}{\ensuremath{\mathrm{Nat}}}

\newcommand{\Cond}[3]{\ensuremath{\mathrm{Cond}\,#1\,#2\,#3\,}}
\newcommand{\Eq}[2]{\ensuremath{\mathrm{Eq}\,#1\,#2\,}}
\newcommand{\Eql}[3]{\ensuremath{\mathrm{Eql}\,#1\,#2\,#3\,}}

\newcommand{\T}{\ensuremath{{\cal T}}}
\newcommand{\A}{\ensuremath{{\cal A}}}
\newcommand{\B}{\ensuremath{{\cal B}}}
\newcommand{\I}{\ensuremath{{\cal I}}}

\newcommand{\Tc}{\ensuremath{{\cal T}}}
\newcommand{\Ac}{\ensuremath{{\cal A}}}
\newcommand{\Bc}{\ensuremath{{\cal B}}}

\newcommand{\Sc}{\ensuremath{{\cal S}}}
\newcommand{\Fc}{\ensuremath{{\cal F}}}
\newcommand{\Cc}{\ensuremath{{\cal C}}}
\newcommand{\Dc}{\ensuremath{{\cal D}}}
\newcommand{\Oc}{\ensuremath{{\cal O}}}
\newcommand{\Mc}{\ensuremath{{\cal M}}}

\newcommand{\Kc}{\ensuremath{{\cal K}}}

\newcommand{\true}{\ensuremath{\mathbf{t}}}
\newcommand{\false}{\ensuremath{\mathbf{f}}}

\makeatletter
\renewcommand{\tableofcontents}{
  \section*{Table of Contents}
  \@starttoc{toc}
}
\makeatother

\makeatletter
\newcommand{\rmnum}[1]{\romannumeral #1}
\newcommand{\Rmnum}[1]{\expandafter\@slowromancap\romannumeral #1@}
\makeatother

\allowdisplaybreaks

\begin{document}

\title{Partiality and Recursion in Higher-Order Logic}

\author{\L{}ukasz Czajka}
\institute{
  Institute of Informatics, University of Warsaw\\
  Banacha 2, 02-097 Warszawa, Poland\\
  \email{lukaszcz@mimuw.edu.pl}
}

\maketitle

\begin{abstract}
  We present an illative system~$\I_s$ of classical higher-order logic
  with subtyping and basic inductive types. The system~$\I_s$ allows
  for direct definitions of partial and general recursive functions,
  and provides means for handling functions whose termination has not
  been proven. We give examples of how properties of some recursive
  functions may be established in our system. In a technical appendix
  to the paper we prove consistency of~$\I_s$. The proof is by model
  construction. We then use this construction to show conservativity
  of~$\I_s$ over classical first-order logic. Conservativity over
  higher-order logic is conjectured, but not proven.
\end{abstract}

\begin{paragraph}
  {\bf Note:} This paper is an extended technical report based on a
  conference paper with the same title published by Springer-Verlag in
  the proceedings of FoSSaCS~2013. The final publication is available
  at \mbox{springerlink.com}.
\end{paragraph}

\tableofcontents

\renewcommand{\T}{\ensuremath{\mathbb T}}
\renewcommand{\reduces}{\ensuremath{\twoheadrightarrow}}
\newcommand{\Is}[2]{\ensuremath{\mathrm{Is}\,#1\,#2\,}}
\newcommand{\uipreduces}[1]{\ensuremath{\stackrel{\;\;*\;#1}{\Rightarrow\;\;}}}
\newcommand{\udipreduces}[2]{\ensuremath{\stackrel{*\;#1}{\Rightarrow_{#2}\;\,}}}
\newcommand{\of}{\ensuremath{\mathfrak{o}}}
\newcommand{\cf}{\ensuremath{\mathfrak{c}}}
\newcommand{\df}{\ensuremath{\mathfrak{d}}}

\newcommand{\zero}{\ensuremath{0}}
\newcommand{\iszero}{\ensuremath{\mathfrak{0}}}
\newcommand{\s}{\ensuremath{\mathfrak{s}}}
\newcommand{\p}{\ensuremath{\mathfrak{p}}}

\newcommand{\Op}{\ensuremath{\mathrm{Op}}}
\newcommand{\All}{\ensuremath{\mathrm{A}}}
\newcommand{\Set}{\ensuremath{\mathrm{S}}}

\section{Introduction}

We present an illative $\lambda$-calculus system~$\I_s$ of classical
higher-order logic with subtyping and basic inductive types. Being
illative means that the system is a combination of higher-order logic
with the \emph{untyped} $\lambda$-calculus. It therefore allows for
unrestricted recursive definitions directly, including definitions of
possibly non-terminating partial functions. We believe that this
feature of~$\I_s$ makes it potentially interesting as a logic for an
interactive theorem prover intended to be used for program
verification.


Most popular proof assistants allow only total functions, and totality
must be ensured by the user, either by very precise specifications of
function domains, restricting recursion in a way that guarantees
termination, explicit well-foundedness proofs, or other means.

Obviously, there is a reason why most proof assistants do not handle
partial functions directly. This is to ensure consistency of the
system. Combining an expressive higher-order logic with unrestricted
recursion is a non-trivial problem.

There are various indirect ways of dealing with general recursion in
popular theorem provers based on total logics. There are also many
non-standard logics allowing partial functions directly. We discuss
some related work in Sect.~\ref{sec_related}.

In Sect.~\ref{sec_I_s} we introduce the system~$\I_s$. Our approach
builds on the old tradition of illative combinatory logic
\cite{illat01,Seldin2009,Czajka2013Accepted,Czajka2011}. This
tradition dates back to early inconsistent systems of Sh\"onfinkel,
Church and Curry proposed in the 1920s and the
1930s~\cite{Seldin2009}. However, after the discovery of paradoxes
most logicians abandoned this approach. A notable exception was
Haskell Curry and his school, but not much progress was made in
establishing consistency of illative systems strong enough to
interpret traditional logic. Only in the 1990s some first-order
illative systems were shown consistent and complete for traditional
first-order logic~\cite{illat01,illat02}. The system~$\I_s$, in terms
of the features it provides, may be considered an extension of the
illative system~$\I_\omega$ from~\cite{Czajka2013Accepted}. We briefly
discuss the relationship between~$\I_s$ and~$\I_\omega$ in
Sect.~\ref{sec_related}.

Because~$\I_s$ is based on the untyped $\lambda$-calculus, its
consistency is obviously open to doubt. In an appendix we give a proof
by model construction of consistency of~$\I_s$. Unfortunately, the
proof is too long to fit within the page limits of a conference
paper. In Sect.~\ref{sec_consistent} we give a general overview of the
proof. The model construction is similar to the one
from~\cite{Czajka2013Accepted} for the traditional illative
system~$\I_\omega$. It is extended and adapted in a non-trivial way to
account for additional features of~$\I_s$. To our knowlege~$\I_s$ is
the first higher-order illative system featuring subtypes and some
form of induction, for which there is a consistency proof.

In Sect.~\ref{sec_partiality} we provide examples of proofs in~$\I_s$
indicating possible applications of our approach to the problem of
dealing with partiality, non-termination and general recursion in
higher-order logic. We are mainly interested in partiality arising
from non-termination of non-well-founded recursive definitions.

For lack of space we omit proofs of most of the lemmas and theorems we
state. The proofs of non-trivial results may be found in technical
appendices to this paper.

\section{The Illative System}\label{sec_I_s}

In this section we present the system~$\I_s$ of illative classical
higher-order logic with subtyping and derive some of its basic
properties.

\begin{definition}\label{def_I_s} \rm

  The system $\I_s$ consists of the following.

  \begin{itemize}
  \item A countably infinite set of variables $V_s = \{x, y, z, \ldots
    \}$ and a set of constants~$\Sigma_{s}$.
  \item The set of sorts $\Sc = \{\Type, \Prop\}$.
  \item The set of \emph{basic inductive types}~$\Tc_I$ is defined
    inductively by the rule: if
    $\iota_{1,1},\ldots,\iota_{1,n_1},\ldots,\iota_{m,1},\ldots,\iota_{m,n_m}
    \in \Tc_I \cup \{\star\}$ then
    \[
    \mu(\langle \iota_{1,1},\ldots,\iota_{1,n_1} \rangle, \ldots,
    \langle \iota_{m,1},\ldots,\iota_{m,n_m} \rangle) \in \Tc_I
    \]
    where $m \in \Nbb_+$ and $n_1,\ldots,n_m \in \Nbb$.
  \item We define the sets of \emph{constructors}~$\Cc$,
    \emph{destructors}~$\Dc$, and \emph{tests}~$\Oc$ as follows. For
    each $\iota \in \Tc_I$ of the form
    \[
      \iota = \mu(\langle \iota_{1,1},\ldots,\iota_{1,n_1} \rangle,
      \ldots, \langle \iota_{m,1},\ldots,\iota_{m,n_m} \rangle) \in
      \Tc_I
    \]
    where $\iota_{i,j} \in \Tc_I \cup \{\star\}$, the set~$\Cc$
    contains~$m$ distinct constants~$c_1^\iota,\ldots,c_m^\iota$. The
    number~$n_i$ is called the \emph{arity} of~$c_i^\iota$, and
    $\langle \iota_{i,1},\ldots,\iota_{i,n_i} \rangle$ is its
    \emph{signature}. With each $c_i^\iota \in \Cc$ of arity~$n_i$ we
    associate~$n_i$ distinct
    destructors~$d_{i,1}^\iota,\ldots,d_{i,n_i}^\iota \in \Dc$ and one
    test $o_i^\iota \in \Oc$. When we use the symbols $c_i^\iota$,
    $o_i^\iota$ and $d_{i,j}^\iota$ we implicitly assume that they
    denote the constructors, tests and destructors associated
    with~$\iota$. When it is clear from the context which type~$\iota$
    is meant, we use the notation~$\iota_{i,j}^*$ for~$\iota_{i,j}$ if
    $\iota_{i,j} \ne \star$, or for~$\iota$ if $\iota_{i,j} = \star$.
  \item The set of \emph{$\I_s$-terms}~$\T$ is defined by the
    following grammar.
    \begin{eqnarray*}
      \T &::=& V_s \;|\; \Sigma_{s} \;|\; \Sc \;|\; \Cc \;|\; \Dc
      \;|\; \Oc \;|\; \Tc_I \;|\; \lambda V_s \,.\, \T \;|\; (\T \T)
      \;|\; \mathrm{Is} \;|\; \mathrm{Subtype} \;|\; \mathrm{Fun}
      \;|\; \\ && \forall \;|\; \vee \;|\; \bot \;|\; \epsilon \;|\;
      \mathrm{Eq} \;|\; \mathrm{Cond}
    \end{eqnarray*}
    We assume application associates to the left and omit spurious
    brackets.
  \item We identify $\alpha$-equivalent terms, i.e. terms differing
    only in the names of bound variables are considered identical. We
    use the symbol~$\equiv$ for identity of terms up to
    $\alpha$-equivalence. We also assume without loss of generality
    that all bound variables in a term are distinct from the free
    variables, unless indicated otherwise.\footnote{So e.g. in the
      axiom~$\beta$ the free variables of~$t_2$ do not become bound in
      $t_1[x/t_2]$.}
  \item In what follows we use the abbreviations:
    \begin{eqnarray*}
      t_1 : t_2 &\equiv& \mathrm{Is}\, t_1\, t_2 \\
      \set{x : \alpha}{\varphi} &\equiv& \mathrm{Subtype}\, \alpha\,
      \lambda x \,.\, \varphi \\
      \alpha \to \beta &\equiv& \mathrm{Fun}\, \alpha\, \beta \\
      \forall x : \alpha \,.\, \varphi &\equiv& \forall\, \alpha\, \lambda x
      \,.\, \varphi \\
      \forall x_1,\ldots,x_n : \alpha \,.\, \varphi &\equiv& \forall x_1 :
      \alpha \,.\, \ldots \forall x_n : \alpha \,.\, \varphi \\
      \varphi \supset \psi &\equiv& \forall x : \set{y : \Prop}{\varphi}
      \,.\, \psi \;\; \mathrm{where\ } x, y \notin FV(\varphi, \psi)
      \\
      \neg \varphi &\equiv& \varphi \supset \bot \\
      \top &\equiv& \bot \supset \bot \\
      \varphi \vee \psi &\equiv& \vee \varphi \psi \\
      \varphi \wedge \psi &\equiv& \neg (\neg \varphi \vee \neg \psi) \\
      \exists x : \alpha \,.\, \varphi &\equiv& \neg \forall x : \alpha
      \,.\, \neg \varphi
    \end{eqnarray*}
    We assume that~$\neg$ has the highest precedence.
  \item The system $\I_s$ is given by the following rules and axioms,
    where $\Gamma$ is a finite set of terms, $t, \varphi, \psi,
    \alpha, \beta$, etc. are arbitrary terms. The notation $\Gamma,
    \varphi$ is a shorthand for $\Gamma \cup \{\varphi\}$. We use
    Greek letters $\varphi$, $\psi$, etc. to highlight that a term is
    to be intuitively interpreted as a proposition, and we use
    $\alpha$, $\beta$, etc. when it is to be interpreted as a type,
    but there is no a priori syntactic distinction. All judgements
    have the form $\Gamma \proves t$ where~$\Gamma$ is a set of terms
    and~$t$ a term. In particular, $\Gamma \proves t : \alpha$ is a
    shorthand for $\Gamma \proves \mathrm{Is}\, t\, \alpha$,
    according to the convention from the previous point.

    \medskip

    {\bf Axioms}

    \begin{itemize}
    \item[1:] $\Gamma, \varphi \proves \varphi$
    \item[2:] $\Gamma \proves \Eq{t}{t}$
    \item[3:] $\Gamma \proves \Prop : \Type$
    \item[4:] $\Gamma \proves \iota : \Type$ for $\iota \in \Tc_I$
    \item[5:] $\Gamma \proves o_i^\iota (c_i^\iota t_1 \ldots
      t_{n_i})$ if $c_i^\iota \in \Cc$ has arity~$n_i$
    \item[6:] $\Gamma \proves \neg (o_i^\iota (c_j^\iota t_1 \ldots
      t_{n_j}))$ if $i \ne j$ and $c_j^\iota \in \Cc$ has arity~$n_j$
    \item[7:] $\Gamma \proves \Eq{(d_{i,k}^\iota (c_i^\iota t_1 \ldots
      t_{n_i}))}{t_k}$ for $k=1,\ldots,n_i$, if $c_i^\iota \in \Cc$
      has arity~$n_i$
    \item[$\bot_t$:] $\Gamma \proves \bot : \Prop$
    \item[$c$:] $\Gamma \proves \forall p : \Prop \,.\, p \lor \neg p$
    \item[$\beta$:] $\Gamma \proves \Eq{((\lambda x \,.\, t_1)
      t_2)}{(t_1[x/t_2])}$
    \end{itemize}

    {\bf Rules}

    \begin{longtable}{cc}
      \multicolumn{2}{c}{
        \(
          {\forall_i:}\; \inferrule{\Gamma \proves
            \alpha : \Type \\ \Gamma, x : \alpha \proves \varphi \\ x
            \notin FV(\Gamma, \alpha)}{\Gamma \proves \forall x : \alpha \,.\, \varphi}
        \)
      }
      \\
      & \\
      \multicolumn{2}{c}{
        \(
          {\forall_e:}\; \inferrule{\Gamma \proves \forall x : \alpha
            \,.\, \varphi \\ \Gamma \proves t : \alpha}
          {\Gamma \proves \varphi[x/t]}
         \)
      }
      \\
      & \\
      \multicolumn{2}{c}{
      \(
      {\forall_t:}\; \inferrule{\Gamma \proves
        \alpha : \Type \\ \Gamma, x : \alpha \proves \varphi : \Prop \\ x
        \notin FV(\Gamma, \alpha)}{\Gamma \proves (\forall x : \alpha
        \,.\, \varphi) : \Prop}
      \)
      }
      \\
      & \\
      \multicolumn{2}{c}{
      \(
      {\exists_i:}\; \inferrule{\Gamma \proves
        \alpha : \Type \\ \Gamma \proves t : \alpha \\ \Gamma \proves \varphi[x/t]
      }
      {\Gamma \proves \exists x : \alpha \,.\, \varphi}
      \)
      }
      \\
      & \\
      \multicolumn{2}{c}{
      \(
      {\exists_e:}\; \inferrule{\Gamma \proves \exists x : \alpha
        \,.\, \varphi \\ \Gamma, x : \alpha, \varphi \proves \psi \\ x
        \notin FV(\Gamma, \psi, \alpha)}
      {\Gamma \proves \psi}
      \)
      }
      \\
      & \\
      \(
      {\vee_{i1}:}\; \inferrule{\Gamma \proves \varphi}
      {\Gamma \proves \varphi \vee \psi}
      \)
      &
      \(
      {\vee_{i2}:}\; \inferrule{\Gamma \proves \psi}
      {\Gamma \proves \varphi \vee \psi}
      \)
      \\
      & \\
      \multicolumn{2}{c}{
      \(
      {\vee_{e}:}\; \inferrule{
        \Gamma \proves \varphi_1 \vee \varphi_2 \\ \Gamma, \varphi_1
        \proves \psi \\ \Gamma, \varphi_2 \proves \psi
      }{
        \Gamma \proves \psi
      }
      \)
      }
      \\
      & \\
      \multicolumn{2}{c}{
      \(
      {\vee_{t}:}\; \inferrule{
        \Gamma \proves \varphi : \Prop \\ \Gamma \proves \psi : \Prop
      }{
        \Gamma \proves (\varphi \vee \psi) : \Prop
      }
      \)
      }
      \\
      & \\
      \(
        {\wedge_{e1}:}\;
        \inferrule{
          \Gamma \proves \varphi \wedge \psi
        }{
          \Gamma \proves \varphi
        }
      \)
      &
      \(
        {\wedge_{e2}:}\;
        \inferrule{
          \Gamma \proves \varphi \wedge \psi
        }{
          \Gamma \proves \psi
        }
      \)
      \\
      & \\
      \(
        {\supset_{t2}:}\;
        \inferrule{
          \Gamma \proves (\varphi \supset \psi) : \Prop
        }{
          \Gamma \proves \varphi : \Prop
        }
      \)
      &
      \(
        {\bot_{e}:}\;
        \inferrule{
          \Gamma \proves \bot
        }{
          \Gamma \proves \varphi
        }
      \)
      \\
      & \\
      \multicolumn{2}{c}{
      \(
      {\to_i:}\; \inferrule{\Gamma \proves
        \alpha : \Type \\ \Gamma, x : \alpha \proves t : \beta \\ x
        \notin FV(\Gamma, \alpha, \beta)}{\Gamma \proves (\lambda x \,.\, t) :
        \alpha \to \beta}
      \)
      }
      \\
      & \\
      \(
      {\to_e:}\; \inferrule{\Gamma \proves t_1 : \alpha \to \beta \\ \Gamma \proves t_2 : \alpha}{\Gamma \proves t_1
        t_2 : \beta}
      \)
      &
      \(
        {\to_t:}\; \inferrule{\Gamma \proves \alpha : \Type \\ \Gamma
          \proves \beta : \Type}
        {\Gamma \proves (\alpha \to \beta) : \Type}
      \)
      \\
      & \\
      \multicolumn{2}{c}{
      \(
        {s_i:}\; \inferrule{\Gamma \proves \{ x : \alpha \;|\; \varphi \} : \Type \\ \Gamma \proves t : \alpha \\ \Gamma \proves
          (\lambda x \,.\, \varphi) t \\ x \notin FV(\alpha)}
        {\Gamma \proves t : \{ x : \alpha \;|\; \varphi \} }
      \)
      }
      \\
      & \\
      \(
        {s_e:}\; \inferrule{\Gamma \proves t : \{ x :
          \alpha \;|\; \varphi \}}{\Gamma \proves \varphi[x/t]}
      \)
      &
      \(
        {s_{et}:}\; \inferrule{\Gamma \proves t : \{ x :
          \alpha \;|\; \varphi\}}{\Gamma \proves t : \alpha}
      \)
      \\
      & \\
      \multicolumn{2}{c}{
        \(
        {s_t:}\; \inferrule{\Gamma \proves \alpha : \Type
          \\ \Gamma, x : \alpha \proves \varphi : \Prop \\ x \notin FV(\alpha)}
        {\Gamma \proves \{ x : \alpha \;|\; \varphi\} : \Type}
        \)
      }
      \\
      & \\
      \(
        {\epsilon_i:}\; \inferrule{\Gamma \proves \exists x : \alpha \,.\, \top}
        {\Gamma \proves (\epsilon \alpha) : \alpha}
      \)
      &
      \(
        {p_i:}\; \inferrule{\Gamma \proves \varphi}
        {\Gamma \proves \varphi : \Prop}
      \)
      \\
      & \\
      \(
      {c_1:}\; \inferrule{
        \Gamma \proves \varphi
      }{
        \Gamma \proves \Eq{(\Cond{\varphi}{t_1}{t_2})}{t_1}
      }
      \)
      &
      \(
      {c_2:}\; \inferrule{
        \Gamma \proves \neg\varphi
      }{
        \Gamma \proves \Eq{(\Cond{\varphi}{t_1}{t_2})}{t_2}
      }
      \)
      \\
      & \\
      \multicolumn{2}{c}{
      \(
      {c_3:}\; \inferrule{
        \Gamma, \varphi \proves \Eq{t_1}{t_1'} \\ \Gamma \proves
        \varphi : \Prop
      }{
        \Gamma \proves \Eq{(\Cond{\varphi}{t_1}{t_2})}{(\Cond{\varphi}{t_1'}{t_2})}
      }
      \)
      }
      \\
      & \\
      \multicolumn{2}{c}{
      \(
      {c_4:}\; \inferrule{
        \Gamma, \neg\varphi \proves \Eq{t_2}{t_2'} \\ \Gamma \proves
        \varphi : \Prop
      }{
        \Gamma \proves \Eq{(\Cond{\varphi}{t_1}{t_2})}{(\Cond{\varphi}{t_1}{t_2'})}
      }
      \)
      }
      \\
      & \\
      \multicolumn{2}{c}{
        \(
        {c_5:}\; \inferrule{
          \Gamma \proves \varphi : \Prop
        }{
          \Gamma \proves \Eq{(\Cond{\varphi}{t}{t})}{t}
        }
        \)
      }
      \\
      & \\
      \(
        {\mbox{eq}:}\; \inferrule{\Gamma \proves \varphi \\ \Gamma \proves
          \Eq{\varphi}{\varphi'}}{\Gamma \proves \varphi'}
      \)
      &
      \(
      {\mbox{eq-sym}:}\; \inferrule{\Gamma \proves
        \Eq{t_1}{t_2}}{\Gamma \proves \Eq{t_2}{t_1}}
      \)
      \\
      & \\
      \multicolumn{2}{c}{
      \(
      {\mbox{eq-trans}:}\; \inferrule{\Gamma \proves
        \Eq{t_1}{t_2} \\ \Gamma \proves \Eq{t_2}{t_3}}{\Gamma \proves \Eq{t_1}{t_3}}
      \)
      }
      \\
      & \\
      \multicolumn{2}{c}{
      \(
      {\mbox{eq-cong-app}:}\; \inferrule{\Gamma \proves
        \Eq{t_1}{t_1'} \\ \Gamma \proves \Eq{t_2}{t_2'}}{\Gamma \proves \Eq{(t_1 t_2)}{(t_1' t_2')}}
      \)
      }
      \\
      & \\
      \multicolumn{2}{c}{
        \(
          {\mbox{eq-$\lambda$-$\xi$}:}\; \inferrule{\Gamma \proves
            \Eq{t}{t'} \\ x \notin FV(\Gamma)}{\Gamma \proves \Eq{(\lambda x \,.\, t)}{(\lambda x \,.\, t')}}
        \)
      }
    \end{longtable}
      \[
        {i_i^\iota:}\; \inferrule{\Gamma, x_1 : \iota_{i,1}^*, \ldots, x_{n_i} : \iota_{i,n_i}^*, t
          x_{j_{i,1}}, \ldots, t x_{j_{i,k_i}} \proves t (c_i^\iota x_1 \ldots
          x_{n_i}) \\ \mathrm{for\ } i=1,\ldots,m
        }{
          \Gamma \proves \forall x : \iota \,.\, t x
        }
      \]
      where $x,x_1,\ldots,x_{n_i} \notin FV(\Gamma, t)$,
      $c_1^\iota,\ldots,c_m^\iota \in \Cc$ are all constructors
      associated with $\iota \in \Tc_I$, and
      $j_{i,1},\ldots,j_{i,k_i}$ is an increasing sequence of all
      indices $1 \le j \le n_i$ such that $\iota_{i,j} = \star$
      \smallskip
      \[
        {i_t^{\iota,k}:}\; \inferrule{
          \Gamma \proves t_j : \iota_{k,j}^* \mathrm{\ for\ } j=1,\ldots,n_k
        }{
          \Gamma \proves (c_k^\iota t_1 \ldots t_{n_k}) : \iota
        }
      \]
  \end{itemize}
  For an arbitrary set of terms~$\Gamma$, we write $\Gamma
  \proves_{\I_s} \varphi$ if there exists a finite subset $\Gamma'
  \subseteq \Gamma$ such that $\Gamma' \proves \varphi$ is derivable
  in the system~$\I_s$. We drop the subscript when irrelevant or
  obvious from the context.
\end{definition}

\begin{lemma}
  If $\Gamma \proves \varphi$ then $\Gamma, \psi \proves \varphi$.
\end{lemma}

\begin{lemma}\label{lem_I_s_subst}
  If $\Gamma \proves \varphi$ then $\Gamma[x/t] \proves \varphi[x/t]$,
  where $\Gamma[x/t] = \{ \psi[x/t] \;|\; \psi \in \Gamma \}$.
\end{lemma}


\subsection{Representing Logic}

The inference rules of~$\I_s$ may be intuitively justified by
appealing to an informal many-valued semantics. A term~$t$ may be
true, false, or something entirely different (``undefined'', a
program, a natural number, a type, \ldots). By way of an example, we
explain an informal meaning of some terms:
\begin{itemize}
\item $t : \Prop$ is true iff $t$ is true or false,
\item $\alpha : \Type$ is true iff $\alpha$ is a type,
\item $t : \alpha$ is true iff $t$ has type~$\alpha$,
  assuming~$\alpha$ is a type,
\item $\forall x : \alpha . \varphi$ is true iff $\alpha$ is a type
  and for all~$t$ of type~$\alpha$, $\varphi[x/t]$ is true,
\item $\forall x : \alpha . \varphi$ is false iff $\alpha$ is a type
  and there exists~$t$ of type~$\alpha$ such that~$\varphi[x/t]$ is
  false,
\item $t_1 \vee t_2$ is true iff $t_1$ is true or $t_2$ is true,
\item $t_1 \vee t_2$ is false iff $t_1$ is false and $t_2$ is false,
\item $t_1 \supset t_2$ is true iff $t_1$ is false or both~$t_1$
  and~$t_2$ are true,
\item $t_1 \supset t_2$ is false iff $t_1$ is true and~$t_2$ is false,
\item $\neg t$ is true iff $t$ is false,
\item $\neg t$ is false iff $t$ is true.
\end{itemize}
Obviously, $\Gamma \proves t$ is then (informally) interpreted as: for
all possible substitution instances~$\Gamma^*,t^*$ of~$\Gamma,t$,
\footnote{To be more precise, for every possible substitution of terms
  for the free variables of~$\Gamma,t$ we perform this substitution
  on~$\Gamma,t$, denoting the result by~$\Gamma^*,t^*$.} if all terms
in~$\Gamma^*$ are true, then the term~$t^*$ is also true.

Note that the logical connectives are ``lazy'', e.g. for $t_1 \vee
t_2$ to be true it suffices that~$t_1$ is true, but~$t_2$ need not
have a truth value at all -- it may be something else: a program, a
type, ``undefined'', etc. This laziness allows us to omit many
restrictions which would otherwise be needed in inference rules, and
would thus make the system less similar to ordinary logic.

The following rules may be derived in~$\I_s$.

\begin{center}
\begin{tabular}{cc}
  \(
  \supset_i:\;
  \inferrule{
    \Gamma \proves \varphi : \Prop \\ \Gamma, \varphi \proves \psi
  }{
    \Gamma \proves \varphi \supset \psi
  }
  \)\;\;
  &
  \;\;\(
  \supset_e:\;
  \inferrule{
    \Gamma \proves \varphi \supset \psi \\ \Gamma \proves \varphi
  }{
    \Gamma \proves \psi
  }
  \)
  \\
  & \\
    \(
    \supset_t:\;
    \inferrule{
      \Gamma \proves \varphi : \Prop \\ \Gamma, \varphi \proves \psi : \Prop
    }{
      \Gamma \proves (\varphi \supset \psi) : \Prop
    }
    \)
  &
  \(
    {\wedge_{i}:}\;
    \inferrule{
      \Gamma \proves \varphi \\ \Gamma \proves \psi
    }{
      \Gamma \proves \varphi \wedge \psi
    }
  \)
\end{tabular}
\end{center}

Note that in general the elimination rules for~$\wedge$ and the rules
for~$\exists$ cannot be derived from the rules for~$\vee$
and~$\forall$, because we would not be able to prove the premise
$\varphi : \Prop$ when trying to apply the rule~$\supset_i$. It is
instructive to try to derive these rules and see where the proof
breaks down.

In~$\I_s$ the only non-standard restriction in the usual inference
rules for logical connectives is the additional premise $\Gamma
\proves \varphi : \Prop$ in the rule~$\supset_i$. It is certainly
unavoidable, as otherwise Curry's paradox may be derived (see
e.g. \cite{illat01,Seldin2009}). However, we have standard classical
higher-order logic if we restrict to terms of type~$\Prop$, in the
sense that the natural deduction rules then become identical to the
rules of ordinary logic. This is made more precise in
Sect.~\ref{sec_consistent} where a sound translation from a
traditional system of higher-order logic into~$\I_s$ is described.

Note that we have the law of excluded middle only in the form $\forall
p : \Prop \,.\, p \lor \neg p$. Adding $\Gamma \proves \varphi \vee
\neg\varphi$ as an axiom for an arbitrary term~$\varphi$ gives an
inconsistent system.\footnote{By defining (see the next subsection)
  $\varphi = \neg \varphi$ one could then easily derive~$\bot$ using
  the rule~$\vee_e$ applied to~$\varphi \vee \neg\varphi$.}

It is well-known (see e.g. \cite[Chapter 11]{SorensenUrzyczyn2006})
that in higher-order logic all logical connectives may be defined
from~$\forall$ and~$\supset$ as follows.
\begin{eqnarray*}
  \bot &\equiv& \forall p : \Prop \,.\, p \\
  \neg \varphi &\equiv& \varphi \supset \bot \\
  \varphi \wedge \psi &\equiv& \forall p : \Prop \,.\, (\varphi \supset \psi \supset p) \supset p \\
  \varphi \vee \psi &\equiv& \forall p : \Prop \,.\, (\varphi \supset p) \supset (\psi \supset p) \supset p \\
  \exists x : \alpha \,.\, \varphi &\equiv& \forall p : \Prop \,.\, (\forall x : \alpha \,.\, \varphi \supset p) \supset p
\end{eqnarray*}
One may therefore wonder why we take~$\vee$ and~$\bot$ as
primitive. The answer is that if we defined the connectives by the
above equations, then the inference rules that could be derived for
them would need to contain additional restrictions. For instance, we
would be able to derive only the following variants of
$\vee$-introduction.

\begin{center}
\begin{tabular}{lr}
  \(
    {\vee_{i1}':}\; \inferrule{\Gamma \proves \varphi \\ \Gamma
      \proves \psi : \Prop}
    {\Gamma \proves \varphi \vee \psi}
  \)\;\;
  &
  \;\;\(
    {\vee_{i2}':}\; \inferrule{\Gamma \proves \psi \\ \Gamma \proves
      \varphi : \Prop}
    {\Gamma \proves \varphi \vee \psi}
  \)
\end{tabular}
\end{center}

\subsection{Equality, Recursive Definitions and Extensionality}

It is well-known (see e.g. \cite[Chapters 2, 6]{Barendregt1984}) that
since untyped $\lambda$-terms are available together with the
axiom~$\beta$ and usual rules for equality, any set of equations of
the following form has a solution for $z_1,\ldots,z_n$, where the
expressions~$\Phi_i(z_1, \ldots, z_n, x_1, \ldots, x_m)$ are arbitrary
terms with the free variables listed.
\begin{eqnarray*}
  z_1 x_1 \ldots x_m &=& \Phi_1(z_1, \ldots, z_n, x_1, \ldots, x_m) \\
  &\vdots& \\
  z_n x_1 \ldots x_m &=& \Phi_n(z_1, \ldots, z_n, x_1, \ldots, x_m)
\end{eqnarray*}
In other words, for any such set of equations, there exist terms $t_1,
\ldots, t_n$ such that for any terms $s_1, \ldots, s_m$ we have
$\proves \Eq{\left(t_i s_1 \ldots s_m\right)}{\left(\Phi_i(t_1,
  \ldots, t_n, s_1, \ldots, s_m)\right)}$ for each $i=1,\ldots,n$.

We will often define terms by such equations. In what follows we
freely use the notation $t_1 = t_2$ for $\proves \Eq{t_1}{t_2}$, or
for $\Gamma \proves \Eq{t_1}{t_2}$ when it is clear which
context~$\Gamma$ is meant. We use $t_1 = t_2 = \ldots = t_n$ to
indicate that $\Eq{t_i}{t_{i+1}}$ may be derived for
$i=1,\ldots,n-1$. We also sometimes write a term of the form
$\Eq{t_1}{t_2}$ as $t_1 = t_2$.

It is worth stressing once again that there is no a priori syntactic
distinction between terms, formulas, types, type assertions,
etc. Formally, there are only terms, but some terms are intuitively
interpreted as formulas, types, etc. In particular, the aforementioned
method of defining terms by arbitrary recursive equations may be
applied to define terms which could be intuitively considered to be
formulas, e.g. we may define a term~$\varphi$ such that $\varphi =
\neg \varphi$. Inconsistency is avoided, because it will not be
actually possible to prove $\proves_{\I_s} \varphi : \Prop$,
hence~$\varphi$ will not really be a formula. In~$\I_s$ the inference
rules serve the purpose of classifying terms into different
categories. This classification is not enforced a priori, but instead
it is a part of derivations in the logic.

In~$\I_s$ there is no rule for typing the equality~$\mathrm{Eq}$. One
consequence is that $\proves \neg(\Eq{t_1}{t_2})$ cannot be derived
for any terms $t_1,t_2$.\footnote{We mean this in a precise
  sense. This follows from our model construction.}  For this
reason~$\mathrm{Eq}$ is more like a meta-level notion of equality.

\begin{definition} \rm
  Leibniz equality $\mathrm{Eql}$ is defined as:
  \[
  \mathrm{Eql} \equiv \lambda \alpha \lambda x \lambda y . \forall p :
  \alpha \to \Prop \,.\, p x \supset p y
  \]
\end{definition}

As with~$=$, we will often write $t_1 =_\alpha t_2$ to denote $\proves
\Eql{\alpha}{t_1}{t_2}$ or $\Gamma \proves \Eql{\alpha}{t_1}{t_2}$, or
write $t_1 =_\alpha t_2$ instead of $\Eql{\alpha}{t_1}{t_2}$.

\begin{lemma}\label{lem_leibniz}
  If $\Gamma \proves \alpha : \Type$ then
  \begin{itemize}
  \item $\Gamma \proves \forall x, y : \alpha \,.\, (x =_\alpha y) : \Prop$,
  \item $\Gamma \proves \forall x : \alpha \,.\, (x =_\alpha x)$,
  \item $\Gamma \proves \forall x, y : \alpha \,.\, (x =_\alpha y) \supset
    (y =_\alpha x)$,
  \item $\Gamma \proves \forall x, y, z : \alpha \,.\, (x =_\alpha y)
    \wedge (y =_\alpha z) \supset (x =_\alpha z)$.
  \end{itemize}
\end{lemma}


The system~$\I_s$, as it is stated, is intensional with respect to
Leibniz equality. We could add the rules
\[
e_f:\;
\inferrule{
  \Gamma \proves \alpha : \Type \\ \Gamma \proves \beta : \Type
}{
  \forall f_1, f_2 : \alpha \to \beta \,.\, (\forall x
  : \alpha \,.\, f_1 x =_\beta f_2 x) \supset (f_1 =_{\alpha\to\beta} f_2)
}
\]
\[
e_b:\;
\inferrule{
  \Gamma \proves \varphi_1 \supset \varphi_2 \\
  \Gamma \proves \varphi_2 \supset \varphi_1
}{
  \Gamma \proves \varphi_1 = \varphi_2
}
\]
to obtain an extensional variant~$e\I_s$ of~$\I_s$. The system~$e\I_s$
is still consistent -- the model we construct for~$\I_s$ validates the
above rules.

\begin{lemma}
  $\proves_{e\I_s} \forall x, y : \Prop \,.\, (x =_\Prop y) \supset (x
  = y)$
\end{lemma}

\subsection{Induction and Natural Numbers}\label{sec_induction}

The system~$\I_s$ incorporates basic inductive types. In accordance
with the terminology from~\cite{BlanquiJounnaudOkada2002}, an
inductive type is basic if its constructors have no functional
arguments. This class of inductive types includes most simple commonly
used inductive types, e.g. natural numbers, lists, finite trees.

In our approach the types of constructors are encoded in the syntactic
form of the inductive type. For instance, if $\iota_0, \iota_1 \in
\Tc_I$, then $\iota = \mu(\langle \rangle, \langle \iota_0, \iota_1,
\star, \star \rangle)$ is an inductive type with constructors:
$c_1^\iota : \iota$ and $c_2^\iota : \iota_0 \to \iota_1 \to \iota \to
\iota \to \iota$.

\begin{lemma}
  If $c_i^\iota \in \Cc$ of arity~$n_i$ has signature $\langle
  \iota_1,\ldots,\iota_{n_i} \rangle$ then $\proves_{\I_s} c_i^\iota :
  \iota_1^* \to \ldots \to \iota_{n_i}^* \to \iota$.
\end{lemma}

\begin{lemma}
  $\proves_{\I_s} o_i^\iota : \iota \to \Prop$ and $\proves_{\I_s}
  \forall x : \iota \,.\, o_i^\iota x \supset (d_{i,j}^\iota x :
  \iota_{i,j}^*)$
\end{lemma}

\begin{lemma}\label{lem_ind_leibniz_implies_eq}
  If $\iota \in \Tc_I$ then $\proves_{\I_s} \forall x, y : \iota \,.\,
  x =_\iota y \supset x = y$.
\end{lemma}

We may define the type of natural numbers by $\Nat \equiv \mu(\langle
\rangle, \langle \star \rangle)$. We use the abbreviations: $0 \equiv
c_1^\Nat$ (zero), $\iszero \equiv o_1^\Nat$ (test for zero), $\s
\equiv c_2^\Nat$ (successor) and $\p \equiv \lambda x \,.\,
\Cond{(\iszero x)}{0}{(d_{2,1}^\Nat x)}$ (predecessor). The
rules~$i_i^\Nat$ and~$i_t^{\Nat,k}$ become:
\begin{longtable}{lr}
\multicolumn{2}{c}{
\(
  {n_i:}\; \inferrule{
    \Gamma \proves t 0 \\ \Gamma, x : \Nat, t x
    \proves t (\s x) \\ x \notin FV(\Gamma, t)
  }{
    \Gamma \proves \forall x : \Nat \,.\, t x
  }
\)
}
\\
\(
  {n_t^1:}\; \inferrule{
  }{
    \Gamma \proves 0 : \Nat
  }
\)
&
\(
  {n_t^2:}\; \inferrule{
    \Gamma \proves t : \Nat
  }{
    \Gamma \proves (\s t) : \Nat
  }
\)
\end{longtable}

To simplify the exposition, we discuss some properties of our
formulation of inductive types using the example of natural
numbers. Much of what we say applies to other basic inductive types,
with appropriate modifications.

The rule~$n_i$ is an induction principle for natural numbers. An
important property of this induction principle is that it places no
restrictions on~$t$. This allows us to prove by induction on natural
numbers properties of terms about which nothing is known
beforehand. In particular, we do not need to know whether~$t$ has a
$\beta$-normal form in order to apply the rule~$n_i$ to it. In
contrast, an induction principle of the form e.g.
\[
n_i':\;\; \forall f : \Nat \to \Prop \,.\, \left(\left(f 0 \wedge (\forall
x : \Nat \,.\, f x \supset f (\s x))\right) \supset \forall x : \Nat \,.\, f
x\right)
\]
would be much less useful, because to apply it to a term~$t$ we would
have to prove $t : \Nat \to \Prop$ \emph{beforehand}. Examples of the
use of the rule~$n_i$ for reasoning about possibly nonterminating
general recursive programs are given in Sect.~\ref{sec_partiality}.

We may define a recursor~$R$ for natural numbers in the following
way:
\begin{eqnarray*}
  R &=& \lambda g h x y \,.\, \Cond{(\iszero y)}{(g x)}{(h\, x\, (\p y)\,
    (R\, g\, h\, x\, (\p y)))}.
\end{eqnarray*}
Note that we need the predecessor as a primitive, because otherwise a
recursor would not be definable.

\newcommand{\mult}{\ensuremath{\cdot}}

Now $+$, $-$, $\mult$, $<$ and $\le$, usually used in infix notation,
are defined as follows.
\begin{eqnarray*}
  x + y &=& R (\lambda x \,.\, x) (\lambda x y z \,.\, \s z) x y \\
  x - y &=& R (\lambda x \,.\, x) (\lambda x y z \,.\, \p z) x y \\
  x \mult y &=& R (\lambda x \,.\, 0) (\lambda x y z \,.\, x + z) x y \\
  x \le y &=& \iszero (x - y) \\
  x < y &=& (\s x) \le y
\end{eqnarray*}

\begin{lemma}\label{lem_nat_op_well_defined}
  The following terms are derivable in the system~$\I_s$:
  \begin{itemize}
  \item $\forall x, y : \Nat \,.\, (x + y) : \Nat$, $\forall x, y :
    \Nat \,.\, (x - y) : \Nat$, $\forall x, y : \Nat \,.\, (x \mult y)
    : \Nat$,
  \item $\forall x, y : \Nat \,.\, (x \le y) : \Prop$, $\forall x, y : \Nat \,.\, (x < y) : \Prop$.
  \end{itemize}
\end{lemma}

\begin{lemma}\label{lem_le_eq}
  $\proves_{\I_s} \forall x, y : \Nat \,.\, (x \ge y) \wedge (x \le y)
  \supset (x =_\Nat y)$.
\end{lemma}

It is possible to derive Peano axioms for~$+$ and $\mult$ defined as
above.

\begin{theorem}\label{thm_peano}
  The following terms are derivable in the system~$\I_s$:
  \begin{itemize}
  \item $\forall x, y : \Nat \,.\, \left(\s x =_\Nat \s y\right)
    \supset \left(x =_\Nat y\right)$,
  \item $\forall x : \Nat \,.\, \neg (\s x =_\Nat 0)$,
  \item $\forall x : \Nat \,.\, (x + 0 =_\Nat x)$,
  \item $\forall x, y : \Nat \,.\, (x + \s y =_\Nat \s (x + y))$,
  \item $\forall x : \Nat \,.\, (x \mult 0 =_\Nat 0)$,
  \item $\forall x, y : \Nat \,.\, (x \mult (\s y) =_\Nat (x \mult y)
    + x)$.
  \end{itemize}
\end{theorem}

The following theorem shows that any function for which there exists a
measure on its arguments, which may be shown to decrease with every
recursive call in each of a finite number of exhaustive cases, is
typable in our system.

\begin{theorem}
  Suppose $\Gamma \proves \forall x_1 : \alpha_1 \ldots \forall x_n :
  \alpha_n \,.\, \varphi_1 \lor \ldots \lor \varphi_m$, $\Gamma
  \proves \alpha_j : \Type$ for $j = 1,\ldots,n$, and for
  $i=1,\ldots,m$: $\Gamma \proves \forall x_1 : \alpha_1 \ldots
  \forall x_n : \alpha_n \,.\, t_i : \beta \to \ldots \to \beta$
  where~$\beta$ occurs~$k_i + 1$ times, $\Gamma \proves \forall x_1 :
  \alpha_1 \ldots \forall x_n : \alpha_n \,.\, t_{i,j,k} : \alpha_k$
  for $j=1,\ldots,k_i$, $k=1,\ldots,n$, $x_1,\ldots,x_n \notin
  FV(f,\alpha_1,\ldots,\alpha_n,\beta)$ and
  \begin{eqnarray*}
  \Gamma \proves \forall x_1 : \alpha_1 \ldots \forall x_n :
  \alpha_n \,.\, \varphi_i &\supset& (f x_1 \ldots x_n = \\ && \;t_i (f t_{i,1,1}
  \ldots t_{i,1,n}) \ldots (f t_{i,k_i,1} \ldots t_{i,k_i,n})).
  \end{eqnarray*}
  If there is a term~$g$ such that $\Gamma \proves g : \alpha_1 \to
  \ldots \to \alpha_n \to \Nat$ and for $i=1,\ldots,m$
  \begin{eqnarray*}
  \Gamma \proves \forall x_1 : \alpha_1 \ldots \forall x_n : \alpha_n
  \,.\, \varphi_i &\supset& (\left((f x_1 \ldots x_n) :
  \beta\right) \lor \\ && \;((g t_{i,1,1} \ldots t_{i,1,n}) < (g x_1 \ldots
  x_n) \land \ldots \land \\ && \;\;(g t_{i,k_i,1} \ldots t_{i,k_i,n}) < (g x_1
  \ldots x_n)))
  \end{eqnarray*}
  where $x_1,\ldots,x_n \notin FV(g)$, then
  \[
  \Gamma \proves f : \alpha_1 \to \ldots \to \alpha_n \to \beta.
  \]
\end{theorem}

\section{Conservativity and Consistency}\label{sec_consistent}

In this section we show a sound embedding of ordinary classical
higher-order logic into~$\I_s$, which we also conjecture to be
complete. We have a completeness proof only for a restriction of this
embedding to first-order logic. We also give a brief overview of the
model construction used to establish consistency of~$\I_s$.

First, let us define the system~CPRED$\omega$ of classical
higher-order logic.
\begin{itemize}
\item The \emph{types} of CPRED$\omega$ are given by
  \[
  \Tc \;\; ::= \;\; o \;|\; \B \;|\; \Tc \rightarrow \Tc
  \]
  where~$\B$ is a specific finite set of base types. The type~$o$ is
  the type of propositions.
\item The set of terms of CPRED$\omega$ of type $\tau$, denoted
  $T_\tau$, is defined as follows:
  \begin{itemize}
  \item $V_\tau, \Sigma_\tau \subseteq T_\tau$,
  \item if $t_1 \in T_{\sigma\to\tau}$ and $t_2 \in T_\sigma$ then
    $t_1 t_2 \in T_\tau$,
  \item if $x \in V_{\tau_1}$ and $t \in T_{\tau_2}$ then $\lambda x :
    \tau_1 \,.\, t \in T_{\tau_1\to\tau_2}$,
  \item if $\varphi, \psi \in T_o$ then $\varphi \supset \psi \in
    T_o$,
  \item if $x \in V_{\tau}$ and $\varphi \in T_o$ then $\forall x :
    \tau \,.\, \varphi \in T_o$,
  \end{itemize}
  where for each type $\tau$ the set $V_\tau$ is a countable set of
  variables and $\Sigma_\tau$ is a countable set of constants. We
  assume that the sets $V_\tau$ and $\Sigma_\sigma$ are all pairwise
  disjoint. Terms of type $o$ are \emph{formulas}. As usual, we omit
  spurious brackets and assume that application associates to the
  left. We identify $\alpha$-equivalent terms, i.e. terms differing
  only in the names of bound variables are considered identical.
\item The system CPRED$\omega$ is given by the following rules and
  axioms, where $\Delta$ is a finite set of formulas, $\varphi, \psi$
  are formulas. The notation $\Delta, \varphi$ is a shorthand for
  $\Delta \cup \{\varphi\}$.

  \smallskip

  {\bf Axioms}
  \begin{itemize}
  \item $\Delta, \varphi \proves \varphi$
  \item $\Delta \proves \forall p : o \,.\, ((p \supset \bot) \supset
    \bot) \supset p$ where $\bot \equiv \forall p : o \,.\, p$
  \end{itemize}

  {\bf Rules}

  \begin{center}
  \begin{tabular}{lr}
    \(
      {\supset_i^P:}\; \inferrule{\Delta, \varphi \proves
        \psi}{\Delta \proves \varphi \supset \psi}
    \)
    &
    \(
      {\supset_e^P:}\; \inferrule{\Delta \proves \varphi \supset
        \psi \;\;\; \Delta \proves \varphi}{\Delta \proves \psi}
    \)
    \\
    & \\
    \(
      {\forall_i^P:}\; \inferrule{\Delta \proves
        \varphi}{\Delta \proves \forall x : \tau \,.\, \varphi} \; x
      \notin FV(\Delta), x \in V_\tau
    \)\;\;\;
    &
    \;\;\;\(
      {\forall_e^P:}\; \inferrule{\Delta \proves \forall x : \tau \,.\, \varphi}{\Delta
        \proves \varphi[x/t]}\; t \in T_\tau
    \)
    \\
    & \\
    \multicolumn{2}{c}{
      \(
      {\mathrm{conv}^P:}\; \inferrule{\Delta \proves \varphi \\
          \varphi =_\beta \psi}{\Delta \proves \psi}
      \)
    }
  \end{tabular}
  \end{center}
\end{itemize}

In CPRED$\omega$, we define Leibniz equality in type~$\tau \in \Tc$ by
\[
t_1 =_\tau t_2 \equiv \forall p : \tau \to o \,.\, p t_1 \supset p t_2
\]
The system~CPRED$\omega$ is intensional. An extensional
variant~E-CPRED$\omega$ may be obtained by adding the following axioms
for all $\tau, \sigma \in \Tc$:
\[
e_f^P: \forall f_1, f_2 : \tau \to \sigma \,.\, \left(\forall x : \tau
\,.\, f_1 x =_\sigma f_2 x \right) \supset (f_1 =_{\tau\to\sigma} f_2)
\]
\[
e_b^P: \forall \varphi_1, \varphi_2 : o \,.\, \left((\varphi_1 \supset
\varphi_2) \wedge (\varphi_2 \supset \varphi_1)\right) \supset
(\varphi_1 =_{o} \varphi_2)
\]

For an arbitrary set of formulas~$\Delta$ we write $\Delta \proves_S
\varphi$ if~$\varphi$ is derivable from a subset of~$\Delta$ in
system~$S$.

We now define a mapping~$\transl{-}$ from types and terms
of~CPRED$\omega$ to terms of~$\I_s$, and a mapping~$\Gamma(-)$ from
sets of terms of~CPRED$\omega$ to sets of terms of~$\I_s$ providing
necessary context. We assume that $\B \subseteq \Sigma_s$ and
$\Sigma_\tau \subseteq \Sigma_s$ for $\tau \in \Tc$, i.e. that all
base types and all constants of~CPRED$\omega$ occur as constants
in~$\I_s$, and also $V_\tau \subseteq V_s$ for $\tau \in \Tc$. The
definition of~$\transl{-}$ is inductive:
\begin{itemize}
\item $\transl{\tau} = \tau$ for $\tau \in \B$,
\item $\transl{o} = \Prop$,
\item $\transl{\tau_1\to\tau_2} = \transl{\tau_1} \to \transl{\tau_2}$
  for $\tau_1,\tau_2 \in \Tc$,
\item $\transl{c} = c$ if $c \in \Sigma_\tau$ for some $\tau \in \Tc$,
\item $\transl{x} = x$ if $x \in V_\tau$ for some $\tau \in \Tc$,
\item $\transl{t_1 t_2} = \transl{t_1} \transl{t_2}$,
\item $\transl{\lambda x : \tau \,.\, t} = \lambda x \,.\,
  \transl{t}$,
\item $\transl{\varphi \supset \psi} = \transl{\varphi} \supset
  \transl{\psi}$,
\item $\transl{\forall x : \tau \,.\, \varphi} = \forall x :
  \transl{\tau} \,.\, \transl{\varphi}$.
\end{itemize}
If $\Delta$ is a set of formulas, then~$\transl{\Delta}$ denotes the
image of~$\transl{-}$ on~$\Delta$. The set~$\Gamma(\Delta)$ is defined
to contain the following:
\begin{itemize}
\item $x : \transl{\tau}$ for all $\tau \in \Tc$ and all $x \in
  FV(\Delta)$ such that $x \in V_\tau$,
\item $c : \transl{\tau}$ for all $\tau \in \Tc$ and all $c \in
  \Sigma_\tau$,
\item $\tau : \Type$ for all $\tau \in \B$,
\item $y : \tau$ for all $\tau \in \B$ and some $y \in V_\tau$ such
  that $y \notin FV(\Delta)$.
\end{itemize}
The last point is needed because in ordinary logic one always assumes
that the domains are non-empty.

\begin{theorem}\label{thm_sound_hol}
  If $\Delta \proves_{\mathrm{CPRED}\omega} \varphi$ then
  $\Gamma(\Delta, \varphi), \transl{\Delta} \proves_{\I_s}
  \transl{\varphi}$. The same holds if we change~CPRED$\omega$ to
  \mbox{E-CPRED$\omega$} and~$\I_s$ to~$e\I_s$.
\end{theorem}

The above theorem shows that~$\I_s$ may be considered an extension of
ordinary higher-order logic. This extension is essentially obtained by
relaxing typing requirements on allowable
$\lambda$-terms. Type-checking is obviously undecidable in~$\I_s$, but
the purpose of types in illative systems is not to have a decidable
method for syntactic correctness checks, but to provide general means
for classifying terms into various categories. In practice, one might
still want to have a decidable (and necessarily incomplete) method for
checking correctness of some designated type assertions. Such a method
may be obtained by employing any type-checking algorithm sound
w.r.t.~$\I_s$.\footnote{By soundness we mean that if the algorithm
  declares $t : \alpha$ correct (in a context~$\Gamma$), then $\Gamma
  \proves_{\I_s} t : \alpha$. Completeness is the implication in the
  other direction. The point is that one would want some standard
  type-checking algorithms to be modified to work on terms of~$\I_s$,
  by declaring incorrect all type assertions not conforming to the
  syntax of (a straightforward translation of) the assertions handled
  by the algorithm.} However, the difference would be that a priori
type-checks would not be enforced in every situation. Occasionally,
with some more complex recursive functions, it might be convenient to
forgo these checks and reason about the types of such functions
explicitly, using the rules of~$\I_s$.

\begin{conjecture}\label{conj_conservative_hol}
  If $\Gamma(\Delta, \varphi), \transl{\Delta} \proves_{\I_s}
  \transl{\varphi}$ then $\Delta \proves_{\mathrm{CPRED}\omega}
  \varphi$. The same holds if we change~CPRED$\omega$ to
  \mbox{E-CPRED$\omega$} and~$\I_s$ to~$e\I_s$.
\end{conjecture}

We were able to prove this conjecture only for first-order logic. The
system of classical first-order logic (FOL) is obtained by
restricting~CPRED$\omega$ in obvious ways (leaving only one base
type~$\iota$, disallowing $\lambda$-abstraction, allowing
quantification only over~$\iota$, and constants only of types $\iota$,
$\iota \to \ldots \to \iota \to \iota$ or $\iota \to \ldots \to \iota
\to o$).

\begin{theorem}\label{thm_embedding_01}
  If $\I = \I_s$ or $\I = e\I_s$ then
  \[
  \Delta \proves_{\mathrm{FOL}} \varphi \mathrm{\ \ \ iff\ \ \ }
  \Gamma(\Delta, \varphi), \transl{\Delta} \proves_{\I}
  \transl{\varphi}
  \]
\end{theorem}

\begin{theorem}\label{thm_consistent_01}
  The systems~$\I_s$ and~$e\I_s$ are consistent,
  i.e. $\not\proves_{\I_s} \bot$ and $\not\proves_{e\I_s} \bot$.
\end{theorem}

This follows from Theorem~\ref{thm_embedding_01}, but we actually
prove Theorem~\ref{thm_consistent_01} first by constructing a model,
and then use this construction to show Theorem~\ref{thm_embedding_01}.

We now give an informal overview of the model construction. To
simplify the exposition we pretend~$\I_s$ allows only function
types. Inductive types and subtypes add some technicalities, but the
general idea of the construction remains the same. This overview is
necessarily very brief. An interested reader is advised to consult a
technical appendix for more details.

An $\I_s$-model is defined essentially as a $\lambda$-model (see
e.g. \cite[Chapter 5]{Barendregt1984}) with designated elements
interpreting the constants of~$\I_s$, satisfying certain
requirements. By~$\valuation{t}{}{\Mc}$ we denote the interpretation
of the $\I_s$-term~$t$ in a model~$\Mc$. The conditions imposed on an
$\I_s$-model express the meaning of each rule of~$\I_s$ according to
the intuitive semantics. For instance, we have the conditions:
\begin{itemize}
\item[($\forall_\top$)] for $a \in \Mc$, if
  $\valuation{\mathrm{Is}}{}{\Mc} \cdot a \cdot
  \valuation{\Type}{}{\Mc} = \valuation{\top}{}{\Mc}$ and for all $c
  \in \Mc$ such that $\valuation{\mathrm{Is}}{}{\Mc} \cdot c \cdot a =
  \valuation{\top}{}{\Mc}$ we have $b \cdot c =
  \valuation{\top}{}{\Mc}$ then $\valuation{\forall}{}{\Mc} \cdot a
  \cdot b = \valuation{\top}{}{\Mc}$,
\item[($\forall_e$)] for $a,b \in \Mc$, if $\valuation{\forall}{}{\Mc}
  \cdot a \cdot b = \valuation{\top}{}{\Mc}$ then for all $c \in A$
  such that $\valuation{\mathrm{Is}}{}{\Mc} \cdot c \cdot a =
  \valuation{\top}{}{\Mc}$ we have $b \cdot c =
  \valuation{\top}{}{\Mc}$.
\end{itemize}
Here $\cdot$ is the application operation in the model.

We show that the semantics based on $\I_s$-models is sound
for~$\I_s$. Then it suffices to construct a non-trivial (i.e. such
that $\valuation{\top}{}{} \ne \valuation{\bot}{}{}$) $\I_s$-model to
establish consistency of~$\I_s$. The model will in fact satisfy
additional conditions corresponding to the rules~$e_f$ and~$e_b$, so
we obtain consistency of~$e\I_s$ as well.

The model is constructed as the set of equivalence classes of a
certain relation~$\ipeqvred$ on the set of so called semantic terms. A
semantic term is a well-founded tree whose leaves are labelled with
variables or constants, and whose internal nodes are labelled with
$\cdot$, $\lambda x$ or $\All\tau$. The intuitive interpretation of
nodes labelled with~$\cdot$ or $\lambda x$ is obvious. For semantic
terms with the roots labelled with~$\cdot$ and~$\lambda x$ we use the
abbreviations~$t_1 t_2$ and~$\lambda x . t$, respectively. A node
labelled with~$\All\tau$ represents universal quantification over a
set of constants~$\tau$, i.e. it ``represents'' the statement: for all
$c \in \tau$, $t c$ is true. Such a node has one child for each $c \in
\tau$. The relation~$\ipeqvred$ will be defined as the equivalence
relation generated by a certain reduction relation~$\ipcontr$ on
semantic terms. The relation~$\ipcontr$ will
satisfy\footnote{Substitution is defined for semantic terms in an
  obvious way, avoiding variable capture.}: $(\lambda x . t_1) t_2
\ipcontr t_1[x/t_2]$, $\lor \top t \ipcontr \top$, $\lor \bot \bot
\ipcontr \bot$, etc. The question is how to define~$\ipcontr$
for~$\forall t_1 t_2$ so that the resulting structure
satisfies~($\forall_\top$). One could try closing~$\ipcontr$ under the
rule:
\begin{itemize}
\item if $\mathrm{Is}\, t_1\, \Type \ipreduces \top$ and for all
  semantic terms~$t$ such that $\mathrm{Is}\, t\, t_1 \ipreduces \top$
  we have $t_2 t \ipreduces \top$, then $\forall t_1 t_2 \ipcontr
  \top$.
\end{itemize}
However, there is a negative reference to~$\ipcontr$ here, so the
definition would not be monotone, and we would not necessarily reach a
fixpoint. This is a major problem. We somehow need to know the range
of all quantifiers beforehand. However, the range (i.e. the set of all
semantic terms~$t$ such that $t_1 t \ipreduces \top$) depends on the
definition of~$\ipcontr$, so it is not at all clear how to achieve
this.

Fortunately, it is not so difficult to analyze a priori the form of
types of~$\I_s$. Informally, if $t : \Type$ is true, then $t$
corresponds to a set in~$\Tc$, where~$\Tc$ is defined as follows,
ignoring subtypes and inductive types, but instead introducing a base
type~$\delta$ of individuals.
\begin{itemize}
\item $\delta, \Bool \in \Tc$ where $\Bool = \{\top,\bot\}$
  and~$\delta$ is an arbitrary set of fresh constants.
\item If $\tau_1,\tau_2 \in \Tc$ then $\tau_2^{\tau_1} \in \Tc$, where
  $\tau_2^{\tau_1}$ is the set of all set-theoretic functions
  from~$\tau_1$ to~$\tau_2$.
\end{itemize}
We take the elements of~$\Tc$ and~$\bigcup \Tc \setminus \Bool$ as
fresh constants, i.e. they may occur as constants in semantic
terms. The elements of~$\bigcup \Tc$ are \emph{canonical
  constants}. If $c \in \tau_2^{\tau_1}$ and $c_1 \in \tau_1$ then we
write $\Fc(c)(c_1)$ instead of $c(c_1)$ to avoid confusion with the
semantic term $c c_1$. We then define a relation~$\succ$ satisfying:
\begin{itemize}
\item $c \succ c$ for a canonical constant~$c$,
\item if $c \in \tau_2^{\tau_1}$ and for all $c_1 \in \tau_1$ there
  exists a semantic term~$t'$ such that $t c_1 \ipreduces t' \succ
  \Fc(c)(c_1)$, then $t \succ c$.
\end{itemize}
Intuitively, $t \succ c \in \tau$ holds if~$c$ ``simulates''~$t$ in
type~$\tau$, i.e. $t$ behaves exactly like~$c$ in every context where
a term of type~$\tau$ is ``expected''.

The relation~$\ipcontr$ is then defined by transfinite induction in a
monotone way. It will satisfy e.g.:
\begin{itemize}
\item if $t \succ c \in \tau \in \Tc$ then $\mathrm{Is}\,t\,\tau
  \ipcontr \top$,
\item if $t \succ c_1 \in \tau_1$ and $c \in \tau_2^{\tau_1}$ then $c
  t \ipcontr \Fc(c)(c_1)$,
\item $\mathrm{Fun}\,\tau_1\,\tau_2 \ipcontr \tau_2^{\tau_1}$,
\item $\forall \tau t \ipcontr t'$ where the label at the root of~$t'$
  is~$\All\tau$, and for each $c \in \tau$, $t'$ has a child~$t c$,
\item $t \ipcontr \top$ if the label of the root of~$t$ is~$\All\tau$,
  and all children of~$t$ are labelled with~$\top$,
\item if $t_c \ipreduces t_c'$ for all $c \in \tau \in \Tc$, the label
  of the root of~$t$ is~$\All\tau$, and $\{ t_c \;|\; c \in \tau \}$
  is the set of children of~$t$, then $t \ipcontr t'$, where the label
  of the root of~$t'$ is~$\All\tau$ and $\{t_c' \;|\; c \in \tau\}$ is
  the set of children of~$t'$.
\end{itemize}
We removed negative references to~$\ipcontr$, but it is not easy to
show that the resulting model will satisfy the required
conditions. Two key properties established in the correctness proof
are:
\begin{enumerate}
\item $\ipcontr$ has the Church-Rosser property,
\item if $t_2 \succ c$ and $t_1 c \ipreduces d \in \{\top,\bot\}$ then
  $t_1 t_2 \ipreduces d$.
\end{enumerate}
The second property shows that quantifying over only canonical
constants of type~$\tau$ is in a sense equivalent to quantifying over
all terms of type~$\tau$. This is essential for establishing e.g. the
condition~($\forall_\top$).

Both of these properties have rather intricate proofs. Essentially,
the proofs show certain commutation and postponement properties
for~$\ipcontr$, $\succ$ and other auxiliary relations. The proofs
proceed by induction on lexicographic products of various ordinals and
other parameters associated with the relations and terms involved.

\section{Partiality and General Recursion}\label{sec_partiality}

In this section we give some examples of proofs in~$\I_s$ of
properties of functions defined by recursion. For lack of space, we
give only informal indications of how formal proofs may be obtained,
assuming certain basic properties of operations on natural
numbers. The transformation of the given informal arguments into
formal proofs in~$\I_s$ is not difficult. Mostly complete formal
proofs may be found in a technical appendix.

\newcommand{\subp}{\ensuremath{\mathrm{subp}}}

\begin{example}
  Consider a term $\subp$ satisfying the following recursive equation:
  \[
  \subp = \lambda i j \,.\, \Cond{(i =_\Nat j)}{0}{\left(\left(\subp
    \, i \, \left(j + 1\right)\right) + 1\right)}.
  \]
  If $i \ge j$ then $\subp\,i\,j = i - j$. If $i < j$ then
  $\subp\,i\,j$ does not terminate. An appropriate specification for
  $\subp$ is $\forall i, j : \Nat \,.\, (i \ge j) \supset (\subp\, i\,
  j = i - j)$.

  Let $\varphi(y) = \forall i : \Nat \,.\, \forall j : \Nat \,.\,
  \left(i \ge j \supset y =_\Nat i - j \supset
  \subp\,i\,j\,=\,i-j\right)$. We show by induction on~$y$ that
  $\forall y : \Nat \,.\, \varphi(y)$.

  First note that under the assumptions $y : \Nat$, $i : \Nat$, $j :
  \Nat$ it follows from Lemma~\ref{lem_nat_op_well_defined} that $(i
  \ge j) : \Prop$ and $(y =_\Nat i - j) : \Prop$. Hence, whenever $y :
  \Nat$, to show $i \ge j \supset y =_\Nat i - j \supset
  \subp\,i\,j\,=\,i-j$ it suffices to derive $\subp\,i\,j\,=\,i-j$
  under the assumptions $i \ge j$ and $y =_\Nat i - j$. By
  Lemma~\ref{lem_ind_leibniz_implies_eq} the assumption~$y =_\Nat i -
  j$ may be weakened to~$y = i - j$.

  In the base step it thus suffices to show $\subp\,i\,j\,=\,i-j$
  under the assumptions $i : \Nat$, $j : \Nat$, $i \ge j$, $i - j =
  0$. From $i - j = 0$ we obtain $\iszero(i - j)$, so $j \ge i$. From
  $i \ge j$ and $i \le j$ we derive $i =_\Nat j$ by
  Lemma~\ref{lem_le_eq}. Then $\subp\,i\,j\,=\,i-j$ follows by simple
  computation (i.e. by applying rules for~\mbox{Eq} and appropriate
  rules for the conditional).

  In the inductive step we have $\varphi(y)$ for $y : \Nat$ and we
  need to obtain $\varphi(\s y)$. It suffices to show $\subp\,i\,j = i
  - j$ under the assumptions~$i : \Nat$, $j : \Nat$ and~$\s y = i -
  j$. Because $\s y \ne_\Nat 0$ we have $i \ne_\Nat j$, hence
  $\subp\,i\,j = \s(\subp\,i\,(\s j))$ follows by computation. Using
  the inductive hypothesis we now conclude $\subp\,i\,(\s j) = i - (\s
  j)$, and thus $\subp\,i\,(\s j) =_\Nat i - (\s j)$ by reflexivity
  of~$=_\Nat$ on natural numbers. Then it follows by properties of
  operations on natural numbers that~$\s(\subp\,i\,(\s j)) =_\Nat i -
  j$. By Lemma~\ref{lem_ind_leibniz_implies_eq} we obtain the thesis.

  We have thus completed an inductive proof of~$\forall y : \Nat \,.\,
  \varphi(y)$. Now we use this formula to derive $\subp\,i\,j\,=\,i-j$
  under the assumptions $i : \Nat$, $j : \Nat$, $i \ge j$. Then it
  remains to apply implication introduction and $\forall$-introduction
  twice.

  In the logic of PVS~\cite{RushbyOwreShankar1998} one may
  define~$\subp$ by specifying its domain precisely using predicate
  subtypes and dependent types, somewhat similarly to what is done
  here. However, an important distinction is that we do not require a
  domain specification to be a part of the definition. In an
  interactive theorem prover based on our formalism no proof
  obligations would need to be generated to establish termination
  of~$\subp$ on its domain.

  Note that because domain specification is not part of the definition
  of~$\subp$, we may easily derive $\varphi \equiv \forall i, j : \Nat
  \,.\, \left(\left(\subp\, i\, j = i - j\right) \vee
  \,\left(\subp\,j\, i = j - i\right)\right)$. This is not possible in
  PVS because the formula~$\varphi$ translated to PVS generates false
  proof obligations~\cite{RushbyOwreShankar1998}.
\end{example}

\begin{example}
  The next example is a well-known ``challenge'' posed by McCarthy:
  \[
  f(n) = \Cond{(n > 100)}{(n - 10)}{(f(f(n + 11)))}
  \]
  For $n \le 101$ we have $f(n) = 91$, which fact may be proven by
  induction on $101 - n$. This function is interesting because of its
  use of nested recursion. Termination behavior of a nested recursive
  function may depend on its functional behavior, which makes
  reasoning about termination and function value interdependent. This
  creates problems for systems with definitional restrictions of
  possible forms of recursion. Below we give an indication of how a
  formal proof of~$\forall n : \Nat \,.\, n \le 101 \supset f(n) = 91$
  may be derived in~$\I_s$. Lemma~\ref{lem_nat_op_well_defined} is
  used implicitly with implication introduction.

  Let $\varphi(y) \equiv \forall n : \Nat \,.\, n \le 101 \supset 101 - n
  \le y \supset f(n) = 91$. We prove $\forall y : \Nat \,.\, \varphi(y)$
  by induction on~$y$.

  In the base step we need to prove~$f(n) = 91$ under the
  assumptions~$n : \Nat$, $n \le 101$ and~$101 - n \le y = 0$. We
  have~$n =_\Nat 101$, hence $n = 101$, and the thesis follows by
  simple computation.

  In the inductive step we distinguish three cases:
  \begin{enumerate}
  \item $n + 11 > 101$ and $n < 101$,
  \item $n + 11 > 101$ and $n \ge 101$,
  \item $n + 11 \le 101$.
  \end{enumerate}
  We need to prove~$f(n) = 91$ under the assumptions of the inductive
  hypothesis~$y : \Nat, \forall m : \Nat \,.\, m \le 101 \supset 101 -
  m \le y \supset f(m) = 91$, and of $n : \Nat$, $n \le 101$ and~$101
  - n \le (\s y)$.

  In the first case we have $f(n + 11) = n + 1$ and $n + 1 \le
  101$. Hence by the inductive hypothesis we conclude~$100 - n \le y
  \supset f(n + 1) = 91$. From~$101 - n \le \s y$ we infer~$100 - n
  \le y$, and hence~$f(n + 1) = 91$. Since~$n \le 100$ it follows by
  computation that~$f(n) = f(f(n + 11)) = f(n + 1) = 91$.

  In the second case~$n = 101$ and the thesis follows by simple
  computation.

  In the third case, from~$101 - n \le \s y$ we infer~$101 - (n + 11)
  \le y$. Since $n + 11 \le 101$ we conclude by the inductive
  hypothesis that~$f(n + 11) = 91$. Because~$n + 11 \le 101$, so~$n
  \le 100$, and by definition we infer~$f(n) = f(f(n + 11)) =
  f(91)$. Now we simply compute~$f(91) = f(f(102)) = f(92) = f(f(103))
  = \ldots = f(100) = f(f(111)) = f(101) = 91$ (i.e. we apply rules
  for~$\mbox{Eq}$ and rules for the conditional an appropriate number
  of times).

  This concludes the inductive proof of~$\forall y : \Nat \,.\, \forall n
  : \Nat \,.\, n \le 101 \supset 101 - n \le y \supset f(n) = 91$. Having
  this it is not difficult to show~$\forall n : \Nat \,.\, n \le 101
  \supset f(n) = 91$.

  Note that the computation of~$f(91)$ in the inductive step relies on
  the fact that in our logic values of functions may always be
  computed for specific arguments, regardless of what we know about
  the function, regardless of whether it terminates in general.
\end{example}

\section{Related Work}\label{sec_related}

In this section we discuss the relationship between~$\I_s$ and the
traditional illative system~$\I_\omega$. We also briefly survey some
approaches to dealing with partiality and general recursion in proof
assistants. A general overview of the literature relevant to this
problem may be found in~\cite{BoveKraussSozeau2012}.

\subsection{Relationship with Systems of Illative Combinatory Logic}

In terms of the features provided, the system~$\I_s$ may be considered
an extension of~$\I_\omega$ from~\cite{Czajka2013Accepted}. However,
there are some technical differences between~$\I_s$ and traditional
systems of illative combinatory logic. For one thing, traditional
systems strive to use as few constants and rules as possible. For
instance, $\I_\omega$ has only two primitive constants, disregarding
constants representing base types. Because of this in~$\I_\omega$
e.g. $\Is = \lambda xy \,.\, y x$ and $\Prop = \lambda x \,.\, \Type
(\lambda y . x)$, using the notation of the present paper. Moreover,
the names of the primitive constants and the notations employed when
discussing traditional illative systems are not in common use
today. We will not explain these technicalities in any more
detail. The reader may
consult~\cite{Seldin2009,illat01,Czajka2013Accepted} for more
information on illative combinatory logic.

The system~$\I_\omega$ from~\cite{Czajka2013Accepted} is a direct
extension of~$\I\Xi$ from~\cite{illat01} to higher-order logic. The
ideas behind~$\I_\omega$ date back to~\cite{Bunder1983}, or even
earlier as far as the general form of restrictions in inference rules
is concerned.

Below we briefly describe a system~$\I_\omega'$ which is a variant
of~$\I_\omega$ adapted to our notation. It differs somewhat
from~$\I_\omega$, mostly by taking more constants as primitive, and
thus having more rules and axioms. However, we believe that despite
these differences it is reasonably close to~$\I_\omega$ and shares its
essential properties.

The terms of~$\I_\omega'$ are those of~$\I_s$, except that we do not
allow subtypes, inductive types, Eq, Cond, $\vee$, $\bot$ and
$\epsilon$. There are also additional primitive constants: $\omega$
(the type of all terms), $\varepsilon$ (the empty type)
and~$\supset$. The axioms are: $\Gamma, \varphi \proves \varphi$,
$\Gamma \proves \Prop : \Type$, $\Gamma \proves \varepsilon : \Type$,
$\Gamma \proves \omega : \Type$, $\Gamma \proves t : \omega$. The
rules are: $\forall_i$, $\forall_e$, $\forall_t$, $\supset_i$,
$\supset_e$, $\supset_t$, $\to_i$, $\to_e$, $\to_t$, $p_i$, and the
rules:
\begin{longtable}{lr}
\(
\mathrm{conv}: \inferrule{
  \Gamma \proves \varphi \\ \varphi =_\beta \psi
}{
  \Gamma \proves \psi
}
\)
&
\(
\varepsilon_\bot: \inferrule{
  \Gamma \proves t : \varepsilon
}{
  \Gamma \proves \bot
}
\)
\\
& \\
\multicolumn{2}{c}{
\(
\to_p: \inferrule{
  \Gamma \proves \alpha : \Type \\ \Gamma, x : \alpha \proves ((t x) :
  \beta) : \Prop \\ x \notin FV(\Gamma, t)
}{
  \Gamma \proves (t : \alpha \to \beta) : \Prop
}
\)
}
\end{longtable}
Here $\varphi =_\beta \psi$ is a meta-level side-condition expressing
$\beta$-equivalence of the terms~$\varphi$ and~$\psi$.

\subsection{Partiality and Recursion in Proof Assistants}

Perhaps the most common way of dealing with recursion in interactive
theorem provers is to impose certain syntactic restrictions on the
form of recursive definitions so as to guarantee well-foundedness. For
instance, the {\bf fix} construct in Coq allows for structurally
recursive definitions whose well-foundedness must be checked by a
built-in automatic syntactic termination checker. Some systems,
e.g. ACL2 or PVS, pass the task of proving termination to the
user. Such systems require that a well-founded relation or a measure
be given with each recursive function definition. Then the system
generates so called proof obligations, or termination conditions,
which state that the recursive calls are made on smaller
arguments. The user must solve, i.e. prove, these obligations.

The method of restricting possible forms of recursive definitions
obviously works only for total functions. If a function does not in
fact terminate on some elements of its specified domain, then it
cannot be introduced by a well-founded definition. One solution is to
use a rich type system, e.g. dependent types combined with predicate
subtyping, to precisely specify function domains so as to rule out the
arguments on which the function does not terminate. This approach is
adopted by PVS~\cite{RushbyOwreShankar1998}. A related method of
introducing general recursive functions in constructive type theory is
to first define a special inductive accessibility predicate which
precisely characterises the domain~\cite{BoveCapretta2005}. The
function is then defined by structural recursion on the proof that the
argument satisfies the accessibility predicate.

A different approach to dealing with partiality and general recursion
is to use a special logic which allows partial functions
directly. Systems adopting this approach are often based on variants
of the logic of partial terms of Beeson \cite{Feferman1995},
\cite{Beeson1986b}. For instance, the IMPS interactive theorem
prover~\cite{FarmerGuttmanThayer1993} uses Farmer's logic PF of
partial functions~\cite{Farmer1990}, which is essentially a variant of
the logic of partial terms adapted to higher-order logic. In these
logics there is an additional definedness predicate which enables
direct reasoning about definedness of terms.

The above gives only a very brief overview. There are many approaches
to the problem of partiality and general recursion in interactive
theorem provers, most of which we didn't mention. We do not attempt
here to provide a detailed comparison with a multitude of existing
approaches or give in-depth arguments in favor of our system. For such
arguments to be entirely convincing, they would need to be backed up
by extensive experimentation in proving properties of sizable programs
using a proof assistant based on our logic. No such proof assistants
yet exist and no such experimentation has been undertaken. In
contrast, our interest is theoretical.

\section{Conclusion}

We have presented a system~$\I_s$ of classical higher-order illative
$\lambda$-calculus with subtyping and basic inductive types. A
distinguishing characteristic of~$\I_s$ is that it is based on the
untyped $\lambda$-calculus. Therefore, it allows recursive definitions
of potentially non-terminating functions directly. The inference rules
of~$\I_s$ are formulated in a way that makes it possible to apply them
even when some of the terms used in the premises have not been proven
to belong to any type. Additionally, our system may be considered an
extension of ordinary higher-order logic, obtained by relaxing the
typing restrictions on allowable $\lambda$-terms. We believe these
facts alone make it relevant to the problem of partiality and
recursion in proof assistants, and the system at least deserves some
attention.

Some open problems related to~$\I_s$ are as follows.
\begin{enumerate}
\item Is~$\I_s$ conservative over higher-order logic
  (Conjecture~\ref{conj_conservative_hol})?
\item How to best incorporate a broader class of inductive types than
  just basic inductive types, e.g. all strictly positive inductive
  types?
\item How far is it possible to broaden the class of allowed types and
  still have a consistency proof in ZFC? For instance, in our model
  construction we could try naively interpreting dependent types by
  set-theoretic cartesian products when constructing the set~$\Tc$
  (see the overview of the consistency proof in
  Sect.~\ref{sec_consistent}), but we would run out of sets. We
  conjecture that our construction may be modified to incorporate
  dependent types in the way indicated if we work in ZFC with one
  strongly inaccessible cardinal. This modification should not pose
  any fundamental difficulties. Is it possible prove consistency
  of~$\I_s$ with dependent types in plain ZFC?
\item Can the premises $\Gamma \proves \varphi : \Prop$ in the
  rules~$c_3$ and~$c_4$ be removed?
\item Can the premise $\Gamma \proves \varphi : \Prop$ in the
  rule~$c_5$ be removed?
\end{enumerate}

Note that the premises $\Gamma \proves \varphi : \Prop$ cannot be
removed in~$c_3$, $c_4$ and~$c_5$ simultaneously. Let~$\varphi$ be
such that $\varphi = \neg \varphi$. Since $\varphi \proves \bot$ and
$\neg\varphi \proves \bot$, it is easy to see that we would have both
$\proves \Cond{\varphi}{\top}{\bot} = \top$ and $\proves
\Cond{\varphi}{\top}{\bot} = \bot$.

\addcontentsline{toc}{section}{References}
\bibliography{biblio}{}

\begin{thebibliography}{10}

\bibitem{illat01}
Barendregt, H., Bunder, M.W., Dekkers, W.:
\newblock Systems of illative combinatory logic complete for first-order
  propositional and predicate calculus.
\newblock Journal of Symbolic Logic \textbf{58}(3) (1993)  769--788

\bibitem{Seldin2009}
Seldin, J.P.:
\newblock The logic of {C}hurch and {C}urry.
\newblock In Gabbay, D.M., Woods, J., eds.: Logic from Russell to Church.
  Volume~5 of Handbook of the History of Logic.
\newblock North-Holland (2009)  819--873

\bibitem{Czajka2013Accepted}
Czajka, {\L}.:
\newblock Higher-order illative combinatory logic.
\newblock Journal of Symbolic Logic (2013) Accepted. Preprint available at
  \verb#http://arxiv.org/abs/1202.3672#.

\bibitem{Czajka2011}
Czajka, {\L}.:
\newblock A semantic approach to illative combinatory logic.
\newblock In: Computer Science Logic, 25th International Workshop / 20th Annual
  Conference of the EACSL, CSL 2011, September 12-15, 2011, Bergen, Norway,
  Proceedings. Volume~12 of LIPIcs., Schloss Dagstuhl -- Leibniz-Zentrum f\"ur
  Informatik (2011)  174--188

\bibitem{illat02}
Dekkers, W., Bunder, M.W., Barendregt, H.:
\newblock Completeness of the propositions-as-types interpretation of
  intuitionistic logic into illative combinatory logic.
\newblock Journal of Symbolic Logic \textbf{63}(3) (1998)  869--890

\bibitem{SorensenUrzyczyn2006}
S{\o}rensen, M., Urzyczyn, P.:
\newblock Lectures on the Curry-Howard isomorphism. Volume 149 of Studies in
  Logic and the Foundations of Mathematics.
\newblock Elsevier (2006)

\bibitem{Barendregt1984}
Barendregt, H.P.:
\newblock The lambda calculus: Its syntax and semantics. Revised edn.
\newblock North Holland (1984)

\bibitem{BlanquiJounnaudOkada2002}
Blanqui, F., Jouannaud, J., Okada, M.:
\newblock Inductive-data-type systems.
\newblock Theoretical Computer Science \textbf{272}(1) (2002)  41--68

\bibitem{RushbyOwreShankar1998}
Rushby, J., Owre, S., Shankar, N.:
\newblock Subtypes for specifications: Predicate subtyping in {PVS}.
\newblock IEEE Transactions on Software Engineering \textbf{24}(9) (1998)

\bibitem{BoveKraussSozeau2012}
Bove, A., Krauss, A., Sozeau, M.:
\newblock Partiality and recursion in interactive theorem pro\-vers: An
  over\-view.
\newblock Mathematical Structures in Computer Science (to appear, 2012)

\bibitem{Bunder1983}
Bunder, M.W.:
\newblock Predicate calculus of arbitrarily high finite order.
\newblock Archive for Mathematical Logic \textbf{23}(1) (1983)  1--10

\bibitem{BoveCapretta2005}
Bove, A., Capretta, V.:
\newblock Modelling general recursion in type theory.
\newblock Mathematical Structures in Computer Science \textbf{15}(4) (2005)
  671--708

\bibitem{Feferman1995}
Feferman, S.:
\newblock Definedness.
\newblock Erkenntnis \textbf{43} (1995)  295--320

\bibitem{Beeson1986b}
Beeson, M.J.:
\newblock Proving programs and programming proofs.
\newblock In Marcus, R.B., Dorn, G., Weingartner, P., eds.: Logic, Methodology
  and Philosophy of Science VII.
\newblock North-Holland (1986)  51--82

\bibitem{FarmerGuttmanThayer1993}
Farmer, W.M., Guttman, J.D., Thayer, F.J.:
\newblock {IMPS}: An interactive mathematical proof system.
\newblock Journal of Automated Reasoning \textbf{11}(2) (1993)  213--248

\bibitem{Farmer1990}
Farmer, W.M.:
\newblock A partial functions version of {C}hurch's simple theory of types.
\newblock Journal of Symbolic Logic \textbf{55}(3) (1990)  1269--1291

\end{thebibliography}
\bibliographystyle{splncs}
\newpage

\appendix

\newcommand{\nocontentsline}[3]{}
\newcommand{\tocless}[2]{\bgroup\let\addcontentsline=\nocontentsline#1{#2}\egroup}

\renewcommand{\thesection}{Appendix~\Alph{section}}
\tocless\section{Derived Rules for Implication}\label{app_logic}
\renewcommand{\thesection}{\Alph{section}}
\addcontentsline{toc}{section}{\thesection\hspace{0.5em} Derived Rules
for Implication}

In all derivations in this and subsequent sections we omit certain
steps and rule assumptions, simplify inferences, and generally only
give sketches of completely formal proofs, omitting the parts which
may be easily reconstructed by the reader.

\begin{lemma}\label{lem_01}
 If $\Gamma, \varphi \proves \psi$, where $x \notin FV(\Gamma,
 \varphi, \psi)$ and $y \notin FV(\varphi)$, then $\Gamma, x : \set{y
   : \Prop}{\varphi} \proves \psi$.
\end{lemma}

\begin{proof}
  Straightforward induction on the length of derivation, using
  rule~$s_e$ to show that $\Gamma \proves x : \set{y :
    \Prop}{\varphi}$ implies $\Gamma \proves \varphi$, if $x \notin
  FV(\Gamma, \varphi)$ and $y \notin FV(\varphi)$.
\end{proof}

Now the rules for~$\supset$ are derived as follows.

\medskip

The rule
\[
\supset_i:\;
\inferrule{
  \Gamma \proves \varphi : \Prop \\ \Gamma, \varphi \proves \psi
}{
  \Gamma \proves \varphi \supset \psi
}
\]
follows by
\[
\inferrule*{
  \inferrule*{
  }{\Gamma \proves \Prop : \Type} \\
  \inferrule*{
    \Gamma \proves \varphi : \Prop \\
    y \notin FV(\Gamma, \varphi)
  }{
    \Gamma, y : \Prop \proves \varphi : \Prop
  }
}{
  \Gamma \proves \set{y : \Prop}{\varphi} : \Type
} (a)
\]
\[
\inferrule*{
 (a) \\
  \inferrule*{
    \Gamma, \varphi \proves \psi
  }{
    \Gamma, x : \set{y : \Prop}{\varphi} \proves \psi
  }
}{
  \Gamma \proves \forall x : \set{y : \Prop}{\varphi} \,.\, \psi
}
\]

The rule
\[
\supset_e:\;
\inferrule{
  \Gamma \proves \varphi \supset \psi \\ \Gamma \proves \varphi
}{
  \Gamma \proves \psi
}
\]
follows by
\[
\inferrule*{
  \inferrule*{
    \inferrule*{
      \inferrule*{
        \Gamma \proves \varphi
      }{
        \Gamma \proves \varphi : \Prop
      }
    }{
      \ldots
    }
  }{\Gamma \proves \set{y : \Prop}{\varphi} : \Type} \\
  \inferrule*{
    \ldots 
  }{
    \Gamma \proves \bot : \Prop
  } \\
  \Gamma \proves \varphi \\
  y \notin FV(\Gamma, \varphi)
}{
  \Gamma \proves \bot : \set{y : \Prop}{\varphi}
} (b)
\]
\[
\inferrule*{
  \Gamma \proves \forall x : \set{y : \Prop}{\varphi} \,.\, \psi \\ (b)
}{
  \Gamma \proves \psi[x / \bot] \;\; (\equiv \psi)
}
\]

Finally, the rule
\[
\supset_t:\;
\inferrule{
  \Gamma \proves \varphi : \Prop \\ \Gamma, \varphi \proves \psi : \Prop
}{
  \Gamma \proves (\varphi \supset \psi) : \Prop
}
\]
follows by
\[
\inferrule*{
  \inferrule*{
  }{\Gamma \proves \Prop : \Type} \\
  \inferrule*{
    \Gamma \proves \varphi : \Prop \\
    y \notin FV(\Gamma, \varphi)
  }{
    \Gamma, y : \Prop \proves \varphi : \Prop
  }
}{
  \Gamma \proves \set{y : \Prop}{\varphi} : \Type
} (c)
\]
\[
\inferrule*{
  (c) \\
  \inferrule*{
    \Gamma, \varphi \proves \psi : \Prop
  }{
    \Gamma, x : \set{y : \Prop}{\varphi} \proves \psi : \Prop
  }\;\raisebox{0.7em}{($\star$)}
}{
  \Gamma \proves (\forall x : \set{y : \Prop}{\varphi} \,.\, \psi) : \Prop
}
\]
where we use Lemma~\ref{lem_01} to perform the inference~($\star$).

\newpage
\renewcommand{\thesection}{Appendix~\Alph{section}}
\tocless\section{Proofs for Section~\ref{sec_induction}}\label{sec_proofs_induction}
\renewcommand{\thesection}{\Alph{section}}
\addcontentsline{toc}{section}{\thesection\hspace{0.5em} Proofs for
  Section~\ref{sec_induction}}

\begin{lemma}\label{lem_constr_type}
  If $c_i^\iota \in \Cc$ of arity~$n_i$ has signature $\langle
  \iota_1,\ldots,\iota_{n_i} \rangle$ then $\proves_{\I_s} c_i^\iota :
  \iota_1^* \to \ldots \to \iota_{n_i}^* \to \iota$.
\end{lemma}

\begin{proof}
  Follows directly from rules~$i_t^{\iota,i}$ and~$\to_i$, and from axiom~4
  ($\Gamma \proves \iota : \Type$ for $\iota \in \Tc_I$).
\end{proof}

\begin{lemma}\label{lem_destr_test_type}
  $\proves_{\I_s} o_i^\iota : \iota \to \Prop$ and $\proves_{\I_s}
  \forall x : \iota \,.\, o_i^\iota x \supset (d_{i,j}^\iota x :
  \iota_{i,j}^*)$
\end{lemma}

\begin{proof}
  We first show $\forall x : \iota \,.\, o_i^\iota x : \Prop$ by
  induction. Once we have this, to derive $\proves_{\I_s} o_i^\iota :
  \iota \to \Prop$ it suffices to apply~$\forall_e$
  and~$\to_i$. Because $\proves o_i^\iota (c_i^\iota x_1 \ldots
  x_{n_i})$ by axiom~5 and $\proves \neg (o_i^\iota (c_j^\iota x_1
  \ldots x_{n_j}))$ for $i \ne j$ by axiom~6, we have $\proves
  (o_i^\iota (c_i^\iota x_1 \ldots x_{n_i})) : \Prop$ using
  rule~$p_i$, and $\proves (o_i^\iota (c_j^\iota x_1 \ldots x_{n_j}))
  : \Prop$ for $i \ne j$, using rules~$p_i$ and~$\supset_{t2}$. It is
  then easy to see that the premises of the rule~$i_i^\iota$ are
  derivable for $t \equiv \lambda x \,.\, o_i^\iota x : \Prop$.

  We show the second claim also by induction. Let
  \(
  \Gamma_i \equiv x_1 : \iota_{i,1}^*, \ldots, x_{n_i} :
  \iota_{i,n_i}^*.
  \)
  Let $1 \le k \le m$, where~$m$ is the number of constructors
  of~$\iota$. If $i \ne k$ then $\Gamma_k \proves \neg (o_i^\iota
  (c_k^\iota x_1 \ldots x_{n_k}))$. Hence, $\Gamma_k, o_i^\iota
  (c_k^\iota x_1 \ldots x_{n_k}) \proves d_{i,j}^\iota (c_k^\iota x_1
  \ldots x_{n_k}) : \iota_{i,j}^*$ by applying rules~$\supset_e$
  and~$\bot_e$. Since $\proves o_i^\iota : \iota \to \Prop$ has
  already been proven in the previous paragraph, and $\Gamma_k \proves
  (c_k^\iota x_1 \ldots x_{n_k}) : \iota$ by
  Lemma~\ref{lem_constr_type}, we obtain $\Gamma_j \proves (o_i^\iota
  (c_j^\iota x_1 \ldots x_{n_j})) : \Prop$. Thus $\Gamma_k \proves
  o_i^\iota (c_k^\iota x_1 \ldots x_{n_k}) \supset (d_{i,j}^\iota
  (c_k^\iota x_1 \ldots x_{n_k}) : \iota_{i,j}^*)$ by
  rule~$\supset_i$. If $i = k$ then $d_{i,j} (c_k^\iota x_1 \ldots
  x_{n_k}) = x_j$ and thus $\Gamma_k, o_i^\iota (c_k^\iota x_1 \ldots
  x_{n_k}) \proves (d_{i,j} (c_k^\iota x_1 \ldots x_{n_k})) :
  \iota_{i,j}^*$ by axiom~1 and rule~eq. We then have $\proves
  o_i^\iota (c_k^\iota x_1 \ldots x_{n_k})$, hence $\proves (o_i^\iota
  (c_k^\iota x_1 \ldots x_{n_k})) : \Prop$ by~$p_i$, and thus we may
  use~$\supset_i$ to derive $\Gamma_k \proves o_i^\iota (c_k^\iota x_1
  \ldots x_{n_k}) \supset (d_{i,j} (c_k^\iota x_1 \ldots x_{n_k})) :
  \iota_{i,j}^*$. Therefore, by applying weakening we derive the
  premises of rule~$i_i^\iota$. Hence, we obtain our thesis
  by~$i_i^\iota$.
\end{proof}

\begin{lemma}\label{lem_ind_eq}
  If $\iota \in \Tc_I$ then $\proves_{\I_s} \forall x, y : \iota \,.\,
  x =_\iota y \supset x = y$.
\end{lemma}

\begin{proof}
  We prove this lemma by induction in the meta-theory on the structure
  of~$\iota \in \Tc_I$.

  Let $\varphi(x) \equiv \forall y : \iota \,.\, x =_\iota y \supset x
  = y$ and
  \[
  \Gamma_i^x \equiv x_1 : \iota_{i,1}^*, \ldots, x_{n_i} :
  \iota_{i,n_i}^*, \varphi(x_{j_{i,1}}), \ldots, \varphi(x_{j_{i,k_i}})
  \]
  where $j_{i,1},\ldots,j_{i,k_i}$ are all indices~$j$ such that
  $\iota_{i,j} = \star$, as in the conditions for rule~$i_i^\iota$. It
  suffices to show $\Gamma_i^x \proves \varphi(c_i^\iota x_1 \ldots
  x_{n_i})$ for $i=1,\ldots,m$. Then $\proves \forall x, y : \iota
  \,.\, x =_\iota y \supset x = y$ is obtained by applying
  rule~$i_i^\iota$.

  For $i=1,\ldots,m$, we prove $\Gamma_i^x \proves \varphi(c_i^\iota
  x_1 \ldots x_{n_i})$, i.e.
  \[
  \Gamma_i^x \proves \forall y : \iota
  \,.\, (c_i^\iota x_1 \ldots x_{n_i}) =_\iota y \supset (c_i^\iota
  x_1 \ldots x_{n_i}) = y
  \]
  by induction on~$y$ (in the theory, i.e. applying
  rule~$i_i^\iota$). Let
  \[
  \Gamma_{i,j}^y \equiv y_1 : \iota_{i,1}^*, \ldots, y_{n_i} :
  \iota_{i,n_i}^*, \psi_i(y_{j_{i,1}}), \ldots, \psi_i(y_{j_{i,k_i}})
  \]
  where
  \[
  \psi_i(y) \equiv (c_i^\iota x_1 \ldots x_{n_i}) =_\iota y \supset
  (c_i^\iota x_1 \ldots x_{n_i}) = y.
  \]
  It therefore suffices to show
  \[
  (\star)\;\;\;\;\;\Gamma_i^x, \Gamma_{i,j}^y \proves (c_i^\iota x_1 \ldots x_{n_i})
  =_\iota (c_j^\iota y_1 \ldots y_{n_j}) \supset (c_i^\iota x_1 \ldots
  x_{n_i}) = (c_j^\iota y_1 \ldots y_{n_j})
  \]
  for all $i=1,\ldots,m$, $j=1,\ldots,m$.

  Thus assume $1 \le i \le m$ and $1 \le j \le m$. By Lemma~\ref{lem_constr_type} we have
  \(
  \Gamma_i^x \proves (c_i^\iota x_1 \ldots x_{n_i}) : \iota
  \)
  and
  \(
  \Gamma_i^x, \Gamma_{i,j}^y \proves (c_j^\iota y_1 \ldots y_{n_j}) : \iota.
  \)
  Thus by Lemma~\ref{lem_leibniz} we obtain
  \[
  \begin{array}{lr}
  (\star\star)\;\;\;\;\; & \Gamma_i^x, \Gamma_{i,j}^y \proves ((c_i^\iota x_1 \ldots
    x_{n_i}) =_\iota (c_j^\iota y_1 \ldots y_{n_j})) : \Prop.
  \end{array}
  \]
  For $k=1,\ldots,n_i$ take
  \[
  f_k = \lambda x \,.\, \Cond{(o_i^\iota x)}{(d_{i,k}^\iota x
    =_{\iota_{i,k}^*} x_k)}{\bot}.
  \]

  Assume $1 \le k \le n_i$. Since $o_i^\iota : \iota \to \Prop$ by
  Lemma~\ref{lem_destr_test_type}, we obtain $x : \iota \proves
  (o_i^\iota x) \lor \neg (o_i^\iota x)$ by~$\to_e$ and the law of
  excluded middle. It is easy to derive $x : \iota, \neg (o_i^\iota x)
  \proves f_k x : \Prop$ using the rules for~\mbox{Eq}
  and~\mbox{Cond}. By Lemma~\ref{lem_destr_test_type} we also have $x
  : \iota, o_i^\iota x \proves d_{i,k}^\iota x : \iota_{i,k}^*$. We
  have $\Gamma_i^x \proves x_k : \iota_{i,k}^*$ by definition
  of~$\Gamma_i^x$. Then $\Gamma_i^x, x : \iota, o_i^\iota x \proves
  (d_{i,k}^\iota x =_{\iota_{i,k}^*} x_k) : \Prop$ follows from
  Lemma~\ref{lem_leibniz}. Using rule~$\vee_e$ it is now easy to
  derive $\Gamma_i^x, x : \iota \proves f_k x : \Prop$. Since $\proves
  \iota : \Type$, we obtain $\Gamma_i^x \proves f_k : \iota \to \Prop$
  by rule~$\to_i$.

  Suppose $i = j$. Assume $1 \le k \le n_i$. Because $\Gamma_i^x
  \proves f_k : \iota \to \Prop$, by definition of~$=_\iota$ we have
  \[
  \Gamma_i^x,\Gamma_{i,i}^y,(c_i^\iota x_1 \ldots x_{n_i})
  =_\iota (c_i^\iota y_1 \ldots y_{n_i}) \proves (f_k (c_i^\iota x_1
  \ldots x_{n_i})) \supset (f_k (c_i^\iota y_1 \ldots y_{n_i}))
  \]
  We have
  \(
  \Gamma_i^x \proves (f_k (c_i^\iota x_1
  \ldots x_{n_i})) = (x_k =_\iota x_k)
  \).
  Hence \( \Gamma_i^x \proves f_k (c_i^\iota x_1 \ldots x_{n_i})
  \). Because
  \[
  \Gamma_i^x, \Gamma_{i,i}^y \proves (f_k (c_i^\iota y_1 \ldots
  y_{n_i})) = (y_k =_{\iota_{i,k}^*} x_k)
  \]
  we thus obtain
  \[
  \Gamma_i^x,\Gamma_{i,i}^y,(c_i^\iota x_1 \ldots x_{n_i}) =_\iota
  (c_i^\iota y_1 \ldots y_{n_i}) \proves y_k =_{\iota_{i,k}^*} x_k.
  \]
  If $\iota_{i,k}^* = \iota$ then $\Gamma_i^x \proves \forall y :
  \iota \,.\, x_k =_\iota y \supset x_k = y$ and $\Gamma_i^x,
  \Gamma_{i,i}^y \proves y_k : \iota$. Hence
  $\Gamma_i^x,\Gamma_{i,i}^y \proves x_k = y_k$. If $\iota_{i,k}^* =
  \iota_{i,k}$ then by the inductive hypothesis in the meta-theory we
  obtain $\proves \forall x, y : \iota_{i,k} \,.\, x =_{\iota_{i,k}} y
  \supset x = y$, which again implies $\Gamma_i^x,\Gamma_{i,i}^y
  \proves x_k = y_k$, because $\Gamma_i^x \proves x_k : \iota_{i,k}$
  and $\Gamma_{i,i}^y \proves y_k : \iota_{i,k}$. Since $1 \le k \le
  n_i$ was arbitrary, we obtain
  \[
  \Gamma_i^x,\Gamma_{i,i}^y,(c_i^\iota x_1 \ldots x_{n_i}) =_\iota
  (c_i^\iota y_1 \ldots y_{n_i}) \proves (c_i^\iota x_1 \ldots x_{n_i}) =
  (c_i^\iota y_1 \ldots y_{n_i}).
  \]
  By $(\star\star)$ and~$\supset_i$ we obtain $(\star)$ for $i=j$.

  Suppose $i \ne j$. Assume $1 \le k \le n_i$. Because $\Gamma_i^x
  \proves f_k : \iota \to \Prop$, by definition of~$=_\iota$ we have
  \[
  \Gamma_i^x,\Gamma_{i,j}^y,(c_i^\iota x_1 \ldots x_{n_i})
  =_\iota (c_j^\iota y_1 \ldots y_{n_i}) \proves (f_k (c_i^\iota x_1
  \ldots x_{n_i})) \supset (f_k (c_j^\iota y_1 \ldots y_{n_j}))
  \]
  As in the case $i = j$ we have $\Gamma_i^x \proves (f_k (c_i^\iota x_1
  \ldots x_{n_i}))$. Because $i \ne j$ we have
  \[
  \Gamma_i^x,\Gamma_{i,j}^\iota \proves (f_k (c_j^\iota y_1 \ldots
  y_{n_j})) = \bot.
  \]
  Therefore
  \[
  \Gamma_i^x,\Gamma_{i,j}^y,(c_i^\iota x_1 \ldots x_{n_i})
  =_\iota (c_j^\iota y_1 \ldots y_{n_i}) \proves \bot.
  \]
  By rule~$\bot_e$ we conclude
  \[
  \Gamma_i^x,\Gamma_{i,j}^y,(c_i^\iota x_1 \ldots x_{n_i})
  =_\iota (c_j^\iota y_1 \ldots y_{n_i}) \proves (c_i^\iota x_1 \ldots x_{n_i}) =
  (c_j^\iota y_1 \ldots y_{n_j})
  \]
  By $(\star\star)$ and~$\supset_i$ we obtain $(\star)$ for $i \ne
  j$. This finishes the proof.
\end{proof}

\begin{lemma}\label{lem_nat_peano_3}
  $\proves_{\I_s} \forall x : \Nat \,.\, \left((\iszero x) \vee \exists y
  : \Nat \,.\, x = (\s y)\right)$.
\end{lemma}

\begin{proof}
  Recall that $x = \s y$ stands for $\Eq{x}{(\s y)}$. Let $\varphi(x)
  \equiv (\iszero x) \vee \exists y : \Nat \,.\, x = \s y$. We have the
  following derivation.
  \[
  \inferrule*{
    \inferrule*{
      \proves \iszero 0
    }{
      \proves \varphi(0)
    } \\
    \inferrule*{
      \inferrule*{
        x : \Nat, \varphi(x) \proves \s x = \s x
        \\ x : \Nat, \varphi(x) \proves x : \Nat
      }{
        x : \Nat, \varphi(x) \proves \exists y : \Nat
        \,.\, \s x = \s y
      }
    }{
      x : \Nat, \varphi(x) \proves \varphi(\s x)
    }
  }{
    \proves \forall x : \Nat \,.\, \left((\iszero x) \vee \exists y : \Nat
    \,.\, x = \s y\right)
  }\;\raisebox{0.7em}{by $n_i$}
  \]
\end{proof}

\begin{lemma}\label{lem_iszero_eq}
  $\proves_{\I_s} \forall x : \Nat \,.\, (\iszero x) \supset (x = 0)$.
\end{lemma}

\begin{proof}
  Let $\varphi(x) \equiv (\iszero x) \supset (x = 0)$.
  \[
  \inferrule*{
    \inferrule*{
      \proves \iszero 0
    }{
      \proves (\iszero 0) : \Prop
    } \\
    \iszero 0 \proves 0 = 0
  }{
    \proves \varphi(0)
  }(a)
  \]
  \[
  \inferrule*{
    \inferrule*{
      \inferrule*{
        x : \Nat, \varphi(x) \proves \neg (\iszero (\s x))
      }{
        x : \Nat, \varphi(x) \proves \neg (\iszero (\s x)) : \Prop
      }
    }{
      x : \Nat, \varphi(x) \proves \iszero (\s x) : \Prop
    }
    \\
    \inferrule*{
      x : \Nat, \varphi(x), \iszero (\s x) \proves \bot
    }{
      x : \Nat, \varphi(x), \iszero (\s x) \proves \s x = 0
    }
  }{
    x : \Nat, \varphi(x) \proves \varphi(\s x)
  } (b)
  \]
  \[
  \inferrule*{
    (a)
    \\
    (b)
  }{
    \proves \forall x : \Nat \,.\, (\iszero x) \supset (x = 0)
  }
  \]
\end{proof}

\begin{lemma}
  The following terms are derivable in the system~$\I_s$:
  \begin{itemize}
  \item $\forall x, y : \Nat \,.\, (x + y) : \Nat$,
  \item $\forall x, y : \Nat \,.\, (x - y) : \Nat$,
  \item $\forall x, y : \Nat \,.\, (x \mult y) : \Nat$,
  \item $\forall x, y : \Nat \,.\, (x \le y) : \Prop$,
  \item $\forall x, y : \Nat \,.\, (x < y) : \Prop$.
  \end{itemize}
\end{lemma}

\begin{proof}
  We give a sample proof for the first term.
  \[
  \inferrule*{
    \inferrule*{
      x : \Nat \proves (x + 0) = ((\lambda x \,.\, x) x)
    }{
      x : \Nat \proves (x + 0) = x
    } \\ x : \Nat \proves x : \Nat
  }{
    x : \Nat \proves (x + 0) : \Nat
  } \; (a)
  \]
  \[
    \inferrule*{
      x : \Nat, y : \Nat, (x + y) : \Nat \proves (x + y) : \Nat
    }{
      x : \Nat, y : \Nat, (x + y) : \Nat \proves (\s (x + y)) : \Nat
    } \; (b)
  \]
  \[
  \inferrule*{
    x : \Nat, y : \Nat, (x + y) : \Nat \proves x + (\s y) = \s
      (x + y)
    \\
    (b)
  }{
    x : \Nat, y : \Nat, (x + y) : \Nat \proves (x + (\s y)) : \Nat
  } \; (c)
  \]
  \[
  \inferrule*{
    \inferrule*{
      (a)
      \\ (c)
    }{
      x : \Nat \proves \forall y : \Nat \,.\, (x + y) : \Nat
    } \\ \proves \Nat : \Type
  }{
    \proves \forall x, y : \Nat \,.\, (x + y) : \Nat
  }
  \]
  The second and third terms are derived in a similar way. The fourth
  and fifth terms are derived from the second using
  Lemma~\ref{lem_constr_type}.
\end{proof}

\begin{lemma}\label{lem_s_minus}
  $\proves_{\I_s} \forall x, y : \Nat \,.\, \left(\s x - \s y = x -
  y\right)$
\end{lemma}

\begin{proof}
  We assume $x : \Nat$ and show $\forall y : \Nat \,.\, \left(\s x -
  \s y = x - y\right)$ by induction on~$y$. For $y = 0$ we have $\s x
  - \s 0 = \p (\s x - 0) = \p (\s x) = x = x - 0$, and under the
  assumptions $x :\Nat, y : \Nat, \left(\s x - \s y\right) = x - y$ we
  may derive $\s x - \s (\s y) = \p (\s x - \s y) = \p (x - y) = x -
  \s y$. By rule~$n_i$ we obtain the thesis.
\end{proof}

\begin{lemma}\label{lem_le_ge_then_eq}
  $\proves_{\I_s} \forall x, y : \Nat \,.\, (x \ge y) \wedge (x \le y)
  \supset (x =_\Nat y)$.
\end{lemma}

\begin{proof}
  Let $\varphi(x) \equiv \forall y : \Nat \,.\, (x \ge y) \wedge (x
  \le y) \supset (x =_\Nat y)$. We proceed by induction on~$x$.

  In the base step we need to show~$\varphi(0)$. We assume $y : \Nat$
  and $(0 \ge y) \wedge (0 \le y)$. From $y \le 0$ we have $\iszero (y
  - 0)$, i.e. $\iszero y$, and thus $y = 0$ by
  Lemma~\ref{lem_iszero_eq}. By rule~$p_i$ we have $((0 \ge y) \wedge
  (0 \le y)) : \Prop$. Hence, we may use implication introduction and
  then universal quantifier introduction to obtain~$\varphi(0)$.

  In the inductive step we need to prove~$\varphi(\s x)$ under the
  assumptions $x : \Nat$ and $\varphi(x)$. We assume further $y :
  \Nat$ and $(\s x \ge y) \wedge (\s x \le y)$. By
  Lemma~\ref{lem_nat_peano_3} there are two possibilities: $y = 0$ or
  $\exists z : \Nat \,.\, y = \s z$. If $y = 0$ then we easily obtain
  $\s x = 0$, which leads to a contradiction, from which we may derive
  $x = y$. If $y = \s z$ then we have $\s x - \s z = 0$ and $\s z - \s
  x = 0$. By Lemma~\ref{lem_s_minus} we obtain $x - z = 0$ and $z - x
  = 0$. From the inductive hypothesis we have $x =_\Nat z$, and hence
  $x = z$ by Lemma~\ref{lem_ind_eq}. Thus $\s x = \s z = y$. Since $x
  : \Nat$ and $y : \Nat$, it is not difficult to show that $((\s x \ge
  y) \wedge (\s x \le y)) : \Prop$. We may therefore use implication
  introduction and then universal quantifier introduction to
  obtain~$\varphi(\s x)$.
\end{proof}

\begin{theorem}
  The following terms are derivable in the system~$\I_s$:
  \begin{itemize}
  \item $\forall x, y : \Nat \,.\, \left(\s x =_\Nat \s y\right)
    \supset \left(x =_\Nat y\right)$,
  \item $\forall x : \Nat \,.\, \neg (\s x =_\Nat 0)$,
  \item $\forall x : \Nat \,.\, (x + 0 =_\Nat x)$,
  \item $\forall x, y : \Nat \,.\, (x + \s y =_\Nat \s (x + y))$,
  \item $\forall x : \Nat \,.\, (x \mult 0 =_\Nat 0)$,
  \item $\forall x, y : \Nat \,.\, (x \mult (\s y) =_\Nat (x \mult y)
    + x)$.
  \end{itemize}
\end{theorem}

\begin{proof}
  All proofs are easy. As an example we give an indication of how the
  first term may be derived. We assume $x : \Nat$, $y : \Nat$ and $\s
  x =_\Nat \s y$. From $\s x =_\Nat \s y$ we obtain $\s x = \s y$ by
  Lemma~\ref{lem_ind_eq}. We thus obtain $x = \p (\s x) = \p (\s y) =
  y$. Since $x : \Nat$ and $y : \Nat$ we have $x =_\Nat y$ by
  Lemma~\ref{lem_leibniz}. We also have $(\s x =_\Nat \s y) : \Prop$
  by Lemma~\ref{lem_constr_type} and
  Lemma~\ref{lem_nat_op_well_defined}. Therefore, we may use
  implication introduction, and then universal quantifier introduction
  twice, to show $\forall x, y : \Nat \,.\, \left(\s x =_\Nat \s
  y\right) \supset \left(x =_\Nat y\right)$.
\end{proof}

\begin{lemma}\label{lem_nat_less}
  $\Gamma \proves_{\I_s} \forall x, y : \Nat \,.\, x < \s y \supset (x
  < y \lor x = y)$
\end{lemma}

\begin{proof}
  Easy induction on~$x$, recalling that $x < y \equiv \iszero (x -
  y)$, using Lemma~\ref{lem_nat_peano_3} for the basis, and
  Lemma~\ref{lem_s_minus} in the inductive step.
\end{proof}

\begin{lemma}\label{lem_nat_ind_2}
  The following rule is admissible in~$\I_s$.
  \[
  {n_i':}\;\;\;\inferrule{
    \Gamma \proves t 0 \\ \Gamma, y : \Nat, \forall x : \Nat \,.\, x <
    y \supset t x \proves t y \\ x,y \notin FV(\Gamma, t)
  }{
    \Gamma \proves \forall x : \Nat \,.\, t x
  }
  \]
\end{lemma}

\begin{proof}
  We show that if the premises are derivable, then so is $\Gamma
  \proves \forall y, x : \Nat \,.\, x < y \supset t x$. We achieve
  this by deriving the premises of the rule~$n_i$ for $\psi(y) \equiv
  \forall x : \Nat \,.\, x < y \supset t x$. Since $x : \Nat, x < 0
  \proves \bot$ and $x : \Nat \proves (x < 0) : \Prop$, we have
  $\Gamma \proves \psi(0)$. By our assumption that the second premise
  of~$n_i'$ is derivable, we have $\Gamma, y : \Nat, \psi(y) \proves t
  y$. Let $\Gamma_1 \equiv \Gamma, y : \Nat, \psi(y), x : \Nat, x < \s
  y$. We show $\Gamma_1 \proves t x$. By Lemma~\ref{lem_nat_less} we
  obtain $\Gamma_1 \proves x < y \lor x = y$. Since $\Gamma_1 \proves
  \forall x : \Nat \,.\, x < y \supset t x$ we have $\Gamma_1, x < y
  \proves t x$. Since $\Gamma_1 \proves t y$ we have $\Gamma_1, x = y
  \proves t x$. Using rule~$\lor_e$ we obtain $\Gamma_1 \proves t
  x$. Since $\Gamma, y : \Nat, \psi(y), x : \Nat \proves (x < \s y) :
  \Prop$ we have $\Gamma, y : \Nat, \psi(y) \proves \forall x : \Nat
  \,.\, x < \s y \supset t x$, i.e. $\Gamma, y : \Nat, \psi(y) \proves
  \psi(\s y)$. By applying rule~$n_i$ we thus obtain $\Gamma \proves
  \forall y,x : \Nat \,.\, x < y \supset t x$. Now $\Gamma, x : \Nat
  \proves \s x : \Nat$ and $\Gamma, x : \Nat \proves x < \s x$, hence
  $\Gamma, x : \Nat \proves t x$. Therefore $\Gamma \proves \forall x
  : \Nat \,.\, t x$.
\end{proof}

\begin{theorem}
  Suppose $\Gamma \proves \forall x_1 : \alpha_1 \ldots \forall x_n :
  \alpha_n \,.\, \varphi_1 \lor \ldots \lor \varphi_m$, $\Gamma
  \proves \alpha_j : \Type$ for $j = 1,\ldots,n$, and for
  $i=1,\ldots,m$: $\Gamma \proves \forall x_1 : \alpha_1 \ldots
  \forall x_n : \alpha_n \,.\, t_i : \beta \to \ldots \to \beta$
  where~$\beta$ occurs~$k_i + 1$ times, $\Gamma \proves \forall x_1 :
  \alpha_1 \ldots \forall x_n : \alpha_n \,.\, t_{i,j,k} : \alpha_k$
  for $j=1,\ldots,k_i$, $k=1,\ldots,n$, $x_1,\ldots,x_n \notin
  FV(f,\alpha_1,\ldots,\alpha_n,\beta)$ and
  \begin{eqnarray*}
  \Gamma \proves \forall x_1 : \alpha_1 \ldots \forall x_n :
  \alpha_n \,.\, \varphi_i &\supset& (f x_1 \ldots x_n = \\ && \;t_i (f t_{i,1,1}
  \ldots t_{i,1,n}) \ldots (f t_{i,k_i,1} \ldots t_{i,k_i,n})).
  \end{eqnarray*}
  If there is a term~$g$ such that $\Gamma \proves g : \alpha_1 \to
  \ldots \to \alpha_n \to \Nat$ and for $i=1,\ldots,m$
  \begin{eqnarray*}
  \Gamma \proves \forall x_1 : \alpha_1 \ldots \forall x_n : \alpha_n
  \,.\, \varphi_i &\supset& (\left((f x_1 \ldots x_n) :
  \beta\right) \lor \\ && \;((g t_{i,1,1} \ldots t_{i,1,n}) < (g x_1 \ldots
  x_n) \land \ldots \land \\ && \;\;(g t_{i,k_i,1} \ldots t_{i,k_i,n}) < (g x_1
  \ldots x_n)))
  \end{eqnarray*}
  where $x_1,\ldots,x_n \notin FV(g)$, then
  \[
  \Gamma \proves f : \alpha_1 \to \ldots \to \alpha_n \to \beta.
  \]
\end{theorem}

\begin{proof}
  We derive the premises of the rule~$n_i'$ for
  \[
  \psi(x) \equiv \forall x_1 : \alpha_1 \ldots \forall x_n : \alpha_n
  \,.\, x =_\Nat (g x_1 \ldots x_n) \supset f x_1 \ldots x_n : \beta.
  \]
  We have $\Gamma, x_1 : \alpha_1,\ldots,x_n : \alpha_n \proves
  \varphi_1 \lor \ldots \lor \varphi_m$. Let $\Gamma_i \equiv \Gamma,
  x_1 : \alpha_1,\ldots,x_n : \alpha_n,\varphi_i$ for
  $i=1,\ldots,m$. We first prove $\Gamma_i, 0 =_\Nat (g x_1 \ldots
  x_n) \proves f x_1 \ldots x_n : \beta$. It suffices to show
  \begin{eqnarray*}
  \Gamma_i, 0 =_\Nat (g x_1 \ldots x_n), && ((g t_{i,1,1} \ldots
  t_{i,1,n}) < (g x_1 \ldots x_n) \land \ldots \land \\ && \;(g
  t_{i,k_i,1} \ldots t_{i,k_i,n}) < (g x_1 \ldots x_n)) \proves f x_1
  \ldots x_n : \beta.
  \end{eqnarray*}
  Since $\proves g : \alpha_1 \to \ldots \to \alpha_n \to \Nat$, by
  Lemma~\ref{lem_ind_eq} and $\proves \neg (\iszero (\s (g
  t_{i,1,1} \ldots t_{i,1,n})))$ we have
  \begin{eqnarray*}
  \Gamma_i, 0 =_\Nat (g x_1 \ldots x_n), && ((g t_{i,1,1} \ldots
  t_{i,1,n}) < (g x_1 \ldots x_n) \land \ldots \land \\ && \;(g t_{i,k_i,1}
  \ldots t_{i,k_i,n}) < (g x_1 \ldots x_n)) \proves \bot.
  \end{eqnarray*}
  so our claim follows by~$\bot_e$. Therefore $\Gamma_i \proves 0
  =_\Nat (g x_1 \ldots x_n) \supset (f x_1 \ldots x_n : \beta)$ by
  $\proves g : \alpha_1 \to \ldots \to \alpha_n \to \Nat$,
  Lemma~\ref{lem_nat_op_well_defined} and~$\supset_i$. Since
  \[
  \Gamma, x_1 : \alpha_1, \ldots, x_n : \alpha_n \proves \varphi_1
  \lor \ldots \lor \varphi_m
  \]
  by $m-1$ applications of~$\lor_e$ we obtain
  \[
  \Gamma, x_1 : \alpha_1, \ldots, x_n : \alpha_n \proves 0 =_\Nat (g
  x_1 \ldots x_n) \supset (f x_1 \ldots x_n : \beta).
  \]
  Because $\Gamma \proves \alpha_i : \Type$ for $i=1,\ldots,n$, we may
  apply $\forall$-introduction~$n$ times to obtain $\Gamma \proves
  \psi(0)$.

  Let
  \[
  \Gamma' \equiv \Gamma, y : \Nat, \forall x : \Nat \,.\, x < y
  \supset \psi(x)
  \]
  and
  \[
  \Gamma_i' \equiv \Gamma', x_1 : \alpha_1, \ldots, x_n : \alpha_n,
  \varphi_i.
  \]
  Because $\Gamma', x_1 : \alpha, \ldots, x_n : \alpha_n \proves
  \varphi_1 \lor \ldots \lor \varphi_m$ and $\Gamma' \proves \alpha_i
  : \Type$ for $i=1,\ldots,n$, to derive the second premise of~$n_i'$
  for~$\psi$, it suffices to show
  \[
  \Gamma_i' \proves y =_\Nat (g x_1 \ldots x_n) \supset f x_1 \ldots
  x_n : \beta
  \]
  for $i=1,\ldots,m$, apply $\lor$-elimination~$m-1$ times, and then
  $\forall$-introduction~$n$ times. Let $\Gamma_i'' \equiv \Gamma_i',
  y =_\Nat (g x_1 \ldots x_n)$. Since
  \[
  \Gamma_i' \proves (y =_\Nat (g x_1 \ldots x_n)) : \Prop
  \]
  it actually suffices to prove \( \Gamma_i'' \proves f x_1 \ldots x_n
  : \beta \) and apply~$\supset_i$. By the assumption on~$g$ in the
  lemma, we have
  \begin{eqnarray*}
  \Gamma_i'' \proves \left((f x_1 \ldots x_n) : \beta\right) &\lor& ((g
  t_{i,1,1} \ldots t_{i,1,n}) < (g x_1 \ldots x_n) \land \ldots \land
  \\ && \;(g t_{i,k_i,1} \ldots t_{i,k_i,n}) < (g x_1 \ldots x_n))
  \end{eqnarray*}
  Let
  \[
  \Gamma_i^+ \equiv \Gamma_i'', (g t_{i,1,1} \ldots t_{i,1,n}) < (g
  x_1 \ldots x_n) \land \ldots \land (g t_{i,k_i,1} \ldots
  t_{i,k_i,n}) < (g x_1 \ldots x_n)
  \]
  It thus suffices to show $\Gamma_i^+ \proves f x_1 \ldots x_n :
  \beta$ and apply rule~$\lor_e$. Let $1 \le j \le k_i$. Since
  $\Gamma_i^+ \proves t_{i,j,k} : \alpha_k$ for $k=1,\ldots,n$ and
  $\Gamma_i^+ \proves g : \alpha_1 \to \ldots \to \alpha_n \to \Nat$,
  we have $\Gamma_i^+ \proves (g t_{i,j,1} \ldots t_{i,j,n}) :
  \Nat$. Thus
  \[
  \Gamma_i^+ \proves (g t_{i,j,1} \ldots t_{i,j,n}) < y \supset \psi(g
  t_{i,j,1} \ldots t_{i,j,n}).
  \]
  But $\Gamma_i^+ \proves y =_\Nat (g x_1 \ldots x_n)$, $\Gamma_i^+
  \proves y : \Nat$ and $\Gamma_i^+ \proves (g x_1 \ldots x_n) :
  \Nat$, so by Lemma~\ref{lem_ind_eq} we obtain $\Gamma_i^+ \proves y
  = (g x_1 \ldots x_n)$ and thus
  \[
  \Gamma_i^+ \proves (g t_{i,j,1} \ldots t_{i,j,n}) < (g x_1 \ldots
  x_n) \supset \psi(g t_{i,j,1} \ldots t_{i,j,n}).
  \]
  By the rule~$\supset_e$ we have $\Gamma_i^+ \proves \psi(g t_{i,j,1}
  \ldots t_{i,j,n})$, i.e.
  \[
  \Gamma_i^+ \proves \forall x_1 : \alpha_1 \ldots \forall x_n : \alpha_n \,.\, (g t_{i,j,1}
  \ldots t_{i,j,n}) =_\Nat (g x_1 \ldots x_n) \supset f x_1 \ldots x_n
  : \beta.
  \]
  Thus
  \(
  \Gamma_i^+ \proves (f t_{i,j,1} \ldots t_{i,j,n}) : \beta
  \)
  for $j=1,\ldots,k_i$. Because $\Gamma_i^+ \proves t_i : \beta \to
  \ldots \to \beta$ where~$\beta$ occurs~$k_i+1$ times, we
  have
  \[
  \Gamma_i^+ \proves (t_i (f t_{i,1,1} \ldots
  t_{i,1,n}) \ldots (f t_{i,k_i,1} \ldots t_{i,k_i,n})) : \beta.
  \]
  But
  \[
  \Gamma_i^+ \proves f x_1 \ldots x_n = t_i (f t_{i,1,1} \ldots
  t_{i,1,n}) \ldots (f t_{i,k_i,1} \ldots t_{i,k_i,n})
  \]
  so $\Gamma_i^+ \proves (f x_1 \ldots x_n) : \beta$.

  Therefore, we have derived the premises of~$n_i'$ for~$\psi$, hence
  by Lemma~\ref{lem_nat_ind_2} we obtain $\Gamma \proves \forall x :
  \Nat \,.\, \psi(x)$, i.e.
  \[
  \Gamma \proves \forall x : \Nat \forall x_1 : \alpha_1 \ldots
  \forall x_n : \alpha_n \,.\, x =_\Nat (g x_1 \ldots x_n) \supset f
  x_1 \ldots x_n : \beta.
  \]
  Hence
  \[
  \Gamma, x_1 : \alpha_1, \ldots, x_n : \alpha_n, x : \Nat \proves x
  =_\Nat (g x_1 \ldots x_n) \supset f x_1 \ldots x_n : \beta
  \]
  so
  \[
  \Gamma, x_1 : \alpha_1, \ldots, x_n : \alpha_n \proves \forall x :
  \Nat \,.\, x =_\Nat (g x_1 \ldots x_n) \supset f x_1 \ldots x_n :
  \beta.
  \]
  Because $\Gamma, x_1 : \alpha_1, \ldots, x_n : \alpha_n \proves (g
  x_1 \ldots x_n) : \Nat$, we easily obtain
  \[
  \Gamma, x_1 : \alpha_1, \ldots, x_n : \alpha_n \proves f x_1 \ldots
  x_n : \beta
  \]
  and thus $\Gamma \proves f : \alpha_1 \to \ldots \to \alpha_n \to
  \beta$.
\end{proof}

\newpage
\renewcommand{\thesection}{Appendix~\Alph{section}}
\tocless\section{Soundness of the Translation of~CPRED$\omega$}\label{sec_cpred_soundness}
\renewcommand{\thesection}{\Alph{section}}
\addcontentsline{toc}{section}{\thesection\hspace{0.5em} Soundness of
  the Translation of~CPRED$\omega$}

In this section by~$\tau$, $\tau_1$, $\tau_2$, $\sigma$, $\sigma_1$,
etc. we denote types of~\mbox{CPRED$\omega$}, by $t$, $t_1$, $t_2$,
etc. terms of~\mbox{CPRED$\omega$}, by $\varphi$, $\psi$,
etc. formulas of~\mbox{CPRED$\omega$}, and by $s$, $s_1$, $s_2$,
etc. terms of~$\I_s$. By~$\Tc$ we denote the set of types
of~\mbox{CPRED$\omega$}. Below by~$\Delta$ we denote an arbitrary set
of formulas of~\mbox{CPRED$\omega$}.

If $t$ is a term of~\mbox{CPRED$\omega$} then we use the
notation~$\Gamma(t)$ for~$\Gamma(\{t\})$. By~$\Gamma(\Delta,\varphi)$
we denote~$\Gamma(\Delta \cup \{\varphi\})$. Below~$\Delta$ denotes an
arbitrary set of formulas of~\mbox{CPRED$\omega$}.

Note that if $\Delta \subseteq \Delta'$ then $\Gamma(\Delta) \subseteq
\Gamma(\Delta')$, after possibly renaming some variables $y \in
FV(\Gamma(\Delta))$ such that $y \notin FV(\Delta)$. Because of
Lemma~\ref{lem_I_s_subst}, we may assume without loss of generality
that this implication always holds verbatim, and also that $\Gamma(t_1
t_2) = \Gamma(t_1,t_2)$, $\Gamma(\varphi \supset \psi) =
\Gamma(\varphi,\psi)$, etc.

\begin{lemma}\label{lem_transl_is_type}
  $\Gamma(\Delta) \proves_{\I_s} \transl{\tau} : \Type$
\end{lemma}

\begin{proof}
  Straightforward induction on the structure of~$\tau$, using
  rule~$\to_t$.
\end{proof}

\begin{lemma}\label{lem_transl_typed}
  If $t \in T_\tau$ then $\Gamma(t) \proves \transl{t} :
  \transl{\tau}$.
\end{lemma}

\begin{proof}
  Induction on the structure of~$t$, using rules~$\to_i$, $\to_e$,
  $\supset_t$ and~$\forall_t$, and Lemma~\ref{lem_transl_is_type} and
  weakening.
\end{proof}

\begin{lemma}\label{lem_transl_subst}
  $\transl{t_1[x/t_2]} = \transl{t_1}[x/\transl{t_2}]$
\end{lemma}

\begin{proof}
  Straightforward induction on the structure of~$t_1$.
\end{proof}

\begin{lemma}\label{lem_transl_eq}
  If $t_1 \to_\beta t_2$ then $\Gamma(t_1,t_2) \proves_{\I_s}
  \transl{t_1} = \transl{t_2}$.
\end{lemma}

\begin{proof}
  Induction on the structure of~$t_1$, using the axiom~$\beta$, the
  rules for~Eq and Lemma~\ref{lem_transl_subst}.
\end{proof}

\begin{lemma}\label{lem_inhabited}
  For every $\tau \in \Tc$ there is an $\I_s$-term~$s$ such that
  $\Gamma(\Delta) \proves_{\I_s} s : \transl{\tau}$.
\end{lemma}

\begin{proof}
  Induction on the structure of~$\tau$. If $\tau \in \Bc$ then
  $\transl{\tau} = \tau$ and there exists~$y \in V_\tau$ such that $y
  \notin FV(\Delta)$ and $(y : \tau) \in \Gamma(\Delta)$. Thus
  $\Gamma(\Delta) \proves y : \transl{\tau}$. If $\tau = o$ then
  $\transl{\tau} = \Prop$ and e.g. $\Gamma(\Delta) \proves \bot :
  \transl{\tau}$. If $\tau \notin \Bc \cup \{o\}$ then $\tau =
  \tau_1\to\tau_2$ for some $\tau_1,\tau_2\in \Tc$. By the inductive
  hypothesis there exists an $\I_s$-term~$s_2$ such that
  $\Gamma(\Delta) \proves s_2 : \transl{\tau_2}$. Suppose $x \notin
  FV(s_2, \Gamma(\Delta))$. Then $\Gamma(\Delta), x : \transl{\tau_1}
  \proves s_2 : \transl{\tau_2}$. Since $\transl{\tau} =
  \transl{\tau_1} \to \transl{\tau_2}$ and $\Gamma(\Delta) \proves
  \transl{\tau_1} : \Type$ we obtain $\Gamma(\Delta) \proves (\lambda
  x \,.\, s_2) : \transl{\tau}$ by Lemma~\ref{lem_transl_is_type}.
\end{proof}

\begin{theorem}
  If $\Delta \proves_{\mathrm{CPRED}\omega} \varphi$ then
  $\Gamma(\Delta, \varphi), \transl{\Delta} \proves_{\I_s}
  \transl{\varphi}$. The same holds if we change~CPRED$\omega$ to
  \mbox{E-CPRED$\omega$} and~$\I_s$ to~$e\I_s$.
\end{theorem}

\newcommand{\CP}{{\ensuremath{{\mathrm{CPRED}\omega}}}}

\begin{proof}
  Induction on the length of the derivation of $\Delta \proves_{\CP}
  \varphi$.

  If $\Delta \proves_{\CP} \varphi$ is an instance of the axiom
  $\Delta', \varphi \proves_{\CP} \varphi$, then obviously
  $\Gamma(\Delta, \varphi), \transl{\Delta} \proves_{\I_s}
  \transl{\varphi}$ by axiom~1 of~$\I_s$. It is also easy to see that
  \[
  \Gamma(\Delta,\forall p : \Prop \,.\, ((p \supset \bot) \supset
  \bot) \supset p), \transl{\Delta} \proves_{\I_s} \transl{\forall p :
    \Prop \,.\, ((p \supset \bot) \supset \bot) \supset p}.
  \]

  Suppose $\Delta \proves_{\CP} \varphi_1 \supset \varphi_2$ is
  obtained from $\Delta, \varphi_1 \proves_{\CP} \varphi_2$ by the
  rule~$\supset_i^P$. Then by the inductive hypothesis $\Gamma(\Delta,
  \varphi_1, \varphi_2), \transl{\Delta}, \transl{\varphi_1} \proves_{\I_s}
  \transl{\varphi_2}$. By Lemma~\ref{lem_transl_typed} we have
  $\Gamma(\varphi_1) \proves_{\I_s} \transl{\varphi_1} : \Prop$. Hence by
  weakening and~$\supset_i$ we obtain $\Gamma(\Delta, \varphi_1,
  \varphi_2), \transl{\Delta} \proves_{\I_s} \transl{\varphi_1} \supset
  \transl{\varphi_2}$, i.e. $\Gamma(\Delta, \varphi_1 \supset
  \varphi_2), \transl{\Delta} \proves_{\I_s} \transl{\varphi_1 \supset
    \varphi_2}$.

  Suppose $\Delta \proves_{\CP} \varphi$ is a direct consequence of
  $\Delta \proves_{\CP} \psi \supset \varphi$ and $\Delta
  \proves_{\CP} \psi$ by the rule~$\supset_e^P$. By the IH
  and~$\supset_e$ we have $\Gamma(\Delta,\psi,\varphi),
  \transl{\Delta} \proves_{\I_s} \transl{\varphi}$. Note that
  $\Gamma(\Delta,\psi,\varphi) = \Gamma(\Delta,\varphi), x_1 :
  \transl{\tau_1}, \ldots, x_n : \transl{\tau_n}$ where
  $\{x_1,\ldots,x_n\} = FV(\psi) \setminus FV(\Delta, \varphi)$, $x_i
  \in V_{\tau_i}$, and without loss of generality $x_i \notin
  \Gamma(\Delta,\varphi)$. Since $x_i \notin FV(\Delta, \varphi)$ we
  have $x_i \notin FV(\transl{\Delta}, \transl{\varphi})$. By
  Lemma~\ref{lem_inhabited} there exist $\I_s$-terms $s_1, \ldots,
  s_n$ such that $\Gamma(\Delta, \varphi) \proves_{\I_s} s_i :
  \transl{\tau_i}$. By Lemma~\ref{lem_I_s_subst} we have
  \[
  \Gamma(\Delta,\varphi), \transl{\Delta}, s_1 : \transl{\tau_1},
  \ldots, s_n : \transl{\tau_n} \proves_{\I_s} \transl{\varphi}.
  \]
  By applying~$p_i$, $\supset_i$ and~$\supset_e$ successively we
  obtain $\Gamma(\Delta, \varphi), \transl{\Delta} \proves_{\I_s}
  \transl{\varphi}$.

  Suppose $\Delta \proves_{\CP} \forall x : \tau \,.\, \varphi$ is
  obtained from $\Delta \proves_{\CP} \varphi$ by
  rule~$\forall_i^P$. By the IH we have $\Gamma(\Delta,\varphi),
  \transl{\Delta} \proves_{\I_s} \transl{\varphi}$. Since $x \notin
  FV(\Delta)$, we have $\Gamma(\Delta, \varphi) = \Gamma(\Delta,
  \forall x : \tau \,.\, \varphi), x : \transl{\tau}$. By
  Lemma~\ref{lem_transl_is_type} and weakening we have
  \[
  \Gamma(\Delta, \forall x : \tau \,.\, \varphi), \transl{\Delta}
  \proves_{\I_s} \transl{\tau} : \Type.
  \]
  Hence $\Gamma(\Delta, \forall x : \tau \,.\, \varphi),
  \transl{\Delta} \proves_{\I_s} \transl{\forall x : \tau \,.\,
    \varphi}$ by~$\forall_i$ and $\transl{\forall x : \tau \,.\,
    \varphi} = \forall x : \transl{\tau} \,.\, \transl{\varphi}$.

  Suppose $\Delta \proves_{\CP} \varphi[x/t]$ is obtained from $\Delta
  \proves_{\CP} \forall x : \tau \,.\, \varphi$ by
  rule~$\forall_i^P$. We then have $x \notin FV(\Delta)$ and $t \in
  T_\tau$. By the IH we have
  \[
  \Gamma(\Delta, \forall x : \tau \,.\, \varphi), \transl{\Delta}
  \proves_{\I_s} \forall x : \transl{\tau} \,.\, \transl{\varphi}.
  \]
  If $x \notin FV(\varphi)$ then $\varphi \equiv \varphi[x/t]$ and
  $\Gamma(\Delta, \forall x : \tau \,.\, \varphi) = \Gamma(\Delta,
  \varphi[x/t])$. By Lemma~\ref{lem_inhabited} there exists an
  $\I_s$-term~$s$ such that $\Gamma(\Delta, \varphi[x/t])
  \proves_{\I_s} s : \transl{\tau}$. By~$\forall_e$ we obtain
  $\Gamma(\Delta, \varphi[x/t]), \transl{\Delta} \proves_{\I_s}
  \transl{\varphi[x/t]}$. If $x \in FV(\varphi)$ then
  $\Gamma(\Delta,\varphi[x/t]) = \Gamma(\Delta, \forall x : \tau \,.\,
  \varphi, t)$. By Lemma~\ref{lem_transl_typed} we have
  $\Gamma(\Delta,\varphi[x/t]) \proves \transl{t} : \transl{\tau}$, so
  by~$\forall_e$ and Lemma~\ref{lem_transl_subst} we obtain
  $\Gamma(\Delta,\varphi[x/t]), \transl{\Delta} \proves
  \transl{\varphi[x/t]}$.

  Suppose $\Delta \proves_{\CP} \psi$ is obtained from $\Delta
  \proves_{\CP} \varphi$ by rule~$\mathrm{conv}^P$. Then $\varphi
  =_\beta \psi$. By the IH we have $\Gamma(\Delta, \varphi),
  \transl{\Delta} \proves_{\I_s} \transl{\varphi}$. It suffices to
  show that if $\varphi \to_\beta \psi$ or $\psi \to_\beta \varphi$
  then $\Gamma(\Delta, \psi), \transl{\Delta} \proves_{\I_s}
  \transl{\psi}$. If $\varphi \to_\beta \psi$ then $\Gamma(\psi)
  \subset \Gamma(\varphi)$ and by Lemma~\ref{lem_transl_eq}, rule~eq
  and weakening we obtain $\Gamma(\Delta, \varphi), \transl{\Delta}
  \proves \transl{\psi}$. Note that $\Gamma(\Delta,\varphi) =
  \Gamma(\Delta,\psi), x_1 : \transl{\tau_1}, \ldots, x_n :
  \transl{\tau_n}$ where $\{x_1,\ldots,x_n\} = FV(\varphi) \setminus
  FV(\Delta, \psi)$, $x_i \in V_{\tau_i}$. Hence, we may use the same
  argument as in the proof for the rule~$\supset_e^P$ to obtain
  $\Gamma(\Delta,\psi), \transl{\Delta} \proves_{\I_s}
  \transl{\psi}$. If $\psi \to_\beta \varphi$ then $\Gamma(\varphi)
  \subseteq \Gamma(\psi)$, and by Lemma~\ref{lem_transl_eq}, rule~eq
  and weakening we obtain $\Gamma(\Delta,\psi), \transl{\Delta}
  \proves \transl{\psi}$.

  To show that if $\Delta \proves_{\mathrm{E{-}CPRED}\omega} \varphi$
  then $\Gamma(\Delta, \varphi), \transl{\Delta} \proves_{e\I_s}
  \transl{\varphi}$, it now suffices to prove that the translations of
  the axioms~$e_f^P$ and~$e_b^P$ are derivable in~$e\I_s$. This is
  straightforward, using the rules~$e_f$, $e_b$, eq,
  Lemma~\ref{lem_transl_is_type} and Lemma~\ref{lem_leibniz}.
\end{proof}

\newpage
\renewcommand{\thesection}{Appendix~\Alph{section}}
\tocless\section{Proofs for Section~\ref{sec_partiality}}\label{sec_proofs_examples}
\renewcommand{\thesection}{\Alph{section}}
\addcontentsline{toc}{section}{\thesection\hspace{0.5em} Proofs for Section~\ref{sec_partiality}}

In this section we give sketches of formal proofs for the examples in
Sect.~\ref{sec_partiality}. The formal derivations are slightly
simplified, by omitting certain steps and rule assumptions,
simplifying inferences, and generally omitting some parts which may be
easily reconstructed by the reader. We also assume certain basic
properties of operations on natural number and don't derive
them. These properties may be derived from Theorem~\ref{thm_peano},
Lemma~\ref{lem_nat_op_well_defined} and the definitions of~$\le$ and~$-$.

\begin{lemma}
  $\proves_{\I_s} \forall i,j : \Nat \,.\, (i \ge j) \supset
  (\subp\,i\,j = i - j)$
\end{lemma}

\begin{proof}
  Let $\varphi(x) \equiv \forall i,j : \Nat \,.\, i \ge j \supset x
  =_\Nat i - j \supset \subp\,i\,j = i - j$.
  \[
  (a0.0)
  \inferrule*{
    i : \Nat, j : \Nat \proves i : \Nat \\
    \inferrule*{
    }{
      \proves \forall x, y : \Nat \,.\, (x - y) : \Nat
    }\;\raisebox{0.7em}{by~\ref{lem_nat_op_well_defined}}
  }{
    i : \Nat, j : \Nat \proves \forall y : \Nat \,.\, (i - y) : \Nat
  }
  \]
  \[
  (a0)
  \inferrule*{
    \inferrule*{
      (a0.0)
      \\
      i : \Nat, j : \Nat \proves j : \Nat
    }{
      i : \Nat, j : \Nat \proves (i - j) : \Nat
    }
    \\
    \proves 0 : \Nat
  }{
    i : \Nat, j : \Nat \proves (0 =_\Nat (i - j)) : \Nat
  }\;\raisebox{0.7em}{by~\ref{lem_leibniz}}
  \]
  Completely analogously, we obtain
  \[
  (a1)\;\; i : \Nat, j : \Nat \proves (i \ge j) : \Nat
  \]
  Now let $\Gamma_0 \equiv i : \Nat, j : \Nat, i \ge j, 0 =_\Nat i -
  j$.
  \[
  (a2)
  \inferrule*{
    \inferrule*{
      \inferrule*{
        \inferrule*{
          \inferrule*{
            \inferrule*{
            }{
              \Gamma_0 \proves 0 = i - j
            }\;\raisebox{0.7em}{by~\ref{lem_ind_eq}}
          }{
            \Gamma_0 \proves \iszero (i - j)
          }
          \\
          \Gamma_0 \proves i \ge j
        }{
          \Gamma_0 \proves i =_\Nat j
        }\;\raisebox{0.7em}{by~\ref{lem_le_ge_then_eq}}
      }{
        \Gamma_0 \proves \subp\,i\,j = 0
      }
      \\
      \Gamma_0 \proves 0 = i - j
    }{
      \Gamma_0 \proves \subp\,i\,j = i - j
    }
    \\
    (a0)
  }{
    i : \Nat, j : \Nat, i \ge j \proves 0 =_\Nat i - j \supset \subp\,i\,j = i - j
  }
  \]
  \[
  (a)
  \inferrule*{
    \inferrule*{
      (a1)
      \\
      (a2)
    }{
      i : \Nat, j : \Nat \proves i \ge j \supset 0 =_\Nat i - j \supset \subp\,i\,j = i - j
    }
    \\
    \proves \Nat : \Type
  }{
    \proves \varphi(0)
  }
  \]
  Let $\Gamma_1 \equiv y : \Nat, \varphi(y), i : \Nat, j : \Nat, \s y
  =_\Nat i - j$.
  \[
  (b0.0)
  \inferrule*{
    \Gamma_1 \proves i : \Nat \\ \Gamma_1 \proves j : \Nat
  }{
    \Gamma_1 \proves (i =_\Nat j) : \Prop
  }\;\raisebox{0.7em}{by~\ref{lem_leibniz}}
  \]
  \[
  (b0.1)
  \inferrule*{
    \inferrule*{
      \inferrule*{
        \inferrule*{
        }{
          \Gamma_1, i =_\Nat j \proves i = j
        }\;\raisebox{0.7em}{by~\ref{lem_ind_eq}}
        \\
        \Gamma_1 \proves \s y =_\Nat i - j
      }{
        \Gamma_1, i =_\Nat j \proves \s y =_\Nat i - i
      }
    }{
      \Gamma_1, i =_\Nat j \proves \s y =_\Nat 0
    }
    \\
    \inferrule*{
      \Gamma_1 \proves y : \Nat
    }{
      \Gamma_1 \proves (\s y) : \Nat
    }
  }{
    \Gamma_1, i =_\Nat j \proves \s y = 0
  }
  \]
  \[
  (b0)
  \inferrule*{
    \inferrule*{
      \inferrule*{
        (b0.1)
        \\
        \Gamma_1, i =_\Nat j \proves \iszero 0
      }{
        \Gamma_1, i =_\Nat j \proves \iszero (\s y)
      }
      \\
      \Gamma_1, i =_\Nat j \proves \neg (\iszero (\s y))
    }{
      \Gamma_1, i =_\Nat j \proves \bot
    }
    \\
    (b0.0)
  }{
    \Gamma_1 \proves \neg (i =_\Nat j)
  }
  \]
  \[
  (c0)
  \inferrule*{
    \Gamma_1 \proves \varphi(y)
    \\
    \Gamma_1 \proves i : \Nat
    \\
    \inferrule*{
      \Gamma_1 \proves j : \Nat
    }{
      \Gamma_1 \proves (\s j) : \Nat
    }
  }{
    \Gamma_1 \proves i \ge (\s j) \supset y =_\Nat i - (\s j) \supset
    \subp\,i\,(\s j) = i - (\s j)
  }
  \]
  \[
  (c1)
  \inferrule*{
    \inferrule*{
      \Gamma_1 \proves i \ge j \\ (b0)
    }{
      \Gamma_1 \proves i \ge (\s j)
    }
    \\
    (c0)
  }{
    \Gamma_1 \proves y =_\Nat i - (\s j) \supset \subp\,i\,(\s j) = i - (\s j)
  }
  \]
  \[
  (c2)
  \inferrule*{
    \Gamma_1 \proves \s y =_\Nat i - j
    \\
    \inferrule*{
        \Gamma_1 \proves i : \Nat \\ \Gamma_1 \proves j : \Nat
    }{
      \Gamma_1 \proves (i - j) : \Nat
    }
    \\
    \inferrule*{
      \Gamma_1 \proves y : \Nat
      }{
      \Gamma_1 \proves (\s y) : \Nat
    }
  }{
    \Gamma_1 \proves \s y = i - j
  }
  \]
  \[
  (c3)
  \inferrule*{
    \inferrule*{
      \inferrule*{
        \Gamma_1 \proves i \ge j \\ (b0)
      }{
        \Gamma_1 \proves i > j
      }
    }{
      \ldots
    }
  }{
    \Gamma_1 \proves \s (i - (\s j)) = i - j
  }
  \]
  \[
  (c4)
  \inferrule*{
    \inferrule*{
      \inferrule*{
        \inferrule*{
          (c3)
          \\
          (c2)
        }{
          \Gamma_1 \proves \s (i - (\s j)) = \s y
        }
      }{
        \Gamma_1 \proves \p (\s (i - (\s j))) = \p (\s y)
      }
    }{
      \Gamma_1 \proves i - (\s j) = y
    }
  }{
    \Gamma_1 \proves y =_\Nat i - (\s j)
  }
  \]
  \[
  (c)
  \inferrule*{
    \inferrule*{
      \inferrule*{
        \inferrule*{
          (b0)
        }{
          \Gamma_1 \proves \subp\,i\,j = \s (\subp\,i\,(\s j))
        }
        \\
        \inferrule*{
          \inferrule*{
            (c1) \\ (c4)
          }{
            \Gamma_1 \proves \subp\,i\,(\s j) = i - (\s j)
          }
        }{
          \Gamma_1 \proves \s (\subp\,i\,(\s j)) = \s (i - (\s j))
        }
      }{
        \Gamma_1 \proves \subp\,i\,j = \s (i - (\s j)) \\ (c3)
      }
    }{
      \Gamma_1 \proves \subp\,i\,j = i - j
    }
  }{
    y : \Nat, \varphi(y) \proves \varphi(\s y)
  }
  \]
  \[
  (d)
  \inferrule*{
    (a) \\ (c)
  }{
    \proves \forall y,i,j : \Nat \,.\, i \ge j \supset y =_\Nat i - j \supset \subp\,i\,j = i - j
  }\;\raisebox{0.7em}{by~$n_i$}
  \]
  \[
  (e)
  \inferrule*{
    \inferrule*{
      \inferrule*{
        \inferrule*{
          \ldots
        }{
          i : \Nat, j : \Nat \proves (i - j) : \Nat
        }
        \\
        (d)
      }{
        i : \Nat, j : \Nat \proves \forall k,l : \Nat \,.\, k \ge l
        \supset i - j =_\Nat k - l \supset \subp\,k\,l = k - l
      }
    }{
      i : \Nat, j : \Nat \proves i \ge j \supset i - j =_\Nat i - j \supset \subp\,i\,j = i - j
    }
  }{
    i : \Nat, j : \Nat, i \ge j \proves i - j =_\Nat i - j \supset \subp\,i\,j = i - j
  }
  \]
  \[
  (f)
  \inferrule*{
    (e)
    \\
    \inferrule*{
      \inferrule*{
        \ldots
      }{
        i : \Nat, j : \Nat \proves (i - j) : \Nat
      }
    }{
      i : \Nat, j : \Nat \proves i - j =_\Nat i - j
    }
  }{
    i : \Nat, j : \Nat, i \ge j \proves \subp\,i\,j = i - j
  }
  \]
  \[
  \inferrule*{
    \inferrule*{
      (f)
      \\
      \inferrule*{
        \ldots
      }{
        i : \Nat, j : \Nat \proves (i \ge j) : \Prop
      }
    }{
      i : \Nat, j : \Nat \proves i \ge j \supset \subp\,i\,j = i - j
    }
    \\
    \proves \Nat : \Type
  }{
    \proves \forall i,j : \Nat \,.\, i \ge j \supset \subp\,i\,j = i - j
  }
  \]
\end{proof}

\begin{lemma}
  If~$f$ is a term such that
  \[
  \proves_{\I_s} f(n) = \Cond{(n > 100)}{(n - 10)}{(f(f(n + 11)))}
  \]
  then
  \[
  \proves_{\I_s} \forall n : \Nat \,.\, n \le 101 \supset f(n) = 91.
  \]
\end{lemma}

\begin{proof}
  Let $\varphi(y) \equiv \forall n : \Nat \,.\, n \le 101 \supset 101
  - n \le y \supset f(n) = 91$.
  \[
  (a0)
  \inferrule*{
    \inferrule*{
      n : \Nat, n \le 101, 101 - n \le 0 \proves 101 \le n
    }{
      n : \Nat, n \le 101, 101 - n \le 0 \proves n =_\Nat 101
    }\;\raisebox{0.7em}{by~\ref{lem_le_ge_then_eq}}
  }{
    n : \Nat, n \le 101, 101 - n \le 0 \proves f(n) = f(101) = 101 -
    10 = 91
  }
  \]
  \[
  (a)
  \inferrule*{
    \inferrule*{
      \inferrule*{
        \inferrule*{
          n : \Nat \proves (101 - n) : \Nat
        }{
          n : \Nat \proves (101 - n \le 0) : \Prop
        }
        \\
        (a0)
      }{
        n : \Nat, n \le 101 \proves 101 - n \le 0 \supset f(n) = 91
      }
      \\
      n : \Nat \proves (n \le 101) : \Prop
    }{
      n : \Nat \proves n \le 101 \supset 101 - n \le 0 \supset f(n) = 91
    }
    \\
    \proves \Nat : \Type
  }{
    \proves \varphi(0)
  }
  \]
  Let $\Gamma_0 \equiv y : \Nat, \varphi(y), n : \Nat, n \le 101, 101
  - n \le \s y$. Let $\Gamma_1 \equiv \Gamma_0, n + 11 > 101, n <
  101$.
  \[
  (b0.0)
  \inferrule*{
    \inferrule*{
      \Gamma_1 \proves n : \Nat
    }{
      \Gamma_1 \proves (n + 1) : \Nat
    }
    \\
    \Gamma_1 \proves \varphi(y)
  }{
    \Gamma_1 \proves n + 1 \le 101 \supset 101 - (n + 1) \le y
    \supset f(n + 1) = 91
  }
  \]
  \[
  (b0)
  \inferrule*{
    (b0.0)
    \\
    \Gamma_1 \proves n + 1 \le 101
  }{
    \Gamma_1 \proves 101 - (n + 1) \le y \supset f(n + 1) = 91
  }
  \]
  \[
  (b1)
  \inferrule*{
    \inferrule*{
      \inferrule*{
        \Gamma_1 \proves 101 - n \le \s y
        \\
        \inferrule*{
          \Gamma_1 \proves 101 - (\s n) = \p (101 - n)
        }{
          \Gamma_1 \proves \s (101 - (\s n)) = 101 - n
        }
      }{
        \Gamma_1 \proves \s (101 - (n + 1)) \le \s y
      }
    }{
      \Gamma_1 \proves 101 - (n + 1) \le y
    }
    \\
    (b0)
  }{
    \Gamma_1 \proves f(n + 1) = 91
  }
  \]
  \[
  (b)
  \inferrule*{
    \inferrule*{
      \inferrule*{
        \Gamma_1 \proves n + 11 > 100
      }{
        \Gamma_1 \proves f(n + 11) = n + 11 - 10 = n + 1
      }
      \\
      \inferrule*{
        \Gamma_1 \proves n < 101
      }{
        \Gamma_1 \proves \neg (n > 100)
      }
    }{
      \Gamma_1 \proves f(n) = f(f(n + 11)) = f(n + 1)
    }
    \\
    (b1)
  }{
    \Gamma_1 \proves f(n) = 91
  }
  \]
  Let $\Gamma_2 \equiv \Gamma_0, n + 11 > 101, \neg (n < 101)$.
  \[
  (c)
  \inferrule*{
    \inferrule*{
      \Gamma_2 \proves n \le 101 \\ \Gamma_2 \proves \neg (n < 101)
    }{
      \Gamma_2 \proves n =_\Nat 101
    }
  }{
    \Gamma_2 \proves f(n) = n - 10 = 101 - 10 = 91
  }
  \]
  \[
  (d)
  \inferrule*{
    \inferrule*{
      \inferrule*{
        \Gamma_0, n + 11 > 101 \proves n : \Nat
      }{
        \Gamma_0, n + 11 > 101 \proves (n < 101) : \Prop
      }
    }{
      \Gamma_0, n + 11 > 101 \proves n < 101 \lor \neg (n < 101)
    }
    \\
    (b)
    \\
    (c)
  }{
    \Gamma_0, n + 11 > 101 \proves f(n) = 91
  }\;\raisebox{0.7em}{by~$\lor_e$}
  \]
  Let $\Gamma_3 \equiv \Gamma_0, n + 11 \le 101$.
  \[
  (e0.1)
  \inferrule*{
    \Gamma_3 \proves \varphi(y)
    \\
    \inferrule*{
      \Gamma_3 \proves n : \Nat \\ \Gamma_3 \proves 11 : \Nat
    }{
      \Gamma_3 \proves (n + 11) : \Nat
    }
  }{
    \Gamma_3 \proves n + 11 \le 101 \supset 101 - (n + 11) \le y
    \supset f(n + 11) = 91
  }
  \]
  \[
  (e0.2)
  \inferrule*{
    (e0.1)
    \\
    \Gamma_3 \proves n + 11 \le 101
  }{
    \Gamma_3 \proves 101 - (n + 11) \le y \supset f(n + 11) = 91
  }
  \]
  \[
  (e0)
  \inferrule*{
    (e0.2)
    \\
    \inferrule*{
      \Gamma_3 \proves 101 - n \le \s y
    }{
      \Gamma_3 \proves 101 - (n + 11) \le y.
    }
  }{
    \Gamma_3 \proves f(n + 11) = 91
  }
  \]
  \[
  (e1)
  \inferrule*{
    \inferrule*{
      \inferrule*{
        \Gamma_3 \proves n + 11 \le 101 \\ \Gamma_3 \proves n : \Nat \\ \ldots
      }{
        \Gamma_3 \proves n \le 100
      }
    }{
      \Gamma_3 \proves f(n) = f(f(n+11))
    }
    \\
    (e0)
  }{
    \Gamma_3 \proves f(n) = f(91)
  }
  \]
  \[
  (e2.1)
  \inferrule*{
    \Gamma_3 \proves f(91) = f(f(102)) \\ \Gamma_3 \proves f(102) = 92
  }{
    \Gamma_3 \proves f(91) = f(92)
  }
  \]
  \[
  (e2.2)
  \inferrule*{
    \Gamma_3 \proves f(92) = f(f(103)) \\ \Gamma_3 \proves f(103) = 93
  }{
    \Gamma_3 \proves f(92) = f(93)
  }
  \]
  \[
  \vdots
  \]
  \[
  (e2.10)
  \inferrule*{
    \inferrule*{
      \Gamma_3 \proves f(100) = f(f(111)) \\ \Gamma_3 \proves f(111) = 101
    }{
      \Gamma_3 \proves f(100) = f(101)
    }
    \\
    \Gamma_3 \proves f(101) = 91
  }{
    \Gamma_3 \proves f(100) = 91
  }
  \]
  \[
  (e2)
  \inferrule*{
    \inferrule*{
      \inferrule*{
        \inferrule*{
          \inferrule*{
            (e2.1)
            \\
            (e2.2)
          }{
            \Gamma_3 \proves f(91) = f(93)
          }
          \\
          (e2.3)
        }{
          \Gamma_3 \proves f(91) = f(94)
        }
        \\
        (e2.4)
      }{
        \vdots
      }
      \\
      \vdots
    }{
      \Gamma_3 \proves f(91) = f(100)
    }
    \\
    (e2.10)
  }{
    \Gamma_3 \proves f(91) = 91
  }
  \]
  \[
  (e)
  \inferrule*{
    (e1) \\ (e2)
  }{
    \Gamma_3 \proves f(n) = 91
  }
  \]
  \[
  (f)
  \inferrule*{
    \inferrule*{
      \inferrule*{
        \Gamma_0 \proves n : \Nat
      }{
        \Gamma_0 \proves (n + 11 > 101) : \Prop
      }
    }{
      \Gamma_0 \proves n + 11 > 101 \lor n + 11 \le 101
    }
    \\
    (d)
    \\
    (e)
  }{
    \Gamma_0 \proves f(n) = 91
  }\;\raisebox{0.7em}{by~$\lor_e$}
  \]
  \[
  (g0)
  \inferrule*{
    y : \Nat, \varphi(y), n : \Nat \proves (101 - n \le \s y) : \Prop
    \\
    (f)
  }{
    y : \Nat, \varphi(y), n : \Nat, n \le 101 \proves 101 - n \le \s y \supset
    f(n) = 91
  }
  \]
  \[
  (g1)
  \inferrule*{
    y : \Nat, \varphi(y), n : \Nat \proves (n \le 101) : \Prop
    \\
    (g0)
  }{
    y : \Nat, \varphi(y), n : \Nat \proves n \le 101 \supset 101 - n
    \le \s y \supset f(n) = 91
  }
  \]
  \[
  (g)
  \inferrule*{
    \proves \Nat : \Type
    \\
    (g1)
  }{
    y : \Nat, \varphi(y) \proves \varphi(\s y)
  }
  \]
  \[
  (h)
  \inferrule*{
    (a) \\ (g)
  }{
    \proves \forall y, n : \Nat \,.\, n \le 101 \supset 101 - n \le y \supset
    f(n) = 91
  }\;\raisebox{0.7em}{by~$n_i$}
  \]
  \[
  (i)
  \inferrule*{
    n : \Nat \proves (101 - n) : \Nat
    \\
    (h)
  }{
    n : \Nat \proves \forall m : \Nat \,.\, m \le 101 \supset 101 - m
    \le 101 - n \supset f(m) = 91
  }
  \]
  \[
  (j)
  \inferrule*{
    \inferrule*{
      \inferrule*{
        n : \Nat \proves n : \Nat
        \\
        (i)
      }{
        n : \Nat \proves n \le 101 \supset 101 - n \le 101 - n \supset
        f(n) = 91
      }
    }{
      n : \Nat, n \le 101 \proves 101 - n \le 101 - n \supset f(n) = 91
    }
  }{
    n : \Nat, n \le 101 \proves f(n) = 91
  }
  \]
  \[
  \inferrule*{
    \inferrule*{
      n : \Nat \proves (n \le 101) : \Prop
      \\
      (j)
    }{
      n : \Nat \proves n \le 101 \supset f(n) = 91
    }
    \\
    \proves \Nat : \Type
  }{
    \proves \forall n : \Nat \,.\, n \le 101 \supset f(n) = 91
  }
  \]
\end{proof}

\newpage
\renewcommand{\thesection}{Appendix~\Alph{section}}
\tocless\section{Semantics}\label{sec_semantics}
\renewcommand{\thesection}{\Alph{section}}
\addcontentsline{toc}{section}{\thesection\hspace{0.5em} Semantics}

In this appendix we define a semantics for~$\I_s$ and
for~$e\I_s$. This semantics will be used in Appendix~F to show
consistency of~$\I_s$ and~$e\I_s$.

\newcommand{\setval}[3]{\ensuremath{\{\{#1 \;|\; #2\}\}_{#3}}}

\begin{definition}\label{def_I_s_structure} \rm
  An \emph{$\I_s$-structure} is a triple $\Ac = \langle A, \cdot,
  \valuation{}{}{} \rangle$ where~$A$ is the \emph{domain} of~$\Ac$,
  $\cdot$ is a binary operation on~$A$, and the \emph{interpretation}
  $\valuation{}{}{} : \T \times A^V \to A$ is a function from
  $\I_s$-terms and valuations to~$A$. We sometimes write $\cdot^\Ac$
  and $\valuation{}{}{\Ac}$ to indicate that these are components
  of~$\Ac$.

  A \emph{valuation}~$v$ is a function $v : V \to A$. We usually write
  $\valuation{t}{v}{}$ instead of $\valuation{}{}{}(t, v)$, and we
  drop the subscript when obvious or irrelevant. To stress that a
  valuation is associated with an $\I_s$-structure~$\Ac$, we sometimes
  call it an $\Ac$-valuation. By~$v[x/a]$ for~$a \in A$ we denote a
  valuation~$u$ such that $u(y) = v(y)$ for $y \ne x$ and $u(x) =
  a$. We use the abbreviations $\Tc^\Ac = \{ a \in A \;|\;
  \valuation{x : \Type}{\{x \mapsto a\}}{\Ac} =
  \valuation{\top}{}{\Ac}\}$ and $\iota^\Ac = \{ a \in A \;|\;
  \valuation{x : \iota}{\{x \mapsto a\}}{\Ac} =
  \valuation{\top}{}{\Ac}\}$ for $\iota \in \Tc_I$. The symbols $a$,
  $a'$, $b$, $b'$, etc. denote elements of~$A$, unless otherwise
  stated. We often confuse $\top$, $\bot$, $\mathrm{Is}$, etc. with
  $\valuation{\top}{}{}$, $\valuation{\bot}{}{}$,
  $\valuation{\mathrm{Is}}{}{}$, etc., to avoid onerous notation. It
  is always clear from the context which interpretation is meant.
\end{definition}

\begin{definition}\label{def_I_s_model} \rm
  An $\I_s$-model is an $\I_s$-structure~$\Ac$ satisfying the
  following requirements:
  \begin{itemize}
  \item[(var)] $\valuation{x}{v}{} = v(x)$ for $x \in V$,
  \item[(app)] $\valuation{t_1 t_2}{v}{} = \valuation{t_1}{v}{} \cdot
    \valuation{t_2}{v}{}$,
  \item[($\beta$)] $\valuation{\lambda x \,.\, t}{v}{} \cdot a =
    \valuation{t}{v[x/a]}{}$ for every $a \in A$,
  \item[(fv)] if $\cut{v}{FV(t)} = \cut{w}{FV(t)}$ then
    $\valuation{t}{v}{} = \valuation{t}{w}{}$,
  \item[($\xi$)] if for all $a \in A$ we have $\valuation{\lambda x
    \,.\, t_1}{v}{} \cdot a = \valuation{\lambda x \,.\, t_2}{v}{} \cdot a$
    then $\valuation{\lambda x \,.\, t_1}{v}{} = \valuation{\lambda x
      \,.\, t_2}{v}{}$,
  \item[(pr)] $\mathrm{Is} \cdot a \cdot \Prop = \top$ iff $a \in
    \{\top, \bot\}$,
  \item[(pt)] $\valuation{\Prop : \Type}{}{} = \top$,
  \item[(it)] $\valuation{\iota : \Type}{}{} = \top$ for $\iota \in
    \Tc_I$,
  \item[($\forall_\top$)] if $a \in \Tc^\Ac$ and for all $c \in A$ such
    that $\mathrm{Is} \cdot c \cdot a = \top$ we have $b \cdot c = \top$
    then $\forall \cdot a \cdot b = \top$,
  \item[($\forall_\bot$)] if $a \in \Tc^\Ac$ and there exists $c \in A$
    such that $\mathrm{Is} \cdot c \cdot a = \top$ and $b \cdot c =
    \bot$ then $\forall \cdot a \cdot b = \bot$,
  \item[($\forall_e$)] if $\forall \cdot a \cdot b = \top$ then for
    all $c \in A$ such that $\mathrm{Is} \cdot c \cdot a = \top$ we
    have $b \cdot c = \top$,
  \item[($\forall_e'$)] if $\forall \cdot a \cdot b = \bot$ then there
    exists $c \in A$ such that $\mathrm{Is} \cdot c \cdot a = \top$
    and $b \cdot c = \bot$,
  \item[($\vee_1$)] $\vee \cdot a \cdot b = \top$ iff $a = \top$ or $b
    = \bot$,
  \item[($\vee_2$)] $\vee \cdot a \cdot b = \bot$ iff $a = \bot$ and
    $b = \bot$,
  \item[($\supset_{t2}$)] if $\supset \cdot a \cdot b = \top$ then $a
    \in \{\top,\bot\}$, where $\supset = \valuation{\lambda x y \,.\,
      x \supset y}{}{}$,
  \item[($\bot$)] if $\bot = \top$ then for all $a \in A$ we have $a =
    \top$,
  \item[($\to_i$)] if $f \in A$, $a \in \Tc^\Ac$ and for all $c \in A$
    such that $\mathrm{Is} \cdot c \cdot a = \top$ we have
    $\mathrm{Is} \cdot (f \cdot c) \cdot b = \top$, then $\mathrm{Is}
    \cdot f \cdot (\mathrm{Fun} \cdot a \cdot b) = \top$,
  \item[($\to_e$)] if $\mathrm{Is} \cdot f \cdot (\mathrm{Fun} \cdot a
    \cdot b) = \top$ and $\mathrm{Is} \cdot c \cdot a = \top$ then
    $\mathrm{Is} \cdot (f \cdot c) \cdot b = \top$,
  \item[($\to_t$)] if $a, b \in \Tc^\Ac$ then $\mathrm{Fun} \cdot a \cdot
    b \in \Tc^\Ac$,
  \item[(s1)] if $\mathrm{Subtype} \cdot a \cdot b \in \Tc^\Ac$,
    $\mathrm{Is} \cdot c \cdot a = \top$ and $b \cdot c = \top$, then
    $\mathrm{Is} \cdot c \cdot (\mathrm{Subtype} \cdot a \cdot b) =
    \top$,
  \item[(s2)] if $\mathrm{Is} \cdot c \cdot (\mathrm{Subtype} \cdot a
    \cdot b) = \top$ then $b \cdot c = \top$,
  \item[(s3)] if $\mathrm{Is} \cdot c \cdot (\mathrm{Subtype} \cdot a
    \cdot b) = \top$ then $\mathrm{Is} \cdot c \cdot a = \top$,
  \item[(s4)] if $a \in \Tc^\Ac$ and for all $c \in A$ such that
    $\mathrm{Is} \cdot c \cdot a = \top$ we have $b \cdot c \in
    \{\top, \bot\}$, then $\mathrm{Subtype} \cdot a \cdot b \in \Tc^\Ac$,
  \item[(o1)] $o_i^\iota \cdot (c_i^\iota \cdot a_1 \cdot \ldots \cdot
    a_{n_i}) = \top$ where $o_i^\iota \in \Oc$ and $c_i^\iota \in \Cc$
    has arity~$n_i$,
  \item[(o2)] $o_i^\iota \cdot (c_j^\iota \cdot a_1 \cdot \ldots \cdot
    a_{n_j}) = \bot$ where $i \ne j$, $o_i^\iota \in \Oc$, $c_j^\iota
    \in \Cc$ has arity~$n_j$,
  \item[(d1)] $d_{i,j}^\iota \cdot (c_i^\iota \cdot a_1 \cdot \ldots
    \cdot a_{n_i}) = a_j$ where $d_{i,j}^\iota \in \Dc$ and $c_i^\iota
    \in \Cc$ has arity~$n_i$,
  \item[(i1)] if for all $i=1,\ldots,m$ we have:
    \begin{itemize}
    \item for all $b_{1},\ldots,b_{n_i}$ such that
      \begin{itemize}
      \item $b_{j} \in {\iota_{i,j}^*}^\Ac$ for $j=1,\ldots,n_i$, and
      \item $a \cdot b_{k} = \top$ for $1 \le k \le n_i$ such that
        $\iota_{i,k} = \star$
      \end{itemize}
      we have $a \cdot (c_i^\iota \cdot b_{1} \cdot \ldots \cdot
      b_{n_i}) = \top$
    \end{itemize}
    then $\forall\cdot \iota^\Ac \cdot a = \top$, where $c_i^\iota \in
    \Cc$ has arity~$n_i$, $\iota \in \Tc_I$, $\iota_{i,j}^* =
    \iota_{i,j}$ if $\iota_{i,j} \in \Tc_I$, $\iota_{i,j}^* = \iota$
    if $\iota_{i,j} = \star$, and $\iota = \mu(\langle \iota_{1,1},
    \ldots, \iota_{1,n_1} \rangle, \ldots, \langle \iota_{m,1},
    \ldots, \iota_{m,m_i} \rangle)$,
  \item[(i2)] if $a_j \in {\iota_{i,j}^*}^\Ac$ for $i=1,\ldots,n_i$
    then $c_i^\iota \cdot a_1 \cdot \ldots \cdot a_{n_i} \in
    \iota^\Ac$, where $c_i^\iota \in \Cc$ has arity~$n_i$, $\iota \in
    \Tc_I$, $\iota_{i,j}^* = \iota_{i,j}$ if $\iota_{i,j} \in \Tc_I$,
    $\iota_{i,j}^* = \iota$ if $\iota_{i,j} = \star$, $\iota =
    \mu(\langle \iota_{1,1}, \ldots, \iota_{1,n_1} \rangle, \ldots,
    \langle \iota_{m,1}, \ldots, \iota_{m,m_i} \rangle)$,
  \item[($\epsilon$)] if for $a \in \Tc^\Ac$ there exists $b \in A$ such
    that $\mathrm{Is} \cdot b \cdot a = \top$ then $\mathrm{Is} \cdot
    (\epsilon \cdot a) \cdot a = \top$,
  \item[(c1)] $\mathrm{Cond} \cdot \top \cdot a \cdot b = a$,
  \item[(c2)] $\mathrm{Cond} \cdot \bot \cdot a \cdot b = b$,
  \item[(eq)] $\mathrm{Eq} \cdot a \cdot b = \top$ iff $a = b$,
  \end{itemize}

  An \emph{$e\I_s$-model} is an \emph{$\I_s$-model} which additionally
  satisfies:
  \begin{itemize}
  \item[(ef)] if $a, b \in \Tc^\Ac$, $\mathrm{Is} \cdot f \cdot
    (\mathrm{Fun} \cdot a \cdot b) = \top$, $\mathrm{Is} \cdot g \cdot
    (\mathrm{Fun} \cdot a \cdot b) = \top$, and for all $c \in A$ such
    that $\mathrm{Is} \cdot c \cdot a = \top$, and all $p \in A$ such
    that $\mathrm{Is} \cdot p \cdot (\mathrm{Fun} \cdot b \cdot \Prop)
    = \top$ we have $p \cdot (f \cdot c) = p \cdot (g \cdot c)$, then
    for all $p \in A$ such that $\mathrm{Is} \cdot p \cdot
    (\mathrm{Fun} \cdot (\mathrm{Fun} \cdot a \cdot b) \cdot \Prop) =
    \top$ we have $p \cdot f = p \cdot g$,
  \item[(eb)] if ${\supset} \cdot a \cdot b = {\supset} \cdot b \cdot
    a = \top$ then $a = b$.
  \end{itemize}
  Here ${\supset} = \valuation{\lambda x y \,.\, x \supset y}{}{}$.

  For a term $\varphi$ and a valuation $u$ we write $\Ac,u \models
  \varphi$ if $\valuation{\varphi}{u}{\Ac} = \top$. Given a set of
  terms~$\Gamma$, we use the notation $\Ac, u \models \Gamma$ if $\Ac,
  u \models \varphi$ for all $\varphi \in \Gamma$. We drop the
  subscript~$\Ac$ when obvious or irrelevant. We write $\Gamma
  \models_{\I_s} \varphi$ if for every $\I_s$-model~$\Ac$ and every
  valuation~$u$, the condition $\Ac, u \models \Gamma$ implies $\Ac, u
  \models \varphi$. The notation $\Gamma \models_{e\I_s} \varphi$ is
  defined analogously. The subscript is dropped when obvious from the
  context or irrelevant.
\end{definition}

Note that every (nontrivial) $\I_s$-model is a $\lambda$-model. See
e.g. \cite[Chapter 5]{Barendregt1984} for a definition of a
$\lambda$-model.

\begin{lemma}\label{lem_I_s_model_conditions}
  In every $\I_s$-model~$\Ac$ the following conditions hold.
  \begin{itemize}
  \item[$(\wedge_1)$] ${\wedge} \cdot a \cdot b = \top$ iff $a = \top$
    and $b = \top$,
  \item[$(\wedge_2)$] ${\wedge} \cdot a \cdot b = \bot$ iff $a = \bot$
    or $b = \bot$,
  \item[$(\exists_\top)$] if $a \in \Tc^\Ac$ and there exists $c \in
    A$ such that $\mathrm{Is} \cdot c \cdot a = \top$ and $b \cdot c =
    \top$ then $\exists \cdot a \cdot b = \top$,
  \item[$(\exists_\bot)$] if $a \in \Tc^\Ac$ and for all $c \in A$
    such that $\mathrm{Is} \cdot c \cdot a = \top$ we have $b \cdot c
    = \bot$ then $\exists \cdot a \cdot b = \bot$,
  \item[$(\exists_e)$] if $\exists \cdot a \cdot b = \top$ then there
    exists $c \in A$ such that $\mathrm{Is} \cdot c \cdot a = \top$
    and $b \cdot c = \top$,
  \item[($\supset_1$)] $\supset \cdot a \cdot b = \top$ iff $a =
    \bot$, or $a = b =\top$,
  \item[($\supset_2$)] $\supset \cdot a \cdot b = \bot$ iff $a = \top$
    and $b = \bot$,
  \item[($\neg_\top$)] $\neg \cdot a = \bot$ iff $a = \top$,
  \item[($\neg_\bot$)] $\neg \cdot a = \top$ iff $a = \bot$.
  \end{itemize}
  Here $\supset = \valuation{\lambda x y \,.\, x \supset y}{}{}$,
  $\neg = \valuation{\lambda x \,.\, \neg x}{}{}$, $\wedge =
  \valuation{\lambda x y \,.\, \neg ((\neg x) \vee (\neg y))}{}{}$ and
  $\exists = \valuation{\lambda x y \,.\, \neg (\forall\, x\, \lambda
    z \,.\, \neg(y z))}{}{}$.
\end{lemma}

\begin{proof}
  Easy.
\end{proof}

\begin{lemma}\label{lem_I_s_ef_condition}
  In every $\I_s$-model~$\Ac$, for any $a \in \Tc^\Ac$ and $b, c \in
  \Ac$ such that $\mathrm{Is} \cdot b \cdot a = \top$ and $\mathrm{Is}
  \cdot c \cdot a = \top$, the condition
  \begin{itemize}
  \item[(1)] for all $p \in \Ac$ such that $\mathrm{Is} \cdot p \cdot
    (\mathrm{Fun} \cdot a \cdot \Prop) = \top$ we have ${\supset}
    \cdot (p \cdot b) \cdot (p \cdot c) = \top$
  \end{itemize}
  is equivalent to
  \begin{itemize}
  \item[(2)] for all $p \in \Ac$ such that $\mathrm{Is} \cdot p \cdot
    (\mathrm{Fun} \cdot a \cdot \Prop) = \top$ we have $p \cdot b = p
    \cdot c$
  \end{itemize}
  where ${\supset} = \valuation{\lambda x y \,.\, x \supset y}{}{}$.
\end{lemma}

\begin{proof}
  It is obvious from the definition of~${\supset}$ and the
  conditions~($\to_e$) and~(pr) that~(2) implies~(1). We show that~(1)
  implies~(2). Let $p \in \Ac$ be such that $\mathrm{Is} \cdot p \cdot
  (\mathrm{Fun} \cdot a \cdot \Prop) = \top$ and let $b, c \in \Ac$ be
  such that $\mathrm{Is} \cdot b \cdot a = \top$ and $\mathrm{Is}
  \cdot c \cdot a = \top$. By~($\to_e$) and~(pr) we have $p \cdot b
  \in \{\top,\bot\}$ and $p \cdot c \in \{\top,\bot\}$. If $p \cdot b
  = \top$ then we must have $p \cdot c = \top$ as well, by definition
  of~${\supset}$. If $p \cdot b = \bot$ then it is easy to see that
  $p' = \valuation{\lambda x . \neg(y x)}{\{y \mapsto p\}}{}$ also
  satisfies $\mathrm{Is} \cdot p' \cdot (\mathrm{Fun} \cdot a \cdot
  \Prop) = \top$. We have $p' \cdot b = \top$, which implies $p' \cdot
  c = \top$, which means that $p \cdot c = \bot$ by~($\neg_\bot$).
\end{proof}

\begin{theorem}\label{thm_sound}
  If $\Gamma \proves_{\I} \varphi$ then $\Gamma \models_\I \varphi$,
  where $\I = \I_s$ or $\I = e\I_s$.
\end{theorem}

\begin{proof}
  Induction on the length of derivation. All cases follow easily from
  appropriate conditions in the definition of an $\I_s$-model or from
  Lemma~\ref{lem_I_s_model_conditions}. With the rule~$e_f$ we also
  need to use Lemma~\ref{lem_I_s_ef_condition}.

  For instance, suppose $\Gamma \proves
  \Eq{(\Cond{\varphi}{t_1}{t_2})}{(\Cond{\varphi}{t_1'}{t_2})}$ was
  obtained by rule~$c_3$. By the inductive hypothesis we have $\Gamma,
  \varphi \models \Eq{t_1}{t_1'}$ and $\Gamma \models \varphi :
  \Prop$. Assume $\Mc, u \models \Gamma$. We have $\Mc, u \models
  \varphi : \Prop$, so $\valuation{\varphi : \Prop}{u}{\Mc} =
  \top$. By condition~(pr) we obtain $\valuation{\varphi}{u}{\Mc} \in
  \{\top, \bot\}$. Suppose $\valuation{\varphi}{u}{\Mc} = \top$. Then
  $\Mc, u \models \Gamma, \varphi$, so
  $\valuation{\Eq{t_1}{t_1'}}{u}{\Mc} = \top$. Hence by
  conditions~(app), (eq) and~(c1), we have
  $\valuation{\Cond{\varphi}{t_1}{t_2}}{u}{\Mc} =
  \valuation{\Cond{\varphi}{t_1'}{t_2}}{u}{\Mc}$. By conditions~(app)
  and~(eq) we obtain $\Mc, u \models
  \Eq{(\Cond{\varphi}{t_1}{t_2})}{(\Cond{\varphi}{t_1'}{t_2})}$. So
  suppose $\valuation{\varphi}{u}{\Mc} = \bot$. Then we obtain the
  thesis by applying conditions~(app), (c2) and~(eq).

  Other cases are established in a similar way.
\end{proof}

\begin{conjecture}
  If $\Gamma \models_{e\I_s} \varphi$ then $\Gamma \proves_{e\I_s}
  \varphi$.
\end{conjecture}

We do not attempt to prove the above conjecture in this paper, as it
is not necessary for establishing consistency of~$e\I_s$, which is our
main concern here. We treat the semantics given merely as a technical
device in the consistency proof.

\newpage
\renewcommand{\thesection}{Appendix~\Alph{section}}
\tocless\section{Model Construction}\label{sec_construction}
\renewcommand{\thesection}{\Alph{section}}
\addcontentsline{toc}{section}{\thesection\hspace{0.5em} Model Construction}

In this appendix we construct a nontrivial $e\I_s$-model, thus
establishing consistency of the systems~$\I_s$ and~$e\I_s$. The
construction is parametrized by a set of constants~$\delta$. We will
use this construction in the next appendix to show a complete
translation of classical first-order logic into~$\I_s$ and~$e\I_s$.

The construction is an adaptation and extension of the one
from~\cite{Czajka2013Accepted} for the traditional illative
system~$\I_\omega$. The proof is perhaps a bit easier to understand,
because we do not have to deal with certain oddities of traditional
illative systems from~\cite{Czajka2013Accepted}, like e.g. the fact
that $t : \Prop$ is defined as equivalent to $(\lambda x \,.\, t) :
\Type$ where $x \notin FV(t)$, using the notation of the present
paper.

For two sets~$\tau_1$ and~$\tau_2$, we denote by $\tau_2^{\tau_1}$ the
set of all (set-theoretical) functions from~$\tau_1$
to~$\tau_2$. Formally, $f \in \tau_2^{\tau_1}$ is a subset of $\tau_1
\times \tau_2$ such that for every $f_1 \in \tau_1$ there is exactly
one~$f_2 \in \tau_2$ such that $\pair{f_1}{f_2} \in f$.

\begin{definition}\label{def_canonical} \rm
  The subset $\T_I$ of $\I_s$-terms is defined inductively by the
  following rule:
  \begin{itemize}
  \item for every $c_i^\iota \in \Cc$ of arity~$n_i$, if
    $t_1,\ldots,t_{n_i} \in \T_I$ then $c_i^\iota t_1 \ldots t_{n_i}
    \in \T_I$.
  \end{itemize}
  The set~$\Kc$ consists of unique fresh constants, one for each
  element of~$\T_I$. By
  $\bar{c}_i^\iota(\bar{t}_1,\ldots,\bar{t}_{n_i}) \in \Kc$ we denote
  the constant corresponding to $c_i^\iota t_1 \ldots t_{n_i} \in
  \Kc$. Each basic inductive type $\iota \in \Tc_I$ determines in an
  obvious way a subset of~$\T_I$ (the set of all terms generated by
  the constructors of~$\iota$, respecting argument types), and thus it
  determines a subset $\bar{\iota} \subseteq \Kc$. Let $\bar{\Tc}_I =
  \{\bar{\iota} \;|\; \iota \in \Tc_I\}$. Note that $\bigcup
  \bar{\Tc}_I = \Kc$ and $\bar{\iota_1} \cap \bar{\iota_2} =
  \emptyset$ if $\iota_1 \ne \iota_2$.

  Let $\Bool = \{\true, \false\}$ and $\Tc_0 = \{\delta, \Bool,
  \emptyset\} \cup \bar{\Tc}_I$, where~$\true$ and~$\false$ are fresh
  constants. We define~$\Tc_{n+1}$ as follows.
  \begin{itemize}
  \item If $\tau \in \Tc_n$ and $\tau' \subseteq \tau$ then $\tau' \in
    \Tc_{n+1}$.
  \item If $\tau_1, \tau_2 \in \Tc_n$ then $\tau_2^{\tau_1} \in
    \Tc_{n+1}$.
  \end{itemize}
  The set of \emph{types} is now defined by $\Tc =
  \bigcup_{n\in\Nbb}\Tc_n$.

  We define a notation $\tau^{(\tau_1,\ldots,\tau_n)}$ inductively as
  follows:
  \begin{itemize}
  \item $\tau^{(\tau_1)} = \tau^{\tau_1}$,
  \item $\tau^{(\tau_1,\tau_2,\ldots,\tau_n)} =
    (\tau^{(\tau_2,\ldots,\tau_n)})^{\tau_1}$.
  \end{itemize}

  We define the set of \emph{canonical constants} as $\Sigma = \delta
  \cup \Bool \cup \Kc \cup \Sigma_f$, where~$\Sigma_f$ contains a
  unique fresh constant for each function in $\bigcup \Tc$. We denote
  the function corresponding to a constant $c \in \Sigma_f$ by
  $\Fc(c)$. To save on notation we often confuse constants $c \in
  \Sigma_f$ with their corresponding functions, and write e.g. $c \in
  \tau$ instead of $\Fc(c) \in \tau$, for $\tau \in \Tc$. It is always
  clear from the context what we mean. Note that if $\Fc(c) \in
  \tau_2^{\tau_1}$ then~$\tau_1$ is uniquely determined. This is
  because~$\Fc(c)$ is a set-theoretical function, i.e. a set of pairs,
  so its domain is uniquely determined. Note also that if $\Fc(c) \in
  \tau_2^{\tau_1}$ and $\Fc(c) \in {\tau_2'}^{\tau_1}$, then $\Fc(c)
  \in {(\tau_2 \cap \tau_2')}^{\tau_1}$.
\end{definition}

Let $v, w \in \Sigma^*$. We write $v \sqsubseteq w$ if $w = vu$ for
some $u \in \Sigma^*$, i.e. when~$v$ is a prefix of~$w$. We use the
notation $v \sqsubset w$ when $v \sqsubseteq w$ and $v \ne w$.

\begin{definition} \rm
  A \emph{$\Sigma$-tree} $T$ is a set of strings over the
  alphabet~$\Sigma$, i.e. a subset of~$\Sigma^*$, satisfying the
  following conditions:
  \begin{itemize}
  \item if $w \in T$ and $v \sqsubseteq w$ then $v \in T$
    (prefix-closedness),
  \item $\sqsubseteq^{-1}$ is well-founded on~$T$ (no infinite
    branches).
  \end{itemize}
  A \emph{node} of a $\Sigma$-tree~$T$ is any $w \in T$. We say that a
  node $w \in T$ is a \emph{leaf} if there is no $w' \in T$ such that
  $w \sqsubset w'$. If $w \in T$ is not a leaf, then it is an
  \emph{internal node}. The \emph{root} of a $\Sigma$-tree is the
  empty string~$\epsilon$.

  We say that~$T_1$ is a \emph{subtree} of~$T_2$ if there exists $w
  \in T_2$ such that $T_1 = \{v \in \Sigma^* \;|\; wv \in T_2\}$.

  Note that a relation~$\le$, defined by $T_1 \le T_2$ iff~$T_1$ is a
  subtree of~$T_2$, is a well-founded partial order, because
  $\Sigma$-trees have no infinite branches. This allows us to perform
  \emph{induction on the structure of a $\Sigma$-tree}. We write $T_1
  < T_2$ if $T_1 \le T_2$ and $T_1 \ne T_2$.

  The \emph{height}~$h(T)$ of a $\Sigma$-tree~$T$ is an ordinal
  defined by induction on the structure of~$T$.
  \begin{itemize}
  \item If $T = \emptyset$ then $h(T) = 0$.
  \item If $T \ne \emptyset$ then $h(T) = \sup_{T' < T} (h(T') + 1)$.
  \end{itemize}
\end{definition}

\begin{definition} \rm
  The set of constants~$\Sigma^+$ is defined as
  \[
  \Sigma^+ = \Sigma \cup \Tc \cup \{\forall, \vee, \mathrm{Is},
  \mathrm{Subtype}, \mathrm{Fun}, \mathrm{Eq}, \mathrm{Cond},
  \mathrm{Choice}, \Type\} \cup \Cc \cup \Dc \cup \Oc
  \]
  where $\Cc$, $\Dc$ and $\Oc$ are the sets of constructors,
  destructors and tests from the definition of~$\I_s$. Note that for
  each $\tau \in \Tc$ we have $\tau \in \Sigma^+$ as a constant.

  We use~$V^+$ to denote a set of variables of cardinality at least
  the cardinality of~$\Sigma^*$.

  The set of \emph{operation symbols} $\Op$ is defined to
  contain the following:
  \begin{itemize}
  \item $\cdot \in \Op$,
  \item if $x \in V^+$ then $\lambda x \in \Op$,
  \item if $\tau \in \Tc$ and $\tau \ne \emptyset$ then $\All \tau \in
    \Op$ and $\Set \tau \in \Op$.
  \end{itemize}
\end{definition}

Intuitively, $\All \tau$ means ``for all elements of~$\tau$ satisfying
\ldots'', and $\Set \tau$ means ``a subtype of~$\tau$ consisting of
elements satisfying \ldots''. These will appear as node labels in a
$\Sigma$-tree representing a semantic term. A node labelled with
e.g. $\Set \tau$ will have a child corresponding to each element
of~$\tau$. The subtype represented by this node will consist of those
elements for which the corresponding child reduces to~$\true$.

\begin{definition} \rm
  A \emph{semantic term} is a pair $\pair{\Pos}{\kappa}$, where~$\Pos$ is a
  $\Sigma$-tree and $\kappa : \Pos \to \Op \cup \Sigma^+ \cup
  V^+$ is a function such that:
  \begin{itemize}
  \item $w \in \Pos$ is a leaf iff $\kappa(w) \in \Sigma^+ \cup V^+$,
  \item if $\kappa(w) = \lambda x$ then $w0 \in \Pos$ and $wc \notin \Pos$
    for $c \ne 0$,
  \item if $\kappa(w) = \cdot$ then $w0, w1 \in \Pos$ and $wc \notin
    \Pos$ for $c \notin \{0,1\}$,
  \item if $\kappa(w) \in \{ \All \tau, \Set \tau\}$ then $wc \in
    \Pos$ iff $c \in \tau$.
  \end{itemize}

  In other words, semantic terms are possibly infinitely branching
  well-founded trees, whose internal nodes are labelled with operation
  symbols~$\Op$, and leaves are labelled with constants
  from~$\Sigma^+$ or variables from~$V^+$.

  We usually denote semantic terms by $t$, $t_1$, $t_2$, etc. We
  write~$\Pos(t)$ for the underlying tree of~$t$, and~$\upos{t}{p}$
  instead of~$\kappa(p)$.

  The \emph{height} of a semantic term~$t$, denoted~$h(t)$, is the
  height of its associated $\Sigma$-tree. When we say that we perform
  \emph{induction on the structure of a semantic term} we mean
  induction on~$h(t)$. By induction on an ordinal~$\alpha$ and the
  structure of a semantic term we mean induction on
  pairs~$\pair{\alpha}{h(t)}$ ordered lexicographically.

  A \emph{position} in a semantic term~$t$ is a string $w \in
  \Pos(t)$. The \emph{subterm} of~$t$ at position~$p \in \Pos(t)$,
  denoted $\pos{t}{p}$, is a semantic term $\pair{\Pos'}{\kappa'}$
  where:
  \begin{itemize}
  \item $\Pos' = \{ w \in \Sigma^* \;|\; pw \in \Pos(t) \}$,
  \item $\kappa'(w) = \kappa(pw)$.
  \end{itemize}

  A variable $x \in V^+$ is \emph{free} in a semantic term~$t$ if
  there exists $p \in \Pos(t)$ such that $\upos{t}{p} = x$ and for no
  $p' \sqsubseteq p$ we have $\upos{t}{p} = \lambda x$. A variable is
  \emph{bound} if it is not free.

  We identify $\alpha$-equivalent semantic terms, i.e. ones differing
  only in the names of bound variables. We assume that no semantic
  term contains some variable both free and bound. We use the
  symbol~$\equiv$ for identity of semantic terms up to
  $\alpha$-equivalence.

  Substitution $t[x/t']$ for semantic terms is defined in an obvious
  way, avoiding variable capture. In other words, we adopt the
  convention that whenever we write a term of the form~$t[x/t']$ we
  assume that no free variables of~$t'$ become bound in~$t[x/t']$.
\end{definition}

In this section, when we speak of terms we mean semantic terms, unless
otherwise stated. We often use abbreviations for semantic terms of the
form $\lambda x \,.\, t$, $t_1 t_2$, $\lambda x
. (\mathrm{Fun}\,{x}\,{(t_1 t_2)})$, etc. The meaning of these
abbreviations is obvious.

\begin{definition} \rm
  A \emph{rewriting system}~$R$ is a set of pairs of semantic
  terms. We usually write $\trsrule{t_1}{t_2} \in R$ instead of
  $\pair{t_1}{t_2} \in R$. A term~$t$ is said to \emph{$R$-contract}
  to a term~$t'$ at position~$p$, denoted $t \contr_R^p t'$, if $p \in
  \Pos(t) \cap \Pos(t')$, the terms~$t$ and~$t'$ differ only in
  subterms at position~$p$, and there exists $\trsrule{t_1}{t_2} \in
  R$ such that $\pos{t}{p} \equiv t_1$ and $\pos{t'}{p} \equiv
  t_2$. We write $t \contr_R t'$ if $t \contr_R^p t'$ for some $p \in
  \Pos(t)$.

  For each ordinal~$\alpha$, we define two relations
  $\ipcontr_R^\alpha$ and $\succ_R^\alpha$ by induction on the
  ordinal~$\alpha$ and the structure of~$t$.
  \begin{itemize}
  \item[(a)] If $c \in \Sigma$ then $c \succ_R^\alpha c$.
  \item[(b)] If $c \in \tau_2^{\tau_1}$ and for all $c_1 \in \tau_1$
    there exists~$t'$ such that $t c_1 \udipreduces{<\alpha}{R} t'
    \succ_R^{<\alpha} \Fc(c)(c_1)$, then $t \succ_R^\alpha c$.
  \item[(1)] If $t \equiv t'$ or $t \contr_R^\epsilon t'$ then $t
    \ipcontr_R^\alpha t'$.
  \item[(2)] If $t_1 \ipcontr_R^\alpha t_1'$ and $t_2
    \ipcontr_R^\alpha t_2'$ then $t_1 t_2 \ipcontr_R^\alpha t_1'
    t_2'$.
  \item[(3)] If $t \ipcontr_R^\alpha t'$ then $\lambda x \,.\, t
    \ipcontr_R^\alpha \lambda x \,.\, t'$.
  \item[(4)] If $t_1 \ipcontr_R^{\alpha} t_1'$ and $t_2
    \ipcontr_R^{\alpha} t_2'$ then $(\lambda x \,.\, t_1) t_2
    \ipcontr_R^\alpha t_1'[x/t_2']$.
  \item[(5)] If $c \in \Sigma^+$, $n \in \Nbb$, $c t_1' \ldots t_n'
    \contr_R^\epsilon t$ and $t_i \ipcontr_R^\alpha t_i'$ for
    $i=1,\ldots,n$, then $c t_1 \ldots t_n \ipcontr_R^\alpha t$.
  \item[(6)] If $t \succ_R^{<\alpha} c_1 \in \tau_1$ and $c \in
    \tau_2^{\tau_1}$, then $c t \ipcontr_R^\alpha \Fc(c)(c_1)$.
  \item[(7)] If $t \succ_R^{<\alpha} c$ for some $c \in \tau \in \Tc$
    then $\Is{t}{\tau} \ipcontr_R^\alpha \true$.
  \item[(8)] If $\upos{t}{\epsilon} \equiv \upos{t'}{\epsilon} \equiv
    \All \tau$ or $\upos{t}{\epsilon} \equiv \upos{t'}{\epsilon}
    \equiv \Set \tau$, and for all $c \in \tau$ there exists~$t''$
    such that $\pos{t}{c} \ipcontr_R^\alpha t'' \udipreduces{<\alpha}{R}
    \pos{t'}{c}$, then $t \ipcontr_R^\alpha t'$.
  \end{itemize}
  The notation~$\ipcontr_R^{<\alpha}$ is an abbreviation
  for~$\bigcup_{\gamma<\alpha}\ipcontr_R^{\gamma}$, and
  $\udipreduces{<\alpha}{R}$ denotes the tran\-si\-tive-re\-flex\-ive
  closure of~$\ipcontr_R^{<\alpha}$. The notation~$\succ_R^{<\alpha}$
  is an abbreviation for~$\bigcup_{\gamma<\alpha}\succ_R^\gamma$.

  We define the relations~$\ipcontr_R$ and~$\succ_R$ as the smallest
  fixpoint of the above construction, i.e. by monotonicity of the
  definition there exists the least ordinal~$\zeta$ such
  that~$\ipcontr_R^\zeta \,=\, \ipcontr_R^{<\zeta}$ and $\succ_R^\zeta
  \,=\, \succ_R^{<\zeta}$, and we take~$\ipcontr_R \,=\,
  \ipcontr_R^\zeta$, $\succ_R \,=\, \succ_R^\zeta$. Note that~$\omega$
  steps do not suffice to reach the fixpoint. In fact, the
  ordinal~$\zeta$ will be quite large.

  We denote by~$\ipreduces_R$ the transitive-reflexive closure
  of~$\ipcontr_R$, and by~$\ipeqvred_R$ the
  transitive-reflexive-symmetric closure of~$\ipcontr_R$. The
  subscript is often dropped when obvious from the context.
\end{definition}

Notice that the relation~$\ipcontr_R$ encompasses an analogon of
$\beta$-reduction, regardless of what the rules of~$R$
are. Intuitively, the relation~$\ipcontr_R$ is a kind of parallel
reduction on semantic terms, parametrized by the rules of~$R$.

\begin{lemma}
  If $R \subseteq R'$ then $\ipcontr_R \,\subseteq\,\ipcontr_{R'}$ and
  $\succ_R \,\subseteq\,\succ_{R'}$.
\end{lemma}

Intuitively, and very informally, $t \succ_R c$ is intended to hold
if~$c \in \tau \in \Tc$ is a ``canonical'' object which
``simulates''~$t$ in type~$\tau$. By~$c$ ``simulating''~$t$ in
type~$\tau$ we mean that~$c$ behaves in essentially the same way
as~$t$, modulo~$\ipcontr_R$, whenever a term of type~$\tau$ is
``expected''. Let us give some examples to elucidate what we mean by
this. For instance, let $c_1 \in \Nat^\Nat = \tau_1$, $c_2 \in
\delta^\delta = \tau_2$ be two constants such that $\Fc(c_1)(c) \equiv
c$ for all $c \in \Nat$ and $\Fc(c_2)(c) \equiv c$ for all $c \in
\delta$. Note that by condition~(6) in the definition of~$\ipcontr_R$
and the fact that $c_1 \succ_R c_1$ and $c_2 \succ_R c_2$ we have $c_1
c \ipcontr_R c$ for all $c \in \Nat$ and $c_2 c \ipcontr_R c$ for all
$c \in \delta$. Now we have both $\lambda x \,.\, x \succ_R c_1$ and
$\lambda x \,.\, x \succ_R c_2$, because $\lambda x \,.\, x$ behaves
exactly like~$c_1$ when given arguments of type~$\Nat$, and exactly
like~$c_2$ when given arguments of type~$\delta$. The condition~(6)
ensures that $\lambda x . x$ and~$c_1$ will be indistinguishable
wherever a term of type~$\tau_1$ is ``expected''. For instance, if $d
\in \delta^{\tau_1}$ then we have $d (\lambda x \,.\, x) \ipcontr_R
\Fc(d)(c_1)$. In fact, we will later prove that, for an appropriate
rewriting system~$R$, the conditions $t c \ipreduces_R t' \succ_R c'$
and $r \succ_R c$ imply the existence of~$t''$ such that $t r
\ipreduces_R t'' \succ_R c'$, where~$t$ is an arbitrary term.

Note that we may have $c_1 \succ_R c_2$ with $c_1, c_2 \in \Sigma$,
$c_1 \not\equiv c_2$ (i.e.~$c_1$ and~$c_2$ being two distinct
canonical constants), if there does not exist a single $\tau \in \Tc$
such that $c_1,c_2 \in \tau$. We will later show that this is not
possible, for a rewriting system~$R$ to be defined, if such a~$\tau
\in \Tc$ does exist.

We will build our model from equivalence classes of~$\ipeqvred_R$ on
semantic terms, for a certain rewriting system~$R$ to be defined
below. One of the main problems in the model construction is to ensure
that the condition~$(\forall_e)$ of Definition~\ref{def_I_s_model}
holds. The problem is that condition~$(\to_i)$ needs to be satisfied
as well, which means that we cannot know a priori which terms~$t$
should satisfy~$\Is{t}{\tau}$ for a function type~$\tau \in \Tc$,
because this must depend on the definition of~$\ipcontr_R$
for~$(\to_i)$ to hold. And we cannot use a conditional rule of the
form
\begin{itemize}
\item[] if for all~$t$ such that $\Is{t}{\tau} \ipreduces_R \true$ we
  have $t_1 t \ipreduces_R \true$ then $\forall\,\tau\,t_1 \ipcontr_R
  \true$
\end{itemize}
because the definition would not be monotone. Our solution is to
restrict quantification to canonical constants only, and to define the
relation~$\ipcontr_R$ in such a way as to ensure that for each
term~$t$ with $\Is{t}{\tau} \ipreduces_R \true$ there exist a canonical
constant $c \in \tau$ and a term~$t'$ such that $t \ipreduces_R t'
\succ c$.

\begin{definition} \rm
  Let $\chi$ be a choice function for the family of sets~$\Tc
  \setminus \{\emptyset\}$. We define a rewriting system~$R$ by the
  rules presented in Fig.~\ref{fig_1}.
  \begin{figure}[p]
  \begin{eqnarray*}
    \Eq{t_1}{t_2} &\to& \true \;\;\; \mathrm{if\ } t_1 \ipeqvred_R t_2 \\
    \forall \tau t &\to& t' \;\;\; \mathrm{where\ } \tau \in \Tc,\,
    \tau \ne \emptyset,\, \upos{t'}{\epsilon} \equiv \All \tau \mathrm{\ and\ }
    \pos{t'}{c} \equiv t c \mathrm{\ for\ } c \in \tau \\
    \mathrm{Subtype}\, \tau\, t &\to& t' \;\;\; \mathrm{where\ }
    \tau \in \Tc,\, \tau \ne \emptyset,\, \upos{t'}{\epsilon} \equiv \Set \tau
    \mathrm{\ and\ } \pos{t'}{c} \equiv t c \mathrm{\ for\ } c \in
    \tau \\
    \forall \emptyset t &\to& \true \\
    \mathrm{Subtype}\, \emptyset\, t &\to& \emptyset \\
    t &\to& \true \;\;\; \mathrm{if\ } \upos{t}{\epsilon} \equiv \All
    \tau \mathrm{\ and\ for\ all\ } c \in \tau \mathrm{\ we\ have\ } \pos{t}{c} \equiv \true \\
    t &\to& \false \;\;\; \mathrm{if\ } \upos{t}{\epsilon} \equiv \All \tau \mathrm{\ and\ there\ exists\ } c \in
    \tau \mathrm{\ such\ that\ } \pos{t}{c} \equiv \false \\
    t &\to& \tau' \;\;\; \mathrm{if\ } \upos{t}{\epsilon} \equiv \Set
    \tau \mathrm{\ and\ for\ all\ } c \in \tau  \mathrm{\ we\ have\ }
    \pos{t}{c} \in \{\true,\false\}, \\ &&
    \;\;\;\;\;\;\;\;\mathrm{\ and\ } \tau' = \{ c \in \tau \;|\; \pos{t}{c} \equiv \true \} \\
    \mathrm{Fun}\,\tau_1\, \tau_2 &\to& \tau_2^{\tau_1} \;\;\; \mathrm{for\ } \tau_1, \tau_2 \in \Tc \\
    \Cond{\true}{t_1}{t_2} &\to& t_1 \\
    \Cond{\false}{t_1}{t_2} &\to& t_2 \\
    \vee \true t &\to& \true \\
    \vee t \true &\to& \true \\
    \vee \false \false &\to& \false \\
    c_i^\iota \bar{t}_1 \ldots \bar{t}_{n_i} &\to& \bar{c}_i^\iota(\bar{t}_1, \ldots, \bar{t}_{n_i})
    \;\;\; \mathrm{if\ } \bar{t}_1, \ldots, \bar{t}_{n_i} \in \bar{\iota} \in \bar{\Tc}_I \\
    d_{i,j}^\iota (\bar{c}_i^\iota(\bar{t}_1,\ldots,\bar{t}_{n_i})) &\to& \bar{t}_j \\
    d_{i,j}^\iota (c_i^\iota t_1 \ldots t_{n_i}) &\to& t_j \\
    o_i^\iota (\bar{c}_i^\iota(\bar{t}_1, \ldots, \bar{t}_{n_i})) &\to& \true \\
    o_i^\iota (\bar{c}_k^\iota(\bar{t}_1, \ldots, \bar{t}_{n_k})) &\to& \false \;\;\; \mathrm{if\ } i \ne k \\
    o_i^\iota (c_i^\iota t_1 \ldots t_{n_i}) &\to& \true \\
    o_i^\iota (c_k^\iota t_1 \ldots t_{n_k}) &\to& \false \;\;\; \mathrm{if\ } i \ne k \\
    \mathrm{Choice}\, \tau &\to& \chi(\tau) \;\;\; \mathrm{if\ } \tau \in \Tc,\, \tau \ne \emptyset \\
    \Is{\tau}{\Type} &\to& \true \;\;\; \mathrm{for\ } \tau \in \Tc
  \end{eqnarray*}
  \caption{The rules of the rewriting system~$R$.}
  \label{fig_1}
  \end{figure}
\end{definition}

The above is a circular definition because the condition in the rule
for~$\mathrm{Eq}$ refers to the system~$R$. Note, however, that this
reference is positive. Formally, we may therefore define a progression
of rewrite systems~$R_\alpha$ consisting of the above rules, but each
using as~$R$ the system~$R_{<\alpha} =
\bigcup_{\gamma<\alpha}R_\gamma$. We note that $R_\alpha \subseteq
R_\beta$ for $\alpha \le \beta$ and take~$R$ to be the least fixpoint.

From now on~$R$ refers to the rewriting system defined above, and the
relations~$\contr$, $\ipcontr$, $\ipreduces$, etc. refer
to~$\contr_R$, $\ipcontr_R$, $\ipreduces_R$, etc. We write $t_1
\contr_\equiv^p t_2$ if $t_1 \contr^p t_2$ or $t_1 \equiv t_2$.

The three simple lemmas below follow by an easy inspection of the
definition of~$R$ and of~$\ipcontr$.

\begin{lemma}
  If $c \ipcontr t$ where $c \in \Sigma^+$, then $t \equiv c$ or $c
  \in \Cc$ has arity~$0$.
\end{lemma}

\begin{lemma}\label{lem_contr_epsilon}
  If $c \in \Sigma^+$, there exist $s_1,\ldots,s_n$ such that $c s_1
  \ldots s_n$ is a left side of a rule in~$R$, and $c t_1 \ldots t_n
  \ipcontr^\alpha t$, then there exist $t_1',\ldots,t_n'$ such that
  $t_i \ipcontr^\alpha t_i'$ for $i=1,\ldots,n$ and $c t_1' \ldots
  t_n' \contr_\equiv^\epsilon t$.
\end{lemma}

\begin{lemma}\label{lem_orthogonal}
  If $t \contr^\epsilon t_1$ and $t \contr^\epsilon t_2$ then $t_1
  \equiv t_2$.
\end{lemma}

\begin{lemma}\label{lem_subst_1}
  The following conditions hold.
  \begin{itemize}
  \item If $t_1 \contr_{R_\beta}^{\epsilon} t_1'$ then $t_1[x/t_2]
    \contr_{R_\beta}^{\epsilon} t_1'[x/t_2]$.
  \item If $t_1 \ipcontr_{R_\beta}^{\alpha} t_1'$ then $t_1[x/t_2]
    \ipcontr_{R_\beta}^{\alpha} t_1'[x/t_2]$.
  \item If $t_1 \succ_{R_\beta}^\alpha c$ then $t_1[x/t_2]
    \succ_{R_\beta}^\alpha c$.
  \end{itemize}
\end{lemma}

\begin{proof}
  Induction on triples~$\langle \beta, \alpha, h(t_1) \rangle$ ordered
  lexicographically.

  We first show that $t_1 \contr_{R_\beta}^{\epsilon} t_1'$ implies
  $t_1[x/t_2] \contr_{R_\beta}^{\epsilon} t_1'[x/t_2]$. The only
  non-obvious case is when $t_1 \equiv \Eq{r_1}{r_2}
  \contr_{R_\beta}^\epsilon \true \equiv t_1'$ by virtue of~$r_1
  \ipeqvred_{R_{<\beta}} r_2$. But then by the inductive hypothesis
  (for smaller~$\beta$) we obtain $r_1[x/t_2] \ipeqvred_{R_{<\beta}}
  r_2[x/t_2]$. This implies that $t_1[x/t_2] \equiv
  \Eq{r_1[x/t_2]}{r_2[x/t_2]} \contr_{R_\beta}^\epsilon \true \equiv
  t_1'[x/t_2]$.

  We show that $t_1 \ipcontr_{R_\beta}^{\alpha} t_1'$ implies
  $t_1[x/t_2] \ipcontr_{R_\beta}^{\alpha} t_1'[x/t_2]$. If $t_1 \equiv
  t_1'$ then this is obvious. If $t_1 \contr_{R_\beta}^\epsilon t_1'$
  then this follows from the previous condition. If $t_1 \equiv c r_1
  \ldots r_n \ipcontr_{R_\beta}^\alpha t_1'$, $c \in \Sigma^+$, $c
  r_1' \ldots r_n' \contr_{R_\beta}^\epsilon t_1'$, $r_i
  \ipcontr_{R_\beta}^{\alpha} r_i'$, then $r_i[x/t_2]
  \ipcontr_{R_\beta}^\alpha r_i'[x/t_2]$ by the inductive hypothesis
  (because $h(r_i) < h(t_1)$). As in the previous paragraph, it also
  follows from the IH that $c r_1'[x/t_2] \ldots r_n'[x/t_2]
  \contr_{R_\beta}^\epsilon t_1'[x/t_2]$. Therefore $t_1[x/t_2] \equiv
  c r_1[x/t_2] \ldots r_n[x/t_2] \ipcontr_{R_\beta}^\alpha
  t_1'[x/t_2]$. All other cases follow easily from the inductive
  hypothesis, and we omit them.

  Now we prove that $t_1 \succ_{R_\beta}^\alpha c$ implies $t_1[x/t_2]
  \succ_{R_\beta}^\alpha c$. If $t_1 \equiv c$ then this is
  obvious. Otherwise $c \in \tau_2^{\tau_1} \in \Tc$ and for all $c_1
  \in \tau_1$ there exists~$t_1'$ such that $t_1 c_1
  \udipreduces{<\alpha}{R_\beta} t_1' \succ_{R_\beta}^{<\alpha}
  \Fc(c)(c_1)$. But then by the IH we have $t_1[x/t_2] c_1
  \udipreduces{<\alpha}{R_\beta} t_1'[x/t_2] \succ_{R_\beta}^{<\alpha}
  \Fc(c)(c_1)$. Therefore $t_1[x/t_2] \succ_{R_\beta}^{\alpha} c$.
\end{proof}

\begin{lemma}\label{lem_subst_2}
  If $t_2 \ipcontr^{\alpha} t_2'$ then $t_1[x/t_2] \ipcontr^{\alpha}
  t_1[x/t_2']$.
\end{lemma}

\begin{proof}
  Induction on the structure of~$t_1$. If $t_1 \equiv x$ or $x \notin
  FV(t_1)$ then this is obvious. If $\upos{t_1}{\epsilon} \in \{\All
  \tau, \Set \tau\}$ then by the inductive hypothesis
  $\pos{t_1}{c}[x/t_2] \ipcontr^{\alpha} \pos{t_1}{c}[x/t_2']$ for $c
  \in \tau$. Therefore $t_1[x/t_2] \ipcontr^\alpha t_1[x/t_2']$. Other
  cases follow directly from the inductive hypothesis in a similar
  fashion.
\end{proof}

\begin{lemma}\label{lem_subst}
  If $t_1 \ipcontr^\alpha t_1'$ and $t_2 \ipcontr^\alpha t_2'$ then
  $t_1[x/t_2] \ipcontr^{\alpha} t_1'[x/t_2']$.
\end{lemma}

\begin{proof}
  Induction on~$\alpha$ and the structure of~$t_1$.

  If $t_1 \equiv t_1'$ then the claim follows from
  Lemma~\ref{lem_subst_2}. If $t_1 \contr^\epsilon t_1'$ then we
  consider possible forms of~$t_1$. Suppose $t_1 \equiv \Eq{r_1}{r_2}
  \contr^\epsilon \true \equiv t_1'$ by virtue of $r_1 \ipeqvred
  r_2$. By Lemma~\ref{lem_subst_1} we have $r_1[x/t_2] \ipeqvred
  r_2[x/t_2]$. Hence $t_1[x/t_2] \equiv \Eq{r_1[x/t_2]}{r_2[x/t_2]}
  \contr^\epsilon \true \equiv t_1'[x/t_2']$. Suppose $t_1 \equiv
  \forall \tau t \contr^\epsilon t_1'$ where $\tau \in \Tc$,
  $\upos{t_1'}{\epsilon} \equiv \All\tau$, and for all $c \in \tau$ we
  have $\pos{t_1'}{c} \equiv t c$. Then $t[x/t_2] \ipcontr^\alpha
  t[x/t_2']$ by the inductive hypothesis, and $\forall \tau
  (t[x/t_2']) \contr^\epsilon t_1'[x/t_2']$ by
  Lemma~\ref{lem_subst_1}. Hence by condition~(5) in the definition
  of~$\ipcontr^\alpha$ we obtain $t_1[x/t_2] \equiv \forall \tau
  (t[x/t_2]) \ipcontr^\alpha t_1'[x/t_2']$. Other cases are
  established in a similar manner.

  If $t_1 \equiv (\lambda y \,.\, r_1) r_2 \ipcontr^\alpha r_1'[y/r_2']
  \equiv t_1'$ where $x \not\equiv y$, $r_1 \ipcontr^\alpha r_1'$ and
  $r_2 \ipcontr^\alpha r_2'$, then by the inductive hypothesis
  $r_1[x/t_2] \ipcontr^\alpha r_1'[x/t_2']$ and $r_2[x/t_2]
  \ipcontr^\alpha r_2'[x/t_2']$. Recall that by our implicit
  assumption that in~$t_1[x/t_2]$ no free variables of~$t_2$ become
  bound, we have $y \notin FV(t_2)$, and hence $y \notin
  FV(t_2')$. Thus $t_1[x/t_2] \equiv (\lambda y \,.\, r_1[x/t_2])
  r_2[x/t_2] \ipcontr^\alpha r_1'[x/t_2'][y/(r_2'[x/t_2'])] \equiv
  (r_1'[y/r_2'])[x/t_2'] \equiv t_1'[x/t_2']$.

  If $t_1 \equiv c r_1 \ldots r_n \ipcontr^\alpha t_1'$ where $c \in
  \Sigma^+$, $r_i \ipcontr^\alpha r_i'$ for $i=1,\ldots,n$, and $c
  r_1' \ldots r_n' \contr^\epsilon t_1'$, then $c r_1'[x/t_2'] \ldots
  r_n'[x/t_2'] \contr^\epsilon t_1'[x/t_2']$ by
  Lemma~\ref{lem_subst_1} and $r_i[x/t_2] \ipcontr^\alpha
  r_i'[x/t_2']$ by the inductive hypothesis. We thus conclude
  $t_1[x/t_2] \equiv c r_1[x/t_2] \ldots r_n[x/t_2] \ipcontr^\alpha
  t_1'[x/t_2']$.

  If $t_1 \equiv \Is{t}{\tau} \ipcontr^\alpha \true \equiv t_1'$ where
  $t \succ^{<\alpha} c$ for some $c \in \tau$, then $t[x/t_2]
  \succ^{<\alpha} c$ by Lemma~\ref{lem_subst_1}. Therefore $t_1[x/t_2]
  \equiv \Is{t[x/t_2]}{\tau} \ipcontr^\alpha \true \equiv
  t_1'[x/t_2']$.

  If $\upos{t_1}{\epsilon} \equiv \upos{t_1'}{\epsilon} \equiv
  \All\tau$ and for all $c \in \tau$ there exists~$t_c$ such that
  $\pos{t_1}{c} \ipcontr^\alpha t_c \uipreduces{<\alpha}
  \pos{t_1'}{c}$, then $\pos{t_1}{c}[x/t_2] \ipcontr^\alpha
  t_c[x/t_2']$ by the inductive hypothesis. By Lemma~\ref{lem_subst_1}
  we obtain $t_c[x/t_2'] \uipreduces{<\alpha}
  \pos{t_1'}{c}[x/t_2']$. This implies that $t_1[x/t_2]
  \ipcontr^\alpha t_1'[x/t_2']$.

  Other cases follow by analogous proofs.
\end{proof}

\begin{lemma}\label{lem_contr_stable}
  If $c \in \Sigma^+$, $n \in \Nbb$, $c t_1 \ldots t_n \contr^\epsilon
  t$ and $t_i \ipcontr^\alpha t_i'$ for $i=1,\ldots,n$, then there
  exists~$t'$ such that $c t_1' \ldots t_n' \contr^\epsilon t'$ and $t
  \ipcontr^\alpha t'$.
\end{lemma}

\begin{proof}
  If $c t_1 t_2 \equiv \Eq{t_1}{t_2} \contr^\epsilon \true$ then $t_1
  \ipeqvred t_2$. Since $\ipcontr^\alpha \,\subseteq\, \ipcontr$, we
  have $t_1' \ipeqvred t_2'$. Thus $\Eq{t_1'}{t_2'} \contr^\epsilon
  \true$.

  Suppose $c t_1 t_2 \equiv \forall \tau t_2 \contr^\epsilon t$
  where $\tau \in \Tc$, $\upos{t}{\epsilon} \equiv \All\tau$,
  $\pos{t}{c} \equiv t_2 c$ for $c \in \tau$. We have $\forall \tau
  t_2' \contr^\epsilon t'$ where $\upos{t'}{\epsilon} \equiv \All\tau$
  and $\pos{t'}{c} \equiv t_2' c$ for $c \in \tau$. Since $t_2 c
  \ipcontr^\alpha t_2' c$, by condition~(8) in the definition
  of~$\ipcontr^\alpha$ we conclude that $t \ipcontr^\alpha t'$.

  If $c t_1 \equiv d_{i,j}^\iota (c_i^\iota r_1 \ldots r_{n_i})
  \contr^\epsilon r_j \equiv t$ then $c_i^\iota r_1 \ldots r_{n_i}
  \ipcontr^\alpha t_1'$. By Lemma~\ref{lem_contr_epsilon} there
  exist~$r_1',\ldots,r_{n_i}'$ such that $r_k \ipcontr^\alpha r_k'$
  for $k=1,\ldots,n_i$ and $c_i^\iota r_1' \ldots r_{n_i}'
  \contr_\equiv^\epsilon t_1'$. If $c_i^\iota r_1' \ldots r_{n_i}'
  \equiv t_1'$ then $c t_1' \equiv d_{i,j}^\iota (c_i^\iota r_1'
  \ldots r_{n_i}') \contr^\epsilon r_j'$ and we are done. Otherwise
  $r_1',\ldots,r_{n_i}'\in\bar{\iota}$ and $t_1' \equiv
  \bar{c}_i^\iota(r_1',\ldots,r_{n_i}')$. Thus $c t_1' \equiv d_{i,j}
  (\bar{c}_i^\iota(r_1',\ldots,r_{n_i}')) \contr^\epsilon r_j'$.

  Other cases are trivial or follow by a similar proof.
\end{proof}

\begin{definition} \rm
  We say that two binary relations on terms~$\to_1$ and~$\to_2$
  \emph{commute} if $t_1 \to_1 t_1'$ and $t_2 \to_2 t_2'$ imply $t_1'
  \to_2^\equiv t_3$ and $t_2' \to_1^\equiv t_3$ for some term $t_3$,
  where $\to_i^\equiv$ is the reflexive closure of~$\to_i$.
\end{definition}

\begin{lemma}\label{lem_hindley_rosen}
  If $\ipcontr^{\alpha'}$ and $\ipcontr^{\beta'}$ commute for all
  $\alpha' \circ_1 \alpha$ and $\beta' \circ_2 \beta$, then
  $\uipreduces{\circ_1\alpha}$ and $\uipreduces{\circ_2\beta}$,
  $\ipcontr^{\circ_1\alpha}$ and $\uipreduces{\circ_2\beta}$, as well
  as $\uipreduces{\circ_1\alpha}$ and $\ipcontr^{\circ_2\beta}$,
  commute. Here $\circ_1,\circ_2 \in \{<, \le\}$ and
  $\ipcontr^{\le\gamma}\,=\,\ipcontr^\gamma$ for $\gamma \in
  \{\alpha,\beta\}$.
\end{lemma}

\begin{proof}
  The proof is a simple tiling argument similar to the proof of the
  Hindley-Rosen lemma, see e.g. \cite[Chapter 3]{Barendregt1984}.
\end{proof}

\begin{lemma}\label{lem_commute}
  For all ordinals $\alpha$, $\beta$ the following
  conditions hold:
  \begin{itemize}
  \item[($\rmnum{1}$)] $\ipcontr^\alpha$ and $\ipcontr^\beta$ commute,
  \item[($\rmnum{2}$)] if $t \succ^\alpha c$ and $t \ipcontr^\beta t'$
    then $t' \succ^\alpha c$,
  \item[($\rmnum{3}$)] if $t \succ^\alpha c_1$, $t \succ^\beta c_2$
    and $c_1, c_2 \in \tau \in \Tc$ then $c_1 \equiv c_2$.
  \end{itemize}
\end{lemma}

\begin{proof}
  The proof is by induction on triples $\langle \alpha, \beta, h(t)
  \rangle$ ordered lexicographically, where in condition~($\rmnum{1}$)
  the term~$t$ is such that $t \ipcontr^\alpha t_1$ and $t
  \ipcontr^\beta t_2$ for some~$t_1, t_2$. Together with
  condition~($\rmnum{2}$) we also prove its dual, i.e. the condition
  with~$\alpha$ and~$\beta$ exchanged. We only give a proof for the
  original condition, but it is easy to see that the dual condition
  follows by exactly the same proof but with~$\alpha$ and~$\beta$
  exchanged.

  We first show condition~($\rmnum{1}$). Assume $t \ipcontr^\alpha
  t_1$ and $t \ipcontr^\beta t_2$. We need to show that there
  exists~$t'$ such that $t_1 \ipcontr^\beta t'$ and $t_2
  \ipcontr^\alpha t'$. It is clear that it suffices to consider only
  the situations when $t \ipcontr^\alpha t_1$ follows by
  condition~($m$) and $t \ipcontr^\beta t_2$ follows by
  condition~($n$) in the definition of~$\ipcontr$ for $m \le n$,
  provided that we never use the inductive hypothesis with~$\beta$
  increased, which is easily verified to be the case. Indeed, then we
  may use exactly the same proofs, but with~$\alpha$ and~$\beta$
  exchanged, to handle the cases when $m > n$.

  If $t \equiv t_1$ or $t \equiv t_2$ then the claim is
  obvious. Suppose $t \ipcontr^\alpha t_1$ follows by condition~(1) in
  the definition of~$\ipcontr^\alpha$. Then $t \contr^\epsilon
  t_1$. By Lemma~\ref{lem_orthogonal} it is impossible that $t
  \ipcontr^\beta t_2$ follows by condition~(1) in the definition
  of~$\ipcontr^\beta$, unless $t_1 \equiv t_2$. Suppose that $t
  \ipcontr^\beta t_2$ follows by condition~(5). Then $t \equiv c r_1
  \ldots r_n \contr^\epsilon t_1$, $r_i \ipcontr^\beta r_i'$ and $c
  r_1' \ldots r_n' \contr^\epsilon t_2$. By
  Lemma~\ref{lem_contr_stable} there exists~$t'$ such that $t_1
  \ipcontr^\beta t'$ and $c r_1' \ldots r_n' \contr^\epsilon t'$. But
  by Lemma~\ref{lem_orthogonal} we have $t' \equiv t_2$. The only
  remaining possibility, when $t \ipcontr^\alpha t_1$ follows by
  condtion~(1), is that $t \ipcontr^\beta t_2$ follows by
  condition~(2). But then the claim follows from
  Lemma~\ref{lem_contr_stable}.

  Suppose $t \equiv r_1 r_2 \ipcontr^\alpha r_1' r_2' \equiv t_1$
  where $r_1 \ipcontr^\alpha r_1'$ and $r_2 \ipcontr^\alpha r_2'$. If
  $t \ipcontr^\beta t_2$ follows by condition~(2) then $t_2 \equiv
  r_1'' r_2''$ where $r_1 \ipcontr^\beta r_1''$ and $r_2
  \ipcontr^\beta r_2''$. By the inductive hypothesis (note that
  $h(r_1), h(r_2) < h(t)$) there exist $q_1$, $q_2$ such that $r_1'
  \ipcontr^\beta q_1$, $r_1'' \ipcontr^\alpha q_1$, $r_2'
  \ipcontr^\beta q_2$ and $r_2'' \ipcontr^\alpha q_2$. Thus $t_1
  \equiv r_1' r_2' \ipcontr^\beta q_1 q_2$ and $t_2 \equiv r_1'' r_2''
  \ipcontr^\alpha q_1 q_2$.

  It is not possible that $t \ipcontr^\beta t_2$ follows by
  condition~(3). If it follows by condition~(4) then $r_1 \equiv
  \lambda x \,.\, s_1$, $r_1' \equiv \lambda x \,.\, s_1'$, $s_1
  \ipcontr^\alpha s_1'$, and $t_2 \equiv s_1''[x/r_2'']$ where $s_1
  \ipcontr^\beta s_1''$, $r_2 \ipcontr^\beta r_2''$. By the inductive
  hypothesis there exist~$q_1$ and~$q_2$ such that $s_1'
  \ipcontr^\beta q_1$, $s_1'' \ipcontr^\alpha q_1$, $r_2'
  \ipcontr^\beta q_2$ and $r_2'' \ipcontr^\alpha q_2$. By
  condition~(4) in the definition of~$\ipcontr^\beta$ we have $t_1
  \equiv (\lambda x \,.\, s_1') r_2' \ipcontr^\beta q_1[x/q_2]$. By
  Lemma~\ref{lem_subst} we obtain $t_2 \equiv s_1''[x/r_2'']
  \ipcontr^\alpha q_1[x/q_2]$.

  If $t \ipcontr^\beta t_2$ follows by condition~(5) then $r_1 \equiv
  c s_1 \ldots s_n$, $s_i \ipcontr^\beta s_i''$, $r_2 \ipcontr^\beta
  r_2''$ and $c s_1'' \ldots s_n'' r_2'' \contr^\epsilon t_2$. By
  inspecting the definition of~$R$ we see that in this case $c q_1
  \ldots q_m \contr^\epsilon r_1'$ is not possible for any
  $q_1,\ldots,q_m$ and any $m \le n$. By inspecting the definition
  of~$\ipcontr^\alpha$ we thus see that $r_1 \ipcontr^\alpha r_1'$ is
  only possible when $r_1' \equiv c s_1' \ldots s_n'$ and $s_i
  \ipcontr^\alpha s_i'$. By the inductive hypothesis there exist
  $q_1,\ldots,q_{n+1}$ such that $s_i'' \ipcontr^\alpha q_i$, $s_i'
  \ipcontr^\beta q_i$ for $i=1,\ldots,n$, and $r_2'' \ipcontr^\alpha
  q_{n+1}$, $r_2' \ipcontr^\beta q_{n+1}$. Therefore $c s_1' \ldots
  s_n' r_2' \ipcontr^\beta c q_1 \ldots q_{n+1}$ and $c s_1'' \ldots
  s_n'' r_2'' \ipcontr^\alpha c q_1 \ldots q_{n+1}$. By
  Lemma~\ref{lem_contr_stable} there exists~$t'$ such that $t_2
  \ipcontr^\alpha t'$ and $c q_1 \ldots q_{n+1} \contr^\epsilon
  t'$. Hence by condition~(5) also $t_1 \equiv c s_1' \ldots s_n' r_2'
  \ipcontr^\beta t'$. See Fig.~\ref{fig_02}.

  \begin{figure}[ht]
    \centerline{
      \xymatrix{
        c s_1 \ldots s_n r_2 \ar@{=>}[r]^>>{\alpha}
        \ar@{=>}[d]^>>{\beta} & c s_1' \ldots
        s_n' r_2' \ar@{==>}[d]^>>{\beta} \\
        c s_1'' \ldots s_n'' r_2'' \ar@{==>}[r]^>>{\alpha}
        \ar@{->}[d]^>>{\epsilon} & c q_1 \ldots q_n q_{n+1} \ar@{-->}[d]^>>{\epsilon} \\
        t_2 \ar@{==>}[r]^>>{\alpha} & t'
      }
    }
    \caption{}
    \label{fig_02}
  \end{figure}

  If $t \ipcontr^\beta t_2$ follows by condition~(6) then $r_1 \equiv
  r_1' \equiv c \in \tau_2^{\tau_1}$, $r_2 \succ^{<\beta} c_1 \in
  \tau_1$ and $t_2 \equiv \Fc(c)(c_1)$. By part~($\rmnum{2}$) of the
  inductive hypothesis we conclude that $r_2' \succ^{<\beta}
  c_1$. Therefore $t_1 \equiv c r_2' \ipcontr^\beta \Fc(c)(c_1) \equiv
  t_2$.

  If $t \ipcontr^\beta t_2$ follows by condition~(7) then $r_1 \equiv
  \mathrm{Is}\, s$, $r_2 \equiv \tau$, $r_1' \equiv \mathrm{Is}\, s'$,
  $s \ipcontr^\alpha s'$, $t_2 \equiv \true$ and $s \succ^{<\beta} c$
  for some $c \in \tau \in \Tc$. By part~($\rmnum{2}$) of the
  inductive hypothesis we have $s' \succ^{<\beta} c$. Therefore $t_1
  \equiv \Is{s'}{r_2} \ipcontr^\beta \true \equiv t_2$. It is easy to
  see that it is impossible that $t \ipcontr^\beta t_2$ follows by
  condition~(8).

  Now suppose that $t \ipcontr^\alpha t_1$ follows by
  condition~(3). Then $t \equiv \lambda x \,.\, r$ and $t_1 \equiv
  \lambda x \,.\, r_1$ where $r \ipcontr^\alpha r_1$. It is easy too
  see that the only possibility is that $t \ipcontr^\beta t_2$ follows
  by condition~(3) as well. Then $t_2 \equiv \lambda x \,.\, r_2$
  where $r \ipcontr^\beta r_2$. By the inductive hypothesis there
  exists~$q$ such that $r_1 \ipcontr^\beta q$ and $r_2 \ipcontr^\alpha
  q$. Therefore $t_1 \equiv \lambda x \,.\, r_1 \ipcontr^\beta \lambda
  x \,.\, q$ and $t_2 \equiv \lambda x \,.\, r_2 \ipcontr^\alpha
  \lambda x \,.\, q$.

  Suppose that $t \ipcontr^\alpha t_1$ follows by condition~(4). Then
  $t \equiv (\lambda x \,.\, r_1) r_2$ and $t_1 \equiv r_1'[x/r_2']$ where
  $r_1 \ipcontr^\alpha r_1'$ and $r_2 \ipcontr^\alpha r_2'$. It is
  easy to see that the only possibility is that $t \ipcontr^\beta t_2$
  follows by condition~(4) as well. Then $t_2 \equiv r_1''[x/r_2'']$
  where $r_1 \ipcontr^\beta r_1''$ and $r_2 \ipcontr^\beta r_2''$. By
  the inductive hypothesis there exist~$q_1$ and~$q_2$ such that $r_1'
  \ipcontr^\beta q_1$, $r_1'' \ipcontr^\alpha q_1$, $r_2'
  \ipcontr^\beta q_2$ and $r_2'' \ipcontr^\alpha q_2$. Therefore by
  Lemma~\ref{lem_subst} we obtain $t_1 \equiv r_1'[x/r_2']
  \ipcontr^\beta q_1[x/q_2]$ and $t_2 \equiv r_1''[x/r_2'']
  \ipcontr^\alpha q_1[x/q_2]$.

  Suppose that $t \ipcontr^\alpha t_1$ follows by condition~(5). Then
  $t \equiv c r_1 \ldots r_n$, $r_i \ipcontr^\alpha r_i'$ and $c r_1'
  \ldots r_n' \contr^\epsilon t_1$. If $t \ipcontr^\beta t_2$ also
  follows by condition~(5), then there exist $r_1'',\ldots,r_n''$ such
  that $r_i \ipcontr^\beta r_i''$ and $c r_1'' \ldots r_n''
  \contr^\epsilon t_2$. By the inductive hypothesis there exist
  $q_1,\ldots,q_n$ such that $r_i' \ipcontr^\beta q_i$ and $r_i''
  \ipcontr^\alpha q_i$. Therefore $c r_1' \ldots r_n' \ipcontr^\beta c
  q_1 \ldots q_n$ and $c r_1'' \ldots r_n'' \ipcontr^\alpha c q_1
  \ldots q_n$. By Lemma~\ref{lem_contr_stable} there exist~$t_1'$
  and~$t_2'$ such that $t_1 \ipcontr^\beta t_1'$, $t_2 \ipcontr^\alpha
  t_2'$, $c q_1 \ldots q_n \contr^\epsilon t_1'$ and $c q_1 \ldots q_n
  \contr^\epsilon t_2'$. But by Lemma~\ref{lem_orthogonal} we have
  $t_1' \equiv t_2'$. See Fig.~\ref{fig_03}.

  \begin{figure}[ht]
    \centerline{
      \xymatrix{
        c r_1 \ldots r_n \ar@{=>}[r]^>>{\alpha} \ar@{=>}[d]^>>{\beta}
        &
        c r_1' \ldots r_n' \ar@{==>}[d]^>>{\beta} \ar@{->}[r]^>>{\epsilon}
        &
        t_1 \ar@{==>}[d]^>>{\beta}
        \\
        c r_1'' \ldots r_n'' \ar@{==>}[r]^>>{\alpha}
        \ar@{->}[d]^>>{\epsilon}
        &
        c q_1 \ldots q_n \ar@{-->}[d]^>>{\epsilon} \ar@{-->}[r]^>>{\epsilon}
        &
        t_1'
        \\
        t_2 \ar@{==>}[r]^>>{\alpha}
        &
        t_2' \ar@3{-}[ru]
        &
      }
    }
    \caption{}
    \label{fig_03}
  \end{figure}

  It is easy to verify that it is not possible that $t \ipcontr^\beta
  t_2$ follows by condition~(6). If $t \ipcontr^\beta t_2$ follows by
  condition~(7), then we must have $c \equiv \mathrm{Is}$, $t_1 \equiv
  t_2 \equiv \true$. It is not possible that $t \ipcontr^\beta t_2$
  follows by condition~(8).

  Suppose $t \ipcontr^\alpha t_1$ follows by condition~(6). Then $t
  \equiv c r \ipcontr^\alpha \Fc(c)(c_1) \equiv t_2$ where $c \in
  \tau_2^{\tau_1}$ and $r \succ^{<\alpha} c_1 \in \tau_1$. It is
  easily verified that the only possibility is when $t \ipcontr^\beta
  t_2$ follows by condition~(6) as well. Then $t_2 \equiv
  \Fc(c)(c_1')$ and $r \succ^{<\beta} c_1'$ for some $c_1' \in
  \tau_1$. By part~($\rmnum{3}$) of the inductive hypothesis we obtain
  $c_1' \equiv c_1$. Hence $t_1 \equiv t_2$.

  Suppose $t \ipcontr^\alpha t_1$ follows by condition~(7). Then $t
  \equiv \Is{r}{\tau}$, $t_1 \equiv \true$, and $t \ipcontr^\beta t_2$
  may only follow by condition~(7). But then we have $t_2 \equiv \true
  \equiv t_1$.

  Finally, suppose $t \ipcontr^\alpha t_1$ follows by condition~(8)
  and so does $t \ipcontr^\beta t_2$. Then e.g. $\upos{t}{\epsilon}
  \equiv \upos{t_1}{\epsilon} \equiv \upos{t_2}{\epsilon} \equiv
  \All\tau$, and for all $c \in \tau$ there exist~$t_c$ and~$t_c'$
  such that $\pos{t}{c} \ipcontr^\alpha t_c \uipreduces{<\alpha}
  \pos{t_1}{c}$ and $\pos{t}{c} \ipcontr^\beta t_c'
  \uipreduces{<\beta} \pos{t_2}{c}$. By the inductive hypothesis there
  exists~$r$ such that $t_c \ipcontr^\beta r$ and $t_c'
  \ipcontr^\alpha r$. By the inductive hypothesis and
  Lemma~\ref{lem_hindley_rosen} there exist~$q_1$ and~$q_2$ such that
  $\pos{t_1}{c} \ipcontr^\beta q_1$, $r \uipreduces{<\alpha} q_1$,
  $\pos{t_2}{c} \ipcontr^\alpha q_2$ and $r \uipreduces{<\beta}
  q_2$. Again, by the inductive hypothesis and
  Lemma~\ref{lem_hindley_rosen} there exists~$q_c$ such that $q_1
  \uipreduces{<\beta} q_c$ and $q_2 \uipreduces{<\alpha} q_c$. Hence
  for all $c \in \tau$ there exists~$q_1$ such that $\pos{t_1}{c}
  \ipcontr^\beta q_1 \uipreduces{<\beta} q_c$, and for all $c \in
  \tau$ there exists~$q_2$ such that $\pos{t_2}{c} \ipcontr^\alpha q_2
  \uipreduces{<\alpha} q_c$. Let~$q$ be such that $\upos{q}{\epsilon}
  \equiv \All\tau$ and $\pos{q}{c} \equiv q_c$ for $c \in \tau$. By
  the above considerations we have $t_1 \ipcontr^\beta q$ and $t_2
  \ipcontr^\alpha q$. See Fig.~\ref{fig_04}.

  \begin{figure}[ht]
    \centerline{
      \xymatrix{
        \pos{t}{c} \ar@{=>}[r]^>>{\alpha} \ar@{=>}[d]^>>{\beta}
        &
        t_c \ar@{==>}[d]^>>{\beta} \ar@{=>}[r]^{*}^>>{<\alpha}
        &
        \pos{t_1}{c} \ar@{==>}[d]^>>{\beta}
        \\
        t_c' \ar@{==>}[r]^>>{\alpha}
        \ar@{=>}[d]^{*}^>>{<\beta}
        &
        r \ar@{==>}[d]^{*}^>>{<\beta} \ar@{==>}[r]^{*}^>>{<\alpha}
        &
        q_1 \ar@{==>}[d]^{*}^>>{<\beta}
        \\
        \pos{t_2}{c} \ar@{==>}[r]^>>{\alpha}
        &
        q_2 \ar@{==>}[r]^{*}^>>{<\alpha}
        &
        q_c
      }
    }
    \caption{}
    \label{fig_04}
  \end{figure}

  Now we show condition~($\rmnum{2}$). Thus suppose $t \succ^\alpha c$
  and $t \ipcontr^\beta t'$. If $t \equiv c$ then $t' \equiv c$ and
  thus $t' \succ^\alpha c$. Otherwise $c \in \tau_2^{\tau_1}$ and for
  all $c_1 \in \tau_1$ there exists~$r$ such that $t c_1
  \uipreduces{<\alpha} r \succ^{<\alpha} \Fc(c)(c_1)$. Then $t c_1
  \ipcontr^\beta t' c_1$ and we conclude by part~($\rmnum{1}$) of the
  inductive hypothesis and Lemma~\ref{lem_hindley_rosen} that there
  exists~$r'$ such that $r \ipcontr^\beta r'$ and $t' c_1
  \uipreduces{<\alpha} r'$. By part~($\rmnum{2}$) of the inductive
  hypothesis we obtain $r' \succ^{<\alpha} \Fc(c)(c_1)$. Thus for
  every $c_1 \in \tau_1$ there exists~$r'$ such that $t' c_1
  \uipreduces{<\alpha} r' \succ^{<\alpha} \Fc(c)(c_1)$. Hence $t'
  \succ^\alpha \Fc(c)(c_1)$. See Fig.~\ref{fig_05}.

  \begin{figure}[ht]
    \centerline{
      \xymatrix{
        t c_1 \ar@{=>}[d]^>>{\beta} \ar@{=>}[r]^{*}^>>{<\alpha}
        &
        r \ar@{==>}[d]^>>{\beta} \ar@0{}[r]|(0.4)*+{\succ^{<\alpha}}
        &
        \Fc(c)(c_1)
        \\
        t' c_1 \ar@{==>}[r]^{*}^>>{<\alpha}
        &
        r \ar@0{}[r]|(0.4)*+{\succ^{<\alpha}}
        &
        \Fc(c)(c_1)
      }
    }
    \caption{}
    \label{fig_05}
  \end{figure}

  It remains to show condition~($\rmnum{3}$). Thus suppose $t
  \succ^\alpha c_1$ and $t \succ^\beta c_2$ for $c_1,c_2 \in \tau \in
  \Tc$. If $\tau \subseteq \Bool$, $\tau \subseteq \delta$ or $\tau
  \subseteq \bar{\iota} \in \bar{\Tc}_I$ then $t \equiv c_1 \equiv
  c_2$ because in this case $t \succ^\alpha c_1$ and $t \succ^\beta
  c_2$ may only be obtained by condition~(a) in the definition
  of~$\succ$. Otherwise $\tau \subseteq \tau_2^{\tau_1}$ and e.g. $t
  \succ^\alpha c_1$ is obtained by condition~(b), hence $\alpha > 0$.

  If $c_1 \not\equiv c_2$ then there exists~$c \in \tau_1$ such that
  $\Fc(c_1)(c) \not\equiv \Fc(c_2)(c)$. There exists~$t_1$ such that
  $t c \uipreduces{<\alpha} t_1 \succ^{<\alpha} \Fc(c_1)(c)$. If $t
  \equiv c_2$ then by inspecting the definitions we see that this is
  only possible when $c_2 c \ipcontr^{\gamma} \Fc(c_2)(c')
  \succ^{<\alpha} \Fc(c_1)(c)$ where~$c \succ^{<\gamma} c' \in \tau_1$
  and $\gamma < \alpha$. By condition~(a) we have $c \succ^\beta
  c$. Since $c, c' \in \tau_1$ and $\gamma < \alpha$ we conlcude by
  part~($\rmnum{3}$) of the inductive hypothesis that $c \equiv
  c'$. Thus $\Fc(c_2)(c) \succ^{<\alpha} \Fc(c_1)(c)$. Obviously
  $\Fc(c_2)(c) \succ^\beta \Fc(c_2)(c)$ and $\Fc(c_2)(c), \Fc(c_1)(c)
  \in \tau_2$, so again by part~($\rmnum{3}$) of the inductive
  hypothesis we obtain~$\Fc(c_1)(c) \equiv
  \Fc(c_2)(c)$. Contradiction.

  Thus assume that $t \succ^\beta c_2$ also follows by condition~(b)
  in the definition of~$\succ$. Then there exists~$t_2$ such that $t c
  \uipreduces{<\beta} t_2 \succ^{<\beta} \Fc(c_2)(c)$. By
  part~($\rmnum{1}$) of the inductive hypothesis and
  Lemma~\ref{lem_hindley_rosen} there exists~$r$ such that $t_2
  \uipreduces{<\alpha} r$ and $t_1 \uipreduces{<\beta} r$. By
  part~($\rmnum{2}$) of the inductive hypothesis we have $r
  \succ^{<\alpha} \Fc(c_1)(c)$ and $r \succ^{<\beta} \Fc(c_2)(c)$. By
  part~($\rmnum{3}$) of the inductive hypothesis we obtain
  $\Fc(c_1)(c) \equiv \Fc(c_2)(c)$. Contradiction. See
  Fig.~\ref{fig_06}.

  \newcommand{\rotsucc}[1]{
    \begin{rotate}{90}
      \(\succ^{\begin{rotate}{-90}\(\mathsmaller{#1}\)\end{rotate}}\)
    \end{rotate}
  }

  \begin{figure}[ht]
    \centerline{
      \xymatrix{
        t c \ar@{=>}[d]^>>{\beta} \ar@{=>}[r]^{*}^>>{<\alpha}
        &
        t_1 \ar@{==>}[d]^>>{\beta} \ar@0{}[r]|(0.4)*+{\succ^{<\alpha}}
        &
        \Fc(c_1)(c)
        \\
        t_2 \ar@0{}[d]|(0.35)*+{\rotsucc{<\beta}} \ar@{==>}[r]^{*}^>>{<\alpha}
        &
        r \ar@0{}[d]|(0.35)*+{\rotsucc{<\beta}} \ar@0{}[r]|(0.4)*+{\succ^{<\alpha}}
        &
        \Fc(c_1)(c)
        \\
        \Fc(c_2)(c)
        &
        \Fc(c_2)(c) \ar@3{-}[ru]
        &
      }
    }
    \caption{}
    \label{fig_06}
  \end{figure}
\end{proof}

\begin{corollary}\label{cor_cr}
  The relation~$\ipcontr$ has the Church-Rosser property.
\end{corollary}

\begin{definition} \rm
  The \emph{rank} of a type $\tau \in \Tc$, denoted $\rank(\tau)$, is
  the smallest $n \in \Nbb$ such that $\tau \in \Tc_n$. The
  \emph{canonical type} of a canonical constant~$c \in \Sigma$,
  denoted $\tau(c)$, is defined as follows.
  \begin{itemize}
  \item If $c \in \delta$ then $\tau(c) = \delta$.
  \item If $c \in \Bool$ then $\tau(c) = \Bool$.
  \item If $c \in \bar{\iota} \in \bar{\Tc}_I$ then $\tau(c) = \iota$.
  \item Otherwise let $\tau_2^{\tau_1} \in \Tc$ be such that $c \in
    \tau_2^{\tau_1}$ and $\rank(\tau_2) \le \rank(\tau_2')$ for every
    $\tau_2' \in \Tc$ such that $c \in {\tau_2'}^{\tau_1}$. Then
    $\tau(c) = \tau_2^{\tau_1}$. Note that there may be more than
    one~$\tau_2$ satisfying the above condition. In this case we
    arbitrarily choose one of them, and it does not matter which.
  \end{itemize}
  The rank of a canonical constant~$c$, denoted $\rank(c)$, is the
  rank of its canonical type.
\end{definition}

\begin{lemma}\label{lem_smaller_rank}
  The following conditions hold.
  \begin{itemize}
  \item For all $\tau_1, \tau_2 \in \Tc$ we have $\rank(\tau_1),
    \rank(\tau_2) < \rank(\tau_2^{\tau_1})$.
  \item If $c \in \tau$ then $\rank(c) \le \rank(\tau)$.
  \end{itemize}
\end{lemma}

\begin{proof}
  For the first condition, note that $\rank(\tau_2^{\tau_1}) > 0$ and
  if $\tau_2^{\tau_1} \in \Tc_{n+1}$ then $\tau_1,\tau_2 \in \Tc_n$.

  If $\tau(c) \in \{\delta, \Bool\} \cup \bar{\Tc}_I$ then the second
  condition is obvious. Otherwise $\tau(c) = \tau_2^{\tau_1}$, and if
  $c \in \tau$ then $\tau \subseteq \tau_3^{\tau_1}$ with $\rank(\tau)
  \ge \rank(\tau_3^{\tau_1})$ and $\rank(\tau_2) \le
  \rank(\tau_3)$. Suppose $\rank(\tau_1) = n_1$, $\rank(\tau_2) =
  n_2$, $\rank(\tau_3) = n_3$. Then $n_2 \le n_3$, $\rank(c) =
  \max(n_1, n_2) + 1$ and $\rank(\tau) \ge \max(n_1, n_3) + 1$. Thus
  $\rank(c) \le \rank(\tau)$.
\end{proof}

\begin{definition} \rm
  We write $t \succsim c$ if $c \in \tau_2^{\tau_1}$ and for every
  $c_1 \in \tau_1$ there exists~$t_{c_1}$ such that for all~$t_1$ with
  $t_1 \succ c_1$ we have $t t_1 \ipreduces t_{c_1} \succ
  \Fc(c)(c_1)$. If for some $c_1 \in \tau_1$ there is more than one
  term~$t_{c_1}$ satisfying the above condition, then we fix one
  arbitrarily, but globally, i.e. given~$t$ and~$c$ such that $t
  \succsim c$ we assume that~$t_{c_1}$ is uniquely determined for
  each~$c_1 \in \tau_1$, and it depends only on~$t$, $c$
  and~$c_1$. Note that if $t \succsim c$ then $t \succ c$.

  Let $t \succ c$. The \emph{mutual rank} of~$t$ and~$c$, denoted
  $\rank(t,c)$, is defined by induction on~$\rank(c)$. If $t \equiv c$
  then $\rank(t,c) = 0$. If $t \not\succsim c$ then $\rank(t,c) =
  \rank(c)$. If $t \succsim c$ but $t \not\equiv c$ then $c \in
  \tau_2^{\tau_1}$ and $\rank(t,c)$ is defined by
  \[
  \rank(t,c) = \sup_{c_1 \in \tau_1} \rank(t_{c_1},\Fc(c)(c_1))
  \]
  where~$t_{c_1}$ is the term required by the definition
  of~$\succsim$, such that for all terms~$t_1$ with $t_1 \succ c_1$ we
  have $t t_1 \ipreduces t_{c_1} \succ \Fc(c)(c_1)$. Note that
  $\rank(t,c) \le \rank(c)$, and if $t \succsim c$ then $\rank(t,c) <
  \rank(c)$.

  Two positions $p_1,p_2 \in \Sigma^*$ are \emph{parallel} if neither
  $p_1 \sqsubseteq p_2$ nor $p_2 \sqsubseteq p_1$. We write $t_1 \gg^n
  t_2$ if there exists a set $P \subseteq \Pos(t_1) \cap \Pos(t_2)$ of
  pairwise parallel positions such that for $p \in P$ we have
  $\pos{t_1}{p} \succ \pos{t_2}{p}$ and $\pos{t_1}{p} \not\equiv
  \pos{t_2}{p}$, no free variables of~$\pos{t_1}{p}$ become bound
  in~$t_1$, for every $p \in \Pos(t_1) \setminus P$ we have $p \in
  \Pos(t_2)$ and $\upos{t_1}{p} \equiv \upos{t_2}{p}$, and
  $\rank_P(t_1,t_2) \le n$, where $\rank_P(t_1,t_2) = \sup_{p\in P}
  \rank(\pos{t_1}{p}, \pos{t_2}{p})$. We write $t_1 \gg^{<n} t_2$ if
  the same conditions hold except that $\rank_P(t_1,t_2) < n$. We
  write $t_1 \gg t_2$ if $t_1 \gg^n t_2$ for some $n \in \Nbb$.
\end{definition}

\begin{lemma}\label{lem_gg_subst}
  If $t_1 \gg^n t_1'$, $t_2 \gg^n t_2'$ and $x \notin
  FV(\pos{t_1}{p})$ for all $p \in P_1$, where~$P_1$ is the set of
  positions required by the definition of $t_1 \gg^n t_1'$, then
  $t_1[x/t_2] \gg^n t_1'[x/t_2']$.
\end{lemma}

\begin{proof}
  Let $P_2$ be the set of positions required by the definition of $t_2
  \gg^n t_2'$. Take
  \[
  P = P_1 \cup \{ p \in \Pos(t_1[x/t_2]) \;|\; p = p_1p_2,\, \pos{t_1}{p_1} \equiv x,\, p_2 \in P_2 \}
  \]
  as the set of positions required by the definition of $t_1[x/t_2]
  \gg^n t_1'[x/t_2']$.
\end{proof}

\begin{lemma}\label{lem_gg}
  The following conditions hold.
  \begin{itemize}
  \item[($\rmnum{1}$)] If $t_1 \gg^n t_2 \ipcontr_{R_\beta}^\alpha
    t_2'$ then $t_1 \ipreduces t_1' \gg^n t_2'$.
  \item[($\rmnum{2}$)] If $t_1 \gg^n t_1'$ and $t_2
    \ipcontr_{R_\beta}^\alpha t_1'$ then there exists~$t_2'$ such that
    $t_2' \ipreduces t_1$ and $t_2' \gg^n t_2$.
  \item[($\rmnum{3}$)] If $t_1 \gg^n t_2 \succ_{R_\beta}^\alpha c$
    then $t_1 \succ c$.
  \end{itemize}
\end{lemma}

\begin{proof}
  Induction on tuples $\langle n, \beta, \alpha, h(t_2) \rangle$
  ordered lexicographically.

  We first show condition~($\rmnum{1}$). Thus suppose $t_1 \gg^n t_2
  \ipcontr_{R_\beta}^\alpha t_2'$. We consider possible forms of~$t_2$
  according to the definition of~$\ipcontr$.

  If $t_2 \equiv t_2'$ then the claim is obvious. If $t_2
  \contr_{R_\beta}^\epsilon t_2'$ and $t_2 \ipcontr_{R_\beta}^\alpha
  t_2'$ follows by condition~(1) in the definition of~$\ipcontr$, then
  the only non-obvious case is when $t_2 \equiv \Eq{r_1}{r_2}
  \contr_{R_\beta}^\epsilon \true \equiv t_2'$. Then $r_1
  \ipeqvred_{R_{<\beta}} r_2$ and $t_1 \equiv \Eq{r_1'}{r_2'}$ where
  $r_1' \gg^n r_1$ and $r_2' \gg^n r_2$. It follows from
  parts~($\rmnum{1}$) and~($\rmnum{2}$) of the inductive hypothesis
  that $r_1' \ipeqvred r_2'$. Thus $t_1 \ipcontr \true \equiv t_2'$.

  If $t_2 \ipcontr_{R_\beta}^\alpha t_2'$ follows by condition~(2)
  then $t_2 \equiv r_1 r_2$, $t_2' \equiv r_1' r_2'$, $r_1
  \ipcontr_{R_\beta}^\alpha r_1'$ and $r_2 \ipcontr_{R_\beta}^\alpha
  r_2'$. We must also have $t_1 \equiv q_1 q_2$ where $q_1 \gg^n r_1$
  and $q_2 \gg^n r_2$. By the inductive hypothesis ($h(r_1), h(r_2) <
  h(t_2)$) there exist~$q_1'$ and~$q_2'$ such that $q_1 \ipreduces
  q_1'$, $q_2 \ipreduces q_2'$, $q_1' \gg^n r_1'$, $q_2' \gg^n
  r_2'$. Thus $t_1 \equiv q_1 q_2 \ipreduces q_1' q_2' \gg^n r_1' r_2'
  \equiv t_2'$. If $t_2 \ipcontr_{R_\beta}^\alpha t_2'$ follows by
  condition~(3) then the argument is analogous.

  Suppose $t_2 \ipcontr_{R_\beta}^\alpha t_2'$ follows by
  condition~(4). Then $t_2 \equiv (\lambda x \,.\, r_1) r_2$ and $t_2'
  \equiv r_1'[x/r_2']$ where $r_1 \ipcontr_{R_\beta}^\alpha r_1'$ and
  $r_2 \ipcontr_{R_\beta}^\alpha r_2'$. We must also have $t_1 \equiv
  (\lambda x \,.\, q_1) q_2$ where $q_1 \gg^n r_1$ and $q_2 \gg^n r_2$. By
  the inductive hypothesis there exist~$q_1'$ and~$q_2'$ such that
  $q_1 \ipreduces q_1'$, $q_2 \ipreduces q_2'$, $q_1' \gg^n r_1'$,
  $q_2' \gg^n r_2'$. By Lemma~\ref{lem_gg_subst} we obtain
  $q_1'[x/q_2'] \gg^n r_1'[x/r_2']$. Thus $t_1 \equiv (\lambda x
  \,.\, q_1) q_2 \ipreduces q_1'[x/q_2'] \gg^n r_1'[x/r_2'] \equiv t_2'$.

  Suppose $t_2 \ipcontr_{R_\beta}^\alpha t_2'$ follows by
  condition~(5). Then $t_2 \equiv c r_1 \ldots r_m$, $r_i
  \ipcontr_{R_\beta}^\alpha r_i'$, $c r_1' \ldots r_m'
  \contr_{R_\beta}^\epsilon t_2'$. By the definition of~$R_\beta$, the
  constant~$c$ is not a canonical constant. This implies that $t_1
  \equiv c q_1 \ldots q_m$ where $q_i \gg^n r_i$. By the inductive
  hypothesis ($h(r_i) < h(t_2)$) there exist~$q_1',\ldots,q_m'$ such
  that $q_i' \gg^n r_i'$ and $q_i \ipreduces q_i'$. Thus $c q_1'
  \ldots q_m' \gg^n c r_1' \ldots r_m' \contr_{R_\beta}^\epsilon
  t_2'$. But we have already verified in this inductive step that this
  implies that there exists~$t_1'$ such that $c q_1' \ldots q_m'
  \ipreduces t_1' \gg^n t_2'$. Therefore $t_1 \equiv c q_1 \ldots q_m
  \ipreduces c q_1' \ldots q_m \ipreduces t_1' \gg^n t_2'$.

  Suppose $t_2 \ipcontr_{R_\beta}^\alpha t_2'$ follows by
  condition~(6). Then $t_2 \equiv c r_2'$, $t_2' \equiv c_2 \equiv
  \Fc(c)(c_1)$, $r_2' \succ_{R_\beta}^{<\alpha} c_1 \in \tau_1$,
  $\tau(c) = \tau_2^{\tau_1}$ for some $\tau_2 \in \Tc$. We also have
  $t_1 \equiv r_1 r_2$ where $r_1 \gg^n c$, hence $r_1 \succ c$, and
  $r_2 \gg^n r_2' \succ_{R_\beta}^{<\alpha} c_1$. By
  part~($\rmnum{3}$) of the inductive hypothesis we obtain $r_2 \succ
  c_1$. First assume that $r_1 \equiv c$. Then $t_1 \equiv r_1 r_2
  \equiv c r_2 \ipcontr \Fc(c)(c_1) \equiv t_2'$ by virtue of $r_2
  \succ c_1$, and we are done. So suppose $r_1 \not\equiv c$. If $r_1
  \succsim c$ then let~$q$ be the term required by the definition
  of~$\succsim$, such that for every term~$r$ with $r \succ c_1$ we
  have $r_1 r \ipreduces q \succ c_2$. Then $\rank(q,c_2) \le
  \rank(r_1,c) \le n$, and $t_1 \equiv r_1 r_2 \ipreduces q \succ
  c_2$, because $r_2 \succ c_1$. So $t_1 \ipreduces q \gg^n c_2 \equiv
  t_2'$, which is our claim. Therefore suppose $r_1 \not\equiv c$ and
  $r_1 \not\succsim c$. Then we have $\rank(r_2,c_1) \le \rank(c_1)
  \le \rank(\tau_1) < \rank(c) = \rank(r_1,c) \le n$ by
  Lemma~\ref{lem_smaller_rank}. Thus $r_1 r_2 \gg^{<n} r_1 c_1$. By
  the fact that $r_1 \succ c$ there exists~$t$ such that $r_1 c_1
  \ipreduces t \succ \Fc(c)(c_1)$. By part~($\rmnum{1}$) of the
  inductive hypothesis there exists~$t'$ such that $r_1 r_2 \ipreduces
  t' \gg^{<n} t$. Let $c_2 \equiv \Fc(c)(c_1)$. We have $t' \gg^{<n} t
  \succ c_2$, so by part~($\rmnum{3}$) of the inductive hypothesis we
  obtain $t' \succ c_2$. Since $\rank(t',c_2) \le \rank(c_2) \le
  \rank(\tau_2) < \rank(c) = \rank(r_1,c) \le n$, we conclude that
  $t_1 \equiv r_1 r_2 \ipreduces t' \gg^n c_2 \equiv t_2'$.

  Suppose $t_2 \ipcontr_{R_\beta}^\alpha t_2'$ follows by
  condition~(7). Then $t_2 \equiv \Is{r_2}{\tau}$, $t_2' \equiv \true$,
  $t_1 \equiv \Is{r_1}{\tau}$, $r_1 \gg^n r_2$, and $r_2
  \succ_{R_\beta}^{<\alpha} c$ for some $c \in \tau \in \Tc$. By
  part~($\rmnum{3}$) of the inductive hypothesis we have $r_1 \succ
  c$. Therefore $t_1 \equiv \Is{r_1}{\tau} \ipcontr \true \equiv t_2'$.

  Finally, suppose $t_2 \ipcontr_{R_\beta}^\alpha t_2'$ follows by
  condition~(8). Then e.g. $\upos{t_2}{\epsilon} \equiv
  \upos{t_2'}{\epsilon} \equiv \upos{t_1}{\epsilon} \equiv \All\tau$,
  and for all $c \in \tau$ there exists~$t_c$ such that $\pos{t_1}{c}
  \gg^n \pos{t_2}{c} \ipcontr_{R_\beta}^\alpha t_c
  \udipreduces{<\alpha}{R_{\beta}} \pos{t_2'}{c}$. By
  part~($\rmnum{1}$) of the inductive hypothesis ($h(\pos{t_2}{c}) <
  h(t_2)$) there exists~$t_c'$ such that $\pos{t_1}{c} \ipreduces t_c'
  \gg^n t_c$. Applying the inductive hypothesis again, we conclude
  that there exists~$q_c$ such that $t_c' \ipreduces q_c \gg^n
  \pos{t_2'}{c}$. Let~$q$ be such that $\upos{q}{\epsilon} \equiv
  \All\tau$ and $\pos{q}{c} \equiv q_c$ for $c \in \tau$. Then
  $\pos{t_1}{c} \ipreduces \pos{q}{c}$ for all $c \in \tau$, and hence
  $t_1 \contr q$. We also have $q \gg^n t_2'$, because $\pos{q}{c}
  \gg^n \pos{t_2'}{c}$ for all $c \in \tau$.

  We now show condition~($\rmnum{2}$). Thus suppose $t_1 \gg^n t_1'$
  and $t_2 \ipcontr_{R_\beta}^\alpha t_1'$. We consider possible forms
  of~$t_2$ according to the definition of $t_2
  \ipcontr_{R_\beta}^\alpha t_1'$.

  If $t_2 \equiv t_1'$ then the claim is obvious. If $t_2
  \ipcontr_{R_\beta}^\alpha t_1'$ follows by condition~(1) in the
  definition of~$\ipcontr_{R_\beta}^\alpha$, then $t_2
  \contr_{R_\beta}^\epsilon t_1'$ and the claim follows easily by
  inspecting the definition of~$R_\beta$. If $t_2
  \ipcontr_{R_\beta}^\alpha t_1'$ follows by condition~(2), then $t_2
  \equiv r_1 r_2$, $t_1' \equiv r_1' r_2'$, $r_1
  \ipcontr_{R_\beta}^\alpha r_1'$, $r_2 \ipcontr_{R_\beta}^\alpha
  r_2'$, $t_1 \equiv q_1' q_2'$, $q_1' \gg^n r_1'$, $q_2' \gg^n
  r_2'$. By the inductive hypothesis ($h(r_1), h(r_2) < h(t_2)$) there
  exist~$q_1$, $q_2$ such that $q_1 \ipreduces q_1'$, $q_2 \ipreduces
  q_2'$, $q_1 \gg^n r_1$, $q_2 \gg^n r_2$. Thus $q_1 q_2 \ipreduces
  t_1$ and $q_1 q_2 \gg^n t_2$. If $t_2 \ipcontr_{R_\beta}^\alpha
  t_1'$ follows by condition~(3) or condition~(4) then the argument is
  analogous.

  If $t_2 \ipcontr_{R_\beta}^\alpha t_1'$ follows by condition~(5)
  then $t_2 \equiv c r_1 \ldots r_m$, $r_i \ipcontr_{R_\beta}^\alpha
  r_i'$, $c r_1' \ldots r_m' \contr_{R_\beta}^\epsilon t_1'$, $t_1
  \equiv c q_1' \ldots q_m'$, $q_i' \gg^n r_i'$. By the inductive
  hypothesis there exist $q_1,\ldots,q_m$ such that $q_i \ipreduces
  q_i'$ and $q_i \gg^n r_i$. Let $q \equiv c q_1 \ldots q_m$. We have
  $q \gg^n t_2$ and $q \ipreduces c q_1' \ldots q_m' \equiv t_1$.

  If $t_2 \ipcontr_{R_\beta}^\alpha t_1'$ follows by condition~(6)
  then $t_2 \equiv c r$, $t_1' \equiv c_2 \equiv \Fc(c)(c_1)$, $c \in
  \tau_2^{\tau_1}$, $r \succ_{R_\beta}^{<\alpha} c_1 \in \tau_1$, $t_1
  \gg^n c_2$. If $t_1 \equiv c_2$ then the claim is obvious. Otherwise
  $t_1 \succ c_2$. Let $e \in (\mbox{\small Bool}^{\tau_1})^{\tau_1}$
  be such that for $d_1, d_2 \in \tau_1$, we have $\Fc(e)(d_1)(d_2)
  \equiv \true$ if $d_1 \equiv d_2$, and $\Fc(e)(d_1)(d_2) \equiv
  \false$ if $d_1 \not\equiv d_2$. Let $q \equiv \lambda x . \Cond{(e
    c_1 x)}{t_1}{(c x)}$. If $s \succ d \in \tau_1$ and $d \not\equiv
  c_1$ then $q s \ipcontr \Cond{(e c_1 s)}{t_1}{(c s)} \ipreduces
  \Cond{\false}{t_1}{\Fc(c)(d)} \ipcontr \Fc(c)(d)$. If $s \succ c_1$
  then $q s \ipcontr \Cond{(e c_1 s)}{t_1}{(c s)} \ipcontr
  \Cond{\true}{t_1}{(c s)} \ipcontr t_1 \succ \Fc(c)(c_1) \equiv
  c_2$. Therefore $q \succsim c$, and $\rank(q,c) = \rank(t_1,c_2) \le
  n$, since $\rank(\Fc(c)(d), \Fc(c)(d)) = 0$. Thus $q r \gg^n c r
  \equiv t_2$ and $q r \ipreduces t_1$, because $r \succ c_1$.

  If $t_2 \ipcontr_{R_\beta}^\alpha t_1'$ follows by condition~(7)
  then the claim is obvious. If $t_2 \ipcontr_{R_\beta}^\alpha t_1'$
  follows by condition~(8) then e.g $\upos{t_2}{\epsilon} \equiv
  \upos{t_1'}{\epsilon} \equiv \upos{t_1}{\epsilon} \equiv \All\tau$
  and for all $c \in \tau$ there exists~$t_c$ such that $\pos{t_2}{c}
  \ipcontr_{R_\beta}^\alpha t_c \ipreduces_{R_\beta}^{<\alpha}
  \pos{t_1'}{c}$. We also have $\pos{t_1}{c} \gg^n \pos{t_1'}{c}$ for
  all $c \in \tau$. Therefore by the inductive hypothesis for every $c
  \in \tau$ there exists~$t_c'$ such that $t_c' \gg^n t_c$ and $t_c'
  \ipreduces \pos{t_1}{c}$. Again, by the inductive hypothesis
  ($h(\pos{t_2}{c}) < h(t_2)$), for every $c \in \tau$ there
  exists~$q_c$ such that $q_c \gg^n \pos{t_2}{c}$ and $q_c \ipreduces
  t_c' \ipreduces \pos{t_1}{c}$. Let~$q$ be such that
  $\upos{q}{\epsilon} \equiv \All\tau$ and $\pos{q}{c} \equiv q_c$ for
  $c \in \tau$. Then $q \gg^n t_2$ and $q \ipcontr t_1$.

  It remains to show condition~($\rmnum{3}$). Thus suppose $t_1 \gg^n
  t_2 \succ_{R_\beta}^\alpha c$. If $t_2 \equiv c$ then the claim is
  obvious. Otherwise $c \in \tau_2^{\tau_1}$ and for every~$c_1 \in
  \tau_1$ there exists~$q_{c_1}$ such that $t_2 c_1
  \udipreduces{<\alpha}{R_\beta} q_{c_1} \succ_{R_\beta}^{<\alpha}
  c$. Since $t_1 c_1 \gg^n t_2 c_1$, by part~($\rmnum{1}$) of the
  inductive hypothesis for each~$c_1 \in \tau_1$ there
  exists~$q_{c_1}'$ such that $t_1 c_1 \ipreduces q_{c_1}' \gg^n
  q_{c_1} \succ_{R_\beta}^{<\alpha} c$. By part~($\rmnum{3}$) of the
  inductive hypothesis we obtain $q_{c_1}' \succ c$. This implies that
  $t_1 \succ c$.
\end{proof}

\begin{corollary}\label{cor_succ_arg}
  If $t_2 \succ c$ and $t_1 c \ipreduces d \in \Bool$, then $t_1 t_2
  \ipreduces d$.
\end{corollary}

The above lemma and the ensuing corollary confirm our intuition about
the meaning of~$\succ$. This basically finishes the hard part of the
proof. What remains are some relatively straightforward lemmas.

\newcommand{\vtr}{\ensuremath{\vartriangleright}}

\begin{definition} \rm
  We write $t \vtr \tau$ if $t \ipreduces t' \succ c$ for some term~$t'$
  and some $c \in \tau$.
\end{definition}

\begin{lemma}\label{lem_is}
  We have $\Is{t_1}{t_2} \ipeqvred \true$ iff there exists~$\tau \in
  \Tc$ such that $t_1 \vtr \tau$ and $t_2 \ipreduces \tau$. Moreover,
  this $\tau \in \Tc$ is uniquely determined.
\end{lemma}

\begin{proof}
  If $t_2 \ipreduces \tau \in \Tc$ and $t_1 \ipreduces t' \succ c$ for
  some $c \in \tau$ then $\Is{t_1}{t_2} \ipreduces \Is{t_1}{\tau}
  \ipreduces \Is{t'}{\tau} \ipcontr \true$ by condition~(7) in the
  definition of~$\ipcontr$. If $\Is{t_1}{t_2} \ipeqvred \true$ then
  $\Is{t_1}{t_2} \ipreduces \true$ by the Church-Rosser property
  of~$\ipcontr$ and the fact that~$\true$ is in normal form. But this
  is only possible when $\Is{t_1}{t_2} \ipreduces \Is{t'}{\tau}
  \ipcontr \true$ where $\tau \in \Tc$, $t_2 \ipreduces \tau$ and $t_1
  \ipreduces t' \succ c$ for some $c \in \tau$, i.e. $t_1 \vtr \tau$.

  To see that~$\tau$ is uniquely determined it suffices to notice that
  it is in normal form w.r.t.~$\ipcontr$ and~$\ipcontr$ has the
  Church-Rosser property.
\end{proof}

\begin{lemma}\label{lem_forall}
  The following conditions hold.
  \begin{itemize}
  \item[(a)] If $\tau \in \Tc$ and for all~$t_2$ such that $t_2 \vtr \tau$
    we have $t_1 t_2 \ipeqvred \true$, then $\forall \tau t_1 \ipeqvred
    \true$.
  \item[(b)] If $\tau \in \Tc$ and there exists~$t_2$ such that $t_2 \vtr
    \tau$ and $t_1 t_2 \ipeqvred \false$, then $\forall \tau t_1
    \ipeqvred \false$.
  \item[(c)] If $\forall t_1 t_2 \ipeqvred \true$ then for all~$t_3$ such
    that $\Is{t_3}{t_1} \ipeqvred \true$ we have $t_2 t_3 \ipeqvred
    \true$.
  \item[(d)] If $\forall t_1 t_2 \ipeqvred \false$ then there
    exists~$t_3$ such that $\Is{t_3}{t_1} \ipeqvred \true$ and $t_2 t_3
    \ipeqvred \false$.
  \end{itemize}
\end{lemma}

\begin{proof}
  We show condition~(a). Suppose $\tau \in \Tc$ and for all~$t_2$ such
  that $t_2 \vtr \tau$ we have $t_1 t_2 \ipeqvred \true$. If $\tau
  \equiv \emptyset$ then $\forall \emptyset t_1 \contr \true$ and the
  claim is obvious. Otherwise we have $\forall \tau t_1 \contr t$
  where $\upos{t}{\epsilon} \equiv \All\tau$ and $\pos{t}{c} \equiv
  t_1 c$ for $c \in \tau$. Since $c \in \tau$ we have $c \vtr \tau$,
  so $t_1 c \ipeqvred \true$, hence $t_1 c \ipreduces \true$ by the
  Church-Rosser property. Therefore $t \ipreduces \true$. Condition~(b)
  is shown in a completely analogous way.

  We show condition~(c). Suppose $\forall t_1 t_2 \ipeqvred
  \true$. Then $\forall t_1 t_2 \ipreduces \true$, which is only
  possible when $\forall \tau t_2 \ipreduces \true$ and $t_1 \ipreduces
  \tau$ for some $\tau \in \Tc$. Suppose~$t_3$ is such that
  $\Is{t_3}{t_1} \ipeqvred \true$. Then $\Is{t_3}{t_1} \ipreduces
  \true$, and we must have $t_1 \ipreduces \tau' \in \Tc$ and $t_3 \vtr
  \tau'$, by Lemma~\ref{lem_is}. Because both~$\tau$ and~$\tau'$ are
  in normal form, we conclude by the Church-Rosser property that $\tau
  \equiv \tau'$. Hence $t_3 \vtr \tau$, i.e. $t_3 \ipreduces t_3'
  \succ c \in \tau$. Since $\forall \tau t_2 \ipreduces \true$, it is
  easy to see by inspecting the definitions that $t_2 c \ipreduces
  \true$. By Corollary~\ref{cor_succ_arg} we obtain $t_2 t_3 \ipreduces
  t_2 t_3' \ipreduces \true$.

  Condition~(c) follows easily from definitions, Lemma~\ref{lem_is}
  and the Church-Rosser property.
\end{proof}

\begin{lemma}\label{lem_fun}
  The following conditions hold.
  \begin{itemize}
  \item[(a)] If $\tau_1, \tau_2 \in \Tc$ and for all~$t_2$ such that $t_2
    \vtr \tau_1$ we have $t_1 t_2 \vtr \tau_2$, then $t_1 \vtr
    \tau_2^{\tau_1}$.
  \item[(b)] If $\tau_1, \tau_2 \in \Tc$, $t_1 \vtr \tau_2^{\tau_1}$ and
    $t_2 \vtr \tau_1$ then $t_1 t_2 \vtr \tau_2$.
  \end{itemize}
\end{lemma}

\begin{proof}
  We show condition~(a). Suppose $\tau_1, \tau_2 \in \Tc$ and for
  all~$t_2$ such that $t_2 \vtr \tau_1$ we have $t_1 t_2 \vtr
  \tau_2$. Let $c \in \tau_1$. We obviously have $c \vtr \tau_1$, so
  $t_1 c \vtr \tau_2$, i.e. there exists a term~$t_c$ and a
  constant~$c' \tau_2$ such that $t_1 c \ipreduces t_c \succ
  c'$. Recall that~$\tau_2^{\tau_1}$ consists of all set-theoretic
  functions from~$\tau_1$ to~$\tau_2$. In particular, there exists $d
  \in \tau_2^{\tau_1}$ such that $\Fc(d)(c) \equiv c'$ for every $c
  \in \tau_1$ and~$c' \in \tau_2$ depending on~$c$ as above. But then
  $t_1 \succ d$, and hence $t_1 \vtr \tau_2^{\tau_1}$.

  We show condition~(b). Suppose $\tau_1, \tau_2 \in \Tc$, $t_1 \vtr
  \tau_2^{\tau_1}$ and $t_2 \vtr \tau_1$. Then $t_1 \ipreduces t_1'
  \succ c \in \tau_2^{\tau_1}$ and $t_2 \ipreduces t_2' \succ c_1 \in
  \tau_1$. By condition~(6) in the definition of~$\ipcontr$ we obtain
  $c t_2' \ipcontr \Fc(c)(c_1) \succ \Fc(c)(c_1) \in \tau_2$. If $t_1'
  \equiv c$ then $t_1 t_2 \ipreduces c t_2' \ipcontr \Fc(c)(c_1) \in
  \tau_2$, so $t_1 t_2 \vtr \tau_2$. Otherwise $t_1' t_2' \gg c t_2'
  \ipcontr \Fc(c)(c_1)$, and by part~($\rmnum{1}$) of
  Lemma~\ref{lem_gg} there exists~$t$ such that $t_1 t_2 \ipreduces
  t_1' t_2' \ipreduces t \succ \Fc(c)(c_1) \in \tau_2$. Hence $t_1 t_2
  \vtr \tau_2$.
\end{proof}

\begin{lemma}\label{lem_subtype}
  If $\mathrm{Subtype}\,t_1\,t_2 \ipeqvred \tau \in \Tc$ then there
  exists $\tau' \in \Tc$ such that $\tau \subseteq \tau'$, $t_1
  \ipreduces \tau'$, and for all terms~$t_3$:
  \begin{itemize}
  \item[(a)] if $t_3 \vtr \tau'$ then $t_2 t_3 \vtr \Bool$,
  \item[(b)] $t_3 \vtr \tau$ iff $t_3 \vtr \tau'$ and $t_2 t_3
    \ipeqvred \true$.
  \end{itemize}
\end{lemma}

\begin{proof}
  Suppose $\mathrm{Subtype}\,t_1\,t_2 \ipeqvred \tau \in
  \Tc$. Since~$\tau$ is in normal form, we conclude by the
  Church-Rosser property that $\mathrm{Subtype}\,t_1\,t_2 \ipreduces
  \tau$. By inspecting the definition of~$R$ we see that this is only
  possible when $t_1 \ipreduces \tau' \in \Tc$. If $\tau' \equiv
  \emptyset$ then $\mathrm{Subtype}\,t_1\,t_2 \ipreduces \emptyset$
  and by the Church-Rosser property we obtain $\tau \equiv
  \emptyset$. If $\tau \equiv \emptyset$ then obviously $\tau
  \subseteq \tau'$ and both conditions~(a) and~(b) are satisfied,
  because $t_3 \not\vtr \emptyset$ for any~$t_3$, as $t_3 \vtr
  \emptyset$ would require the existence of some $c \in \emptyset$.

  So suppose $\tau' \not\equiv \emptyset$ and $\tau \not\equiv
  \emptyset$. Then $\mathrm{Subtype}\,\tau'\,t_2 \contr t$ where
  $\upos{t}{\epsilon} \equiv \Set\tau'$ and $\pos{t}{c} \equiv t_2 c$
  for $c \in \tau'$. By the Church-Rosser property we have $t
  \ipreduces \tau$, and by inspecting the definitions we easily see
  that this is only possible when $t_2 c \equiv \pos{t}{c} \ipreduces
  d_c \in \Bool$ for every $c \in \tau'$ and $\tau = \{ c \in \tau'
  \;|\; d_c \equiv \true \}$. But then obviously $\tau \subseteq
  \tau'$. To show~(a) suppose $t_3 \vtr \tau'$, i.e. $t_3 \ipreduces
  t_3' \succ c \in \tau'$. By Corollary~\ref{cor_succ_arg} we obtain
  $t_2 t_3 \ipreduces d_c \in \Bool$. We now prove~(b). Suppose $t_3
  \vtr \tau$, i.e. $t_3 \ipreduces t_3' \succ c \in \tau$. Since $\tau
  \subseteq \tau'$ we obviously have $t_3 \vtr \tau'$. By
  Corollary~\ref{cor_succ_arg} we obtain $t_2 t_3 \ipreduces
  d_c$. Since $c \in \tau$ we have $d_c \equiv \true$. For the other
  direction, assume $t_2 t_3 \ipreduces \true$ and $t_3 \vtr \tau'$,
  i.e. $t_3 \ipreduces t_3' \succ c \in \tau'$. By
  Corollary~\ref{cor_succ_arg} we have $t_2 t_3 \ipreduces d_c \in
  \Bool$. By the Church-Rosser property we conclude $d_c \equiv
  \true$. Hence $c \in \tau$ and $t_3 \vtr \tau$.
\end{proof}

\begin{lemma}\label{lem_or_eq_choice}
  \begin{itemize}
  \item[(a)] $\vee t_1 t_2 \ipeqvred \true$ iff $t_1 \ipeqvred \true$ or
    $t_2 \ipeqvred \true$.
  \item[(b)] $\vee t_1 t_2 \ipeqvred \false$ iff $t_1 \ipeqvred \false$ and
    $t_2 \ipeqvred \false$.
  \item[(c)] $\Eq{t_1}{t_2} \ipeqvred \true$ iff $t_1 \ipeqvred t_2$.
  \item[(d)] If $\forall t \lambda x . \false \ipeqvred \false$ with
    $x \notin FV(t_2)$, then $\mathrm{Is}\, (\mathrm{Choice}\, t)\, t
    \ipeqvred \true$.
  \end{itemize}
\end{lemma}

\begin{proof}
  Follows easily from definitions and the Church-Rosser property.
\end{proof}

\begin{lemma}\label{lem_inductive}
  Suppose $\iota \in \Tc_I$, $\iota = \mu(\langle \iota_{1,1}, \ldots,
  \iota_{1,n_1}\rangle, \ldots, \langle
  \iota_{m,1},\ldots,\iota_{m,n_m} \rangle)$, $c_i^\iota \in \Cc$ has
  arity~$n_i$, $\iota_{i,j}^* = \iota_{i,j}$ if $\iota_{i,j} \in
  \Tc_I$, $\iota_{i,j}^* = \iota$ if $\iota_{i,j} = \star$, and
  $\bar{\iota}_{i,j}^* \in \bar{\Tc}_I$ is the subset of~$\Kc$
  determined by~$\iota_{i,j}^*$ (see
  Definition~\ref{def_canonical}). Let $t$ be a term.

  If for all $i=1,\ldots,m$ we have:
  \begin{itemize}
  \item for all $t_1,\ldots,t_{n_i}$ such that
    \begin{itemize}
    \item $t_j \vtr \bar{\iota}_{i,j}^*$ for $j=1,\ldots,n_i$, and
    \item $t t_j \ipreduces \true$ for all $1 \le j \le n_i$ such that
      $\iota_{i,j} = \star$
    \end{itemize}
    we have $t (c_i^\iota t_1 \ldots t_{n_i}) \ipreduces \true$
  \end{itemize}
  then $\forall \bar{\iota} t \ipreduces \true$.
\end{lemma}

\begin{proof}
  Note that for every $\I_s$-term in~$\T_I$ (see
  Definition~\ref{def_canonical}) there exists a semantic term
  corresponding to it in the obvious way. Thus, we may assume without
  loss of generality that~$\T_I$ is a set of semantic terms. Let $s
  \in \bar{\iota}$. It follows by straightforward induction on the
  structure of the term~$r \in \T_I$ corresponding to~$s$ (see
  Definition~\ref{def_canonical}) that $t r \ipreduces \true$. It is
  easy to see from the definition of~$R$ that $r \ipcontr
  s$. Therefore $t s \ipreduces \true$.
\end{proof}

\begin{lemma}\label{lem_truth_value}
  $\forall \, (\mathrm{Subtype}\,\Bool\,(\lambda x . \false)) \,
  (\lambda x . \false) \ipcontr \true$.
\end{lemma}

\begin{proof}
  It is easy to see that $\mathrm{Subtype}\,\Bool\,(\lambda x \,.\,
  \false) \ipcontr \emptyset$, and by definition of~$R$ we have
  $\forall\, \emptyset \lambda x \,.\, \false \contr \true$.
\end{proof}

\begin{lemma}\label{lem_supset_t2}
  If $\forall \, (\mathrm{Subtype}\,\Bool\,(\lambda x . t)) \ipreduces
  \true$ and $x \notin FV(t)$ then $t \ipreduces \top$ or $t
  \ipreduces \bot$.
\end{lemma}

\begin{proof}
  Follows easily from definitions and the Church-Rosser property.
\end{proof}

\begin{definition}\rm
  For $\tau \in \Tc$, we write $t_1 \sim_\tau t_2$ if for every
  term~$t$ such that $t \vtr \Bool^\tau$ we have $t t_1 \ipeqvred t
  t_2$.
\end{definition}

\begin{lemma}\label{lem_ext_fun}
  If $\tau_1, \tau_2 \in \Tc$, $\tau = \tau_2^{\tau_1}$, $t_1 \vtr
  \tau$, $t_2 \vtr \tau$ and for all~$s$ such that $s \vtr \tau_1$ we
  have $t_1 s \sim_{\tau_2} t_2 s$, then $t_1 \sim_{\tau} t_2$.
\end{lemma}

\begin{proof}
  We have $t_1 \ipreduces t_1' \succ c_1 \in \tau$ and $t_2 \ipreduces
  t_2' \succ c_2 \in \tau$.  Let $d_1 \in \tau_1$, $d_2 \in \tau_2$
  and let $f \in \Bool^{\tau_2}$ be such that $\Fc(f)(d_2) \equiv
  \true$ and $\Fc(f)(d) \equiv \false$ for $d \ne d_2$. Since $t_1'
  \succ c_1$, there exists~$s_1$ such that $t_1' d_1 \ipreduces s_1
  \succ \Fc(c_1)(d_1)$. Analogously $t_2' d_1 \ipreduces s_2 \succ
  \Fc(c_2)(d_1)$. Because $t_1 d_1 \sim_{\tau_2} t_2 d_1$, we have $f
  (t_1 d_1) \ipeqvred f (t_2 d_1)$. We have $f (t_1 d_1) \ipreduces f
  (t_1' d_1) \ipreduces f s_1$. Since $s_1 \succ \Fc(c_1)(d_1) \in
  \tau_2$ and $f \in \Bool^{\tau_2}$, we obtain $f s_1 \ipcontr
  \Fc(f)(\Fc(c_1)(d_1))$ by condition~(6) in the definition
  of~$\ipcontr$. Analogously $f (t_2 d_1) \ipreduces
  \Fc(f)(\Fc(c_2)(d_1))$. Hence, by the Church-Rosser property
  $\Fc(f)(\Fc(c_1)(d_1)) \equiv \Fc(f)(\Fc(c_2)(d_1))$. In other
  words, $\Fc(c_1)(d_1) \equiv d_2$ iff $\Fc(c_2)(d_1) \equiv
  d_2$. Since~$d_2 \in \tau_2$ was arbitrary, $\Fc(c_1)(d_1) \equiv
  \Fc(c_2)(d_1)$. But $d_1 \in \tau_1$ was also arbitrary, so $c_1
  \equiv c_2$.

  Now, let~$t$ be such that $t \vtr \Bool^{\tau}$. Then $t \ipreduces
  t' \succ p \in \Bool^{\tau}$. We have $t t_1 \ipreduces t' t_1
  \ipreduces t' t_1' \gg p c_1 \ipcontr \Fc(p)(c_1)$. By
  part~($\rmnum{1}$) of Lemma~\ref{lem_gg} we have $t' t_1' \ipreduces
  r_1 \gg \Fc(p)(c_1)$ for some term~$r_1$. But $\Fc(p)(c_1) \in
  \Bool$, so in fact $r_1 \equiv \Fc(p)(c_1)$ Analogously $t t_2
  \ipreduces \Fc(p)(c_2) \equiv \Fc(p)(c_1)$. Therefore $t t_1
  \ipeqvred t t_2$. Hence $t_1 \sim_\tau t_2$.
\end{proof}

\begin{lemma}\label{lem_ext_bool}
  Define~$t_1 \supset t_2$ by $t_1 \supset t_2 \equiv \forall
  (\mathrm{Subtype}\, \Bool\, \lambda x . t_1) \lambda y . t_2$, where
  $x,y \notin FV(t_1,t_2)$. If $t_1 \supset t_2 \ipreduces \true$ and
  $t_2 \supset t_1 \ipreduces \true$, then $t_1 \ipeqvred t_2$.
\end{lemma}

\begin{proof}
  Suppose $t_1 \supset t_2 \ipreduces \true$ and $t_2 \supset t_1
  \ipreduces \true$. We thus have
  \[
  \forall (\mathrm{Subtype}\, \Bool\, \lambda x . t_1) \lambda y . t_2
  \ipreduces \true
  \]
  which is only possible when $(\mathrm{Subtype}\, \Bool\, \lambda x
  . t_1) \ipreduces \tau \in \Tc$. But this in turn is only possible
  when $t_1 \ipreduces t_1' \in \{\true, \false\}$. By an analogous
  argument $t_2 \ipreduces t_2' \in \{\true, \false\}$. Now it is easy
  to see that if $t_1' \ne t_2'$ then $t_1 \supset t_2 \ipreduces
  \false$ or $t_2 \supset t_1 \ipreduces \false$.
\end{proof}

From now on, to avoid confusion, we use $r$, $r_1$, $r_2$, etc. to
denote terms of~$\I_s$, and $t$, $t_1$, $t_2$, etc. to denote semantic
terms.

\newcommand{\transla}[1]{\ensuremath{\llceil #1 \rrceil}}

\begin{definition}\label{def_model} \rm
  We define an $\I_s$-structure $\Mc = \langle A, \cdot,
  \valuation{}{}{} \rangle$ as follows. As~$A$ we take the set of
  equivalence classes of~$\ipeqvred$ on semantic terms. We denote the
  equivalence class of a semantic term~$t$ by $[t]$. We define $[t_1]
  \cdot [t_1] = [t_1 t_2]$. This is well-defined because $t_1
  \ipeqvred t_1'$ and $t_2 \ipeqvred t_2'$ imply $t_1 t_2 \ipeqvred
  t_1' t_2'$. To save on notation, we sometimes confuse~$[t]$
  with~$t$, where it does not lead to ambiguities.

  Let $v : V \to A$ be an $\Mc$-valuation. We define a function~$u$
  from~$V$ to semantic terms by $u(x) = t$ where~$t$ is an arbitrary
  but fixed semantic term such that $[t] = v(x)$. By $u[x/t]$ we
  denote a function~$w$ from~$V$ to semantic terms such that $w(y) =
  u(y)$ for $y \ne x$, and $w(x) = t$. By induction on the structure
  of an $\I_s$-term~$r$, we define a translation~$\transla{r}^u$ from terms
  of~$\I_s$ to semantic terms, parametrized by a function~$u$
  from~$V$ to semantic terms:
  \begin{itemize}
  \item $\transla{\mathrm{Is}} = \mathrm{Is}$,
    $\transla{\mathrm{Subtype}} = \mathrm{Subtype}$,
    $\transla{\mathrm{Fun}} = \mathrm{Fun}$, $\transla{\forall} =
    \forall$, $\transla{\vee} = \vee$, $\transla{\bot} = \false$,
    $\transla{\epsilon} = \mathrm{Choice}$, $\transla{\mathrm{Eq}} =
    \mathrm{Eq}$, $\transla{\mathrm{Cond}} = \mathrm{Cond}$,
    $\transla{\mathrm{Type}} = \mathrm{Type}$,
    $\transla{\mathrm{Prop}} = \mathrm{Bool}$, $\transla{\iota} =
    \bar{\iota}$ for $\iota \in \Tc_I$, $\transla{c_i^\iota} =
    c_i^\iota$ for $c_i^\iota \in \Cc$, $\transla{d_{i,j}^\iota} =
    d_{i,j}^\iota$ for $d_{i,j}^\iota \in \Dc$, $\transla{o_i^\iota} =
    o_i^\iota$ for $o_i^\iota \in \Oc$,
  \item $\transla{x}^u = u(x)$ for $x \in V$,
  \item $\transla{r_1 r_2}^u = \transla{r_1}^u \cdot \transla{r_2}^u$,
  \item $\transla{\lambda x \,.\, r}^u = \lambda y \,.\, \transla{r}^{u[x/y]}$
    where $x \in V$ and $y \in V^+$ is a fresh variable.
  \end{itemize}

  Now the interpretation~$\valuation{}{}{}$ is defined by
  \[
  \valuation{r}{v}{} = [\transla{r}^u]
  \]
  where~$u$ is the function from~$V$ to semantic terms, corresponding
  to~$v$, as defined above.
\end{definition}

\begin{theorem}\label{thm_consistent}
  The system~$e\I_s$ is consistent, i.e. $\not\proves_{e\I_s} \bot$.
\end{theorem}

\begin{proof}
  We show that~$\Mc$ is an $e\I_s$-model. We need to check the
  conditions in Definition~\ref{def_I_s_model}. Conditions~(var)
  and~(app) follow directly from the definition of~$\Mc$. For
  condition~($\beta$) note that $\valuation{\lambda x \,.\, r}{v}{}
  \cdot [t] = [\transla{\lambda x \,.\, r}^u] \cdot [t] = [(\lambda y
    \,.\, \transla{r}^{u[x/y]}) t] = [\transla{r}^{u[x/t]}] =
  \valuation{r}{v[x/[t]]}{}$ where~$u$ is like in
  Definition~\ref{def_model}. Condition~(fv) is obvious from the
  definition of~$\valuation{}{}{}$. For condition~($\xi$) suppose
  $\valuation{\lambda x \,.\, r_1}{v}{} \cdot [t] = \valuation{\lambda
    x \,.\, r_2}{v}{} \cdot [t]$ for every semantic term~$t$. Then in
  particular this holds for any variable $y \in V^+$. We have
  $\valuation{\lambda x \,.\, r_1}{v}{} \cdot [t] = [(\lambda y \,.\,
    \transla{r_1}^{u[x/y]}) y] = [\transla{r_1}^{u[x/y]}]$ and
  $\valuation{\lambda x \,.\, r_2}{v}{} \cdot [t] = [(\lambda y \,.\,
    \transla{r_2}^{u[x/y]}) y] = [\transla{r_2}^{u[x/y]}]$, where~$u$
  is like in Definition~\ref{def_model}. Hence
  $[\transla{r_1}^{u[x/y]}] = [\transla{r_2}^{u[x/y]}]$. But
  $\valuation{\lambda x \,.\, r_1}{v}{} = [\lambda y \,.\,
    \transla{r_1}^{u[x/y]}]$ and $\valuation{\lambda x \,.\, r_2}{v}{}
  = [\lambda y \,.\, \transla{r_2}^{u[x/y]}]$.

  Condition~(pr) follows from
  Lemma~\ref{lem_truth_value}. Condition~(pt) follows directly from
  the definition of~$R$.  Conditions~($\forall_\top$),
  ($\forall_\bot$), ($\forall_e$) and~($\forall_e'$) follow from
  Lemma~\ref{lem_forall} and Lemma~\ref{lem_is}. Conditions~($\vee_1$)
  and~($\vee_2$) follow from
  Lemma~\ref{lem_or_eq_choice}. Condition~($\supset_{t2}$) follows
  from Lemma~\ref{lem_supset_t2}. Condition~($\bot$) follows from
  Lemma~\ref{lem_truth_value} and the Church-Rosser property.
  Conditions~($\to_i$) and~($\to_e$) follow from Lemma~\ref{lem_fun}
  and Lemma~\ref{lem_is}. Condition~($\to_t$) follows from the
  definition of~$R$ and the Church-Rosser
  propery. Conditions~(s1)-(s3) follow from Lemma~\ref{lem_subtype}
  and Lemma~\ref{lem_is}. Condition~(s4) follows easily from
  definitions and Lemma~\ref{lem_is}. Conditions~(o1), (o2), (d1) and
  (i2) follow directly from definitions. Condition~(i1) follows from
  Lemma~\ref{lem_inductive}. Conditions~($\epsilon$) and~(eq) follow
  from Lemma~\ref{lem_or_eq_choice}. Conditions~(c1) and~(c2) are
  obvious from the definition of~$R$ and the Church-Rosser
  property. Condition~(ef) follows from
  Lemma~\ref{lem_ext_fun}. Condition~(eb) follows from
  Lemma~\ref{lem_ext_bool}.

  Now suppose $\proves_{e\I_s} \bot$. Then by Theorem~\ref{thm_sound}
  we have $\Mc \models \bot$, so $\valuation{\bot}{}{} = [\true]$. But
  we have $\valuation{\bot}{}{} = [\false]$, which implies $\true
  \ipeqvred \false$. This is impossible by the Church-Rosser property
  of~$\ipcontr$.
\end{proof}

\newpage
\renewcommand{\thesection}{Appendix~\Alph{section}}
\tocless\section{Complete Translation of First-Order Logic}\label{sec_embedding}
\renewcommand{\thesection}{\Alph{section}}
\addcontentsline{toc}{section}{\thesection\hspace{0.5em} Complete
  Translation of First-Order Logic}

In this appendix we show that the translation from
Sect.~\ref{sec_consistent} restricted to first-order logic is
complete, i.e. we prove Theorem~\ref{thm_embedding_01}. The method of
the proof is essentially the same as in~\cite{Czajka2011}, and it is a
relatively simple application of the construction from the previous
appendix.

First, let us precisely state the definition of the
system~$\mathrm{FO}$ of classical first-order logic.

\newcommand{\FO}{{\ensuremath{{\mathrm{FO}}}}}

\begin{definition} \rm
  \begin{itemize}
  \item The \emph{types} of~$\FO$ are given by
    \[
    \Tc \;\; ::= \;\; o \;|\; \delta \;|\; \delta \to \Tc
    \]
  \item The set of terms of $\FO$ of type~$\tau$, denoted~$T_\tau$, is
    defined as follows:
    \begin{itemize}
    \item $V_\delta \subseteq T_\delta$,
    \item $\Sigma_\tau \subseteq T_\tau$,
    \item if $t_1 \in T_{\delta\to\tau}$ and $t_2 \in T_\delta$ then
      $t_1 t_2 \in T_\tau$,
    \item $\bot \in T_o$,
    \item if $\varphi, \psi \in T_o$ then $\varphi \supset \psi \in
      T_o$,
    \item if $x \in V_\delta$ and $\varphi \in T_o$ then $\forall x :
      \delta \,.\, \varphi \in T_o$,
    \end{itemize}
    where~$V_\delta$ is a countable set of variables and for each
    $\tau \in \Tc \setminus \{o\}$ the set~$\Sigma_\tau$ is a
    countable set of constants. We assume that the sets $\Sigma_\tau$
    are all pairwise disjoint, and disjoint with~$V_\delta$. Terms of
    type~$o$ are \emph{formulas}. As usual, we omit spurious brackets
    and assume that application associates to the left. We identify
    $\alpha$-equivalent terms, i.e. terms differing only in the names
    of bound variables are considered identical.
  \item The system~$\FO$ is given by the following rules and axioms,
    where~$\Delta$ is a finite set of formulas, $\varphi, \psi$ are
    formulas. The notation $\Delta, \varphi$ is a shorthand for
    $\Delta \cup \{\varphi\}$.

    \smallskip

    {\bf Axioms}
    \begin{itemize}
    \item $\Delta, \varphi \proves \varphi$
    \item $\Delta \proves ((\varphi \supset \bot) \supset \bot)
      \supset \varphi$
    \end{itemize}

    {\bf Rules}

    \begin{center}
      \begin{tabular}{lr}
        \(
          {\supset_i^{\FO}:}\; \inferrule{\Delta, \varphi \proves
            \psi}{\Delta \proves \varphi \supset \psi}
        \)
        &
        \(
          {\supset_e^{\FO}:}\; \inferrule{\Delta \proves \varphi \supset
            \psi \;\;\; \Delta \proves \varphi}{\Delta \proves \psi}
        \)
        \\
        & \\
        \(
          {\forall_i^{\FO}:}\; \inferrule{\Delta \proves
            \varphi}{\Delta \proves \forall x : \delta \,.\, \varphi} \; x \notin FV(\Delta)
        \)
        \;\;\;
        &
        \;\;\;
        \(
          {\forall_e^{\FO}:}\; \inferrule{\Delta \proves \forall x : \delta \,.\, \varphi}{\Delta
            \proves \varphi[x/t]}\; t \in T_\delta
        \)
      \end{tabular}
    \end{center}
  \end{itemize}
  Note that the system~$\FO$ is a restriction of~\mbox{CPRED$\omega$}.

  A first-order structure~$\Ac$ is a pair $\langle A, \{ f_c \;|\; c
  \in \Sigma_\tau, \tau \in \Tc \} \rangle$ where~$A$ is the non-empty
  universe and~$f_c$ are interpretations of constants (functions or
  relations on~$A$ of appropriate arity, or elements
  of~$A$). By~$\valuation{c}{}{\Ac}$ we denote the interpretation
  of~$c$ in~$\Ac$. A first-order valuation~$u$ is a function from
  variables to the universe of a first-order structure. We define the
  relations of satisfaction $\Ac \models_\FO \varphi$ and semantic
  consequence $\Delta \models_\FO \varphi$ in an obvious way.
\end{definition}

The following is a well-known result from elementary logic.

\begin{theorem}
  $\Delta \proves_\FO \varphi$ iff $\Delta \models_\FO \varphi$
\end{theorem}

The mappings~$\transl{-}$ and~$\Gamma(-)$ are defined exactly as in
Sect.~\ref{sec_consistent}, restricting to terms of~$\FO$. Recall that
we assume $V_\delta \subseteq V$ and $\Sigma_\tau \subseteq \Sigma_s$
for~$\tau \in \Tc$.

\begin{theorem}
  If $\Delta \proves_\FO \varphi$ then $\Gamma(\Delta, \varphi),
  \transl{\Delta} \proves_{\I_s} \transl{\varphi}$.
\end{theorem}

\begin{proof}
  This is a special case of Theorem~\ref{thm_sound_hol}.
\end{proof}

In this section we use $s$, $s_1$, $s_2$, etc. for terms of~$\I_s$. We
assume that there is a fresh constant~$\delta \in \Sigma_s$. This
constant will represent the first-order universe. We use the notation
$\tau^n\to\tau$ for $\tau\to\ldots\to\tau$ where~$\tau$ occurs~$n+1$
times, where $n > 0$.

\begin{theorem}
  If $\Gamma(\Delta, \varphi), \transl{\Delta} \proves_{e\I_s}
  \transl{\varphi}$ then $\Delta \proves_\FO \varphi$.
\end{theorem}

\begin{proof}
  Suppose $\Gamma(\Delta, \varphi), \transl{\Delta} \proves_{e\I_s}
  \transl{\varphi}$ but $\Delta \not\proves_\FO \varphi$. Then $\Delta
  \not\models_\FO \varphi$, so there exist a first-order
  structure~$\Ac$ and a first-order valuation~$u$ such that $\Ac, u
  \models_\FO \Delta$, but $\Ac, u \not\models_\FO \varphi$.

  We use the construction of Definition~\ref{def_model} to
  transform~$\Ac$ into an $e\I_s$-mo\-del~$\Mc$, by taking the set of
  constants~$\delta$ of~$\Mc$ to consist of the elements of the
  universe of~$\Ac$. We also need to extend the definition of the
  translation~$\transla{-}$ from Definition~\ref{def_model} to
  interpret the new constants that we added to the language
  of~$e\I_s$. We set $\transla{\delta} = \delta$. If $c \in
  \Sigma_{\delta}$ then we set $\transla{c} =
  \valuation{c}{}{\Ac}$. If~$f \in \Sigma_{\delta^n\to\delta}$ then
  let $c_f \in \delta^{(\delta,\ldots,\delta)}$ be such that
  $\Fc(c_f)(a_1)(a_2)\ldots(a_n) = \valuation{f}{}{\Ac}(a_1, \ldots,
  a_n)$ for any $a_1,\ldots,a_n \in \delta$. We set $\transla{f} =
  c_f$. Similarly, if~$r \in \Sigma_{\delta^n\to o}$ then let $c_r \in
  \Bool^{(\delta,\ldots,\delta)}$ be such that for all $a_1,\ldots,a_n
  \in \delta$ we have $\Fc(c_r)(a_1)(a_2)\ldots(a_n) = \true$ iff
  $\valuation{r}{}{\Ac}(a_1, \ldots, a_n)$ holds.

  Note that in~$\Mc$ we have $\delta = \{ a \in \Mc \;|\; \mathrm{Is}
  \cdot a \cdot \delta = \top \}$. This follows directly from
  Lemma~\ref{lem_is} and the fact that for $a \in \delta$ we have $a
  \notin \tau_2^{\tau}$, and hence if $t \vtr a$ then $t \ipeqvred
  a$. Note also that any first-order $\Ac$-valuation~$v$ is also an
  $\Mc$-valuation, if for a variable~$x$ we interpret~$v(x)$ as an
  element of the set~$\delta$ in~$\Mc$.

  Now it is easy to show by induction on the structure of a term~$t
  \in T_\tau$ that $\valuation{\transl{t}}{v}{\Mc} =
  \valuation{t}{v}{\Ac} \in \delta$ for any first order-valuation~$v$,
  where $\tau \in \Tc \setminus \{o\}$. Using this we verify by
  straightforward induction on the structure of a first-order
  formula~$\varphi$ that for any first-order valuation~$v$ we have
  $\valuation{\transl{\varphi}}{v}{\Mc} = \true$ iff $\Ac, v \models
  \varphi$, and $\valuation{\transl{\varphi}}{v}{\Mc} = \false$ iff
  $\Ac, v \not\models \varphi$. For instance, suppose $\varphi \equiv
  \forall x : \delta \,.\, \psi$. Then $\transl{\varphi} = \forall x :
  \delta \,.\, \transl{\psi}$. Since $\valuation{\delta :
    \Type}{v}{\Mc} = \true$, and $a \in \Ac$ iff $\mathrm{Is} \cdot a
  \cdot \delta = \true$, we have $\valuation{\transl{\varphi}}{v}{\Mc}
  = \true$ iff $\valuation{\transl{\psi}}{v[x/a]}{\Mc} = \true$ for
  every $a \in \A$ iff $\Ac, v[x/a] \models \psi$ for every $a \in \A$
  iff $\Ac, v \models \forall x : \delta \,.\, \psi \equiv \varphi$,
  where we use the inductive hypothesis in the second equivalence.

  It is easy to check that $\Mc, u \models_{e\I_s} \Gamma(\Delta,
  \varphi)$. Hence $\Mc, u \models_{e\I_s} \Gamma(\Delta, \varphi),
  \transl{\Delta}$, because $\Ac, u \models_\FO \Delta$. Since $\Ac, u
  \not\models_\FO \varphi$, we also have $\Mc, u \not\models_{e\I_s}
  \transl{\varphi}$. But by Theorem~\ref{thm_sound} and
  $\Gamma(\Delta, \varphi), \transl{\Delta} \proves_{e\I_s}
  \transl{\varphi}$ we have $\Mc, u \models_{e\I_s}
  \transl{\varphi}$. Contradiction.
\end{proof}

\end{document}